\documentclass[reqno]{article}

\usepackage{amsthm}
\usepackage{amsmath}
\usepackage{amssymb}
\usepackage{graphicx}
\usepackage{color}
\usepackage{bbm}
\usepackage[top=2.5cm, bottom=2.5cm, left=2.5cm, right=2.5cm]{geometry}
\usepackage{enumerate}
\usepackage{faktor}
\usepackage{nicefrac}
\usepackage[nottoc]{tocbibind}
\usepackage{slashed}
\usepackage{authblk}
\usepackage[normalem]{ulem}
\usepackage{comment}

\usepackage[scr=euler]{mathalfa}

\def\mathclap#1{\text{\hbox to 0pt{\hss$\mathsurround=0pt#1$\hss}}}

\newcommand{\loc}{\mathrm{loc}}

\newcommand{\N}{\mathbb{N}}
\newcommand{\R}{\mathbb{R}}
\newcommand{\id}{\mathrm{id}}

\newcommand{\ux}{\underline{x}}
\newcommand{\uy}{\underline{y}}
\newcommand{\uz}{\underline{z}}
\newcommand{\Lg}{L_{\mathrm{gap}, e_\alpha}}
\newcommand{\gap}{\mathrm{gap}}

\newcommand{\rd}{\partial}

\begin{document}

\numberwithin{equation}{section}
\newtheorem{theorem}[equation]{Theorem}
\newtheorem{remark}[equation]{Remark}
\newtheorem{assumption}[equation]{Assumption}
\newtheorem{claim}[equation]{Claim}
\newtheorem{lemma}[equation]{Lemma}
\newtheorem{definition}[equation]{Definition}
\newtheorem{corollary}[equation]{Corollary}
\newtheorem{proposition}[equation]{Proposition}
\newtheorem*{theorem*}{Theorem}
\newtheorem{conjecture}[equation]{Conjecture}
\newtheorem{example}[equation]{Example}

\setcounter{tocdepth}{2}

\title{Uniqueness and non-uniqueness results for spacetime extensions}
\author{Jan Sbierski\thanks{School of Mathematics, 
University of Edinburgh,
James Clerk Maxwell Building,
Peter Guthrie Tait Road, 
Edinburgh, 
EH9 3FD,
United Kingdom; email: Jan.Sbierski@ed.ac.uk}}
\date{\today}

\maketitle

\begin{abstract}
Given a function $f : A \to \R^n$ of a certain regularity defined on some open subset $A \subseteq \R^m$, it is  a classical problem of analysis to investigate whether the function can be extended to all of $\R^m$ in a certain regularity class. If an extension exists and is continuous, then certainly it is uniquely determined on the closure of $A$. A similar problem arises in general relativity for Lorentzian manifolds instead of functions on $\R^m$. It is well-known, however, that even if the extension of a Lorentzian manifold $(M,g)$ is analytic, various choices are in general possible at the boundary. This paper establishes a uniqueness condition for extensions of globally hyperbolic Lorentzian manifolds $(M,g)$ with a focus on low regularities: any two extensions which are \emph{anchored} by an inextendible causal curve $\gamma : [-1,0) \to M$ in the sense that $\gamma$ has limit points in both extensions, must agree locally around those limit points on the boundary \emph{as long as the extensions are at least locally Lipschitz continuous}. We also show that this is sharp: anchored extensions which are only H\"older continuous do in general not enjoy this local uniqueness result.
\end{abstract}

\tableofcontents

\section{Introduction}

A classical problem in analysis is the extension problem: given a subset $A \subseteq \R^m$ and a function $f : A \to \R^n$, can one find a function $F : \R^m \to \R^n$ with $F|_A = f$? The answer of course depends on the regularity of $f$, the desired regularity of the extension $F$, and the structure of the set $A \subseteq \R^m$.\footnote{The literature on this subject is of course vast. As an illustration of the type of results available let us refer the reader to Chapter VI in \cite{Stein70}.} However, it is immediate that if the extension $F$ is continuous, then it is uniquely determined on the closure $\overline{A}$ of $A$. In other words, the extension of $f$ to the boundary $\rd A$ of $A$ is the same among all continuous extensions.

In Einstein's theory of general relativity the following fundamental question arises: given a Lorentzian manifold $(M,g)$ which solves the, say, vacuum Einstein equations $\mathrm{Ric}(g) = 0$, can it be extended as a (weak) solution of the vacuum Einstein equations? As a first step in approaching this question one can drop the requirement of $(M,g)$ satisfying  the (weak) vacuum  Einstein equations and just ask whether one can extend $(M,g)$ in a certain regularity class. The definition here is the following: Let $\Gamma$ be a regularity class, for example $\Gamma = C^\infty, C^{0,1}_{\loc}$, etc. The roughest regularity we consider is $\Gamma = C^0$. A \textbf{$\Gamma$-extension} of a Lorentzian manifold $(M,g)$ consists of an isometric embedding $\iota : M \hookrightarrow \tilde{M}$ of $(M,g)$ into a Lorentzian manifold $(\tilde{M}, \tilde{g})$ of the same dimension as $M$ where $\tilde{g}$ is $\Gamma$-regular and such that $\rd \iota(M) \subseteq \tilde{M}$ is non-empty.\footnote{All manifolds considered in this paper are assumed to be smooth. See also Definition \ref{DefExtension} and Footnote \ref{FootnoteSmooth}.} One crucial difference between the extension problem for Lorentzian manifolds and that for functions on $\R^m$ is of course that in the former case no fixed background space (analogue of $\R^m$) is given and the manifold on which the metric is being extended has to be constructed at the same time as the metric itself. This makes the extension problem for Lorentzian manifolds much more subtle -- hardly any results are available. Another well-known manifestation of the absence of a fixed background is that extensions, even if they are analytic, are, in general, not uniquely determined on the boundary $\rd\iota(M) \subseteq \tilde{M}$ of $M$. This is in stark contrast to the extension problem for functions on $\R^m$ as discussed earlier. Before we present the results obtained in this paper, let us make precise in the next section the notion of `extensions of Lorentzian manifolds being the same on the boundary'.

\subsection{The problem of the uniqueness of extensions of spacetimes}

Let $(M,g)$ be a Lorentzian manifold and let $\tilde{\iota} : M \hookrightarrow \tilde{M}$ and $\hat{\iota} : M \hookrightarrow \hat{M}$ be two extensions of $(M,g)$. We define the identification map $\id$ by 
\begin{equation*}
\id := \hat{\iota} \circ \tilde{\iota}^{-1}|_{\tilde{\iota}(M)} : \tilde{\iota}(M) \to \hat{\iota}(M)
\end{equation*}
which identifies the image of $M$ in $ \tilde{M}$ with that of $M$ in $\hat{M}$. Clearly, $\id$ is an isometry. One might call the two extensions equivalent up to the boundary if $\id$ extends (sufficiently smoothly) to a diffeomorphism $\overline{\id} : \overline{\tilde{\iota}(M)} \to \overline{\hat{\iota}(M)}$.

This is a global definition. It does not only take into account how the geometry is extended through the boundary, but also how the boundary is globally attached to the manifold. This point is best illustrated by Figure \ref{FigGlobExt}.
\begin{figure}
\centering
\begin{minipage}{0.8\textwidth}
  \centering
  \def\svgwidth{8cm}
   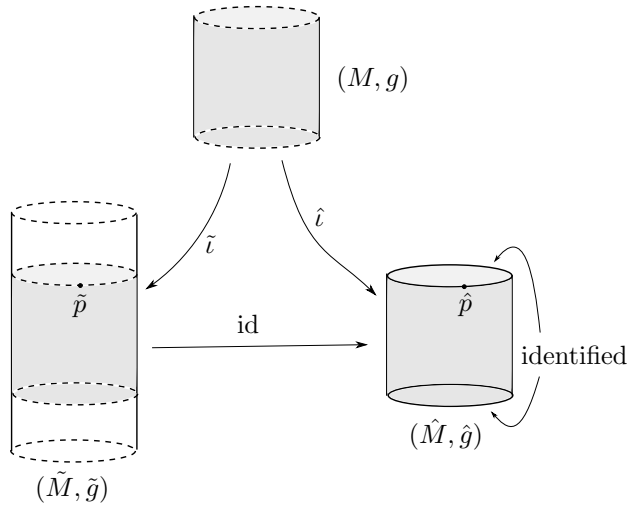 
      \caption{Here $M = (-1,1) \times \mathbb{S}^1$ with canonical $(t,\varphi)$-coordinates and metric $g = - dt^2 + d \varphi^2$. The first extension is $\tilde{M} = \R \times \mathbb{S}^1$ and the second is $\hat{M} = \big([-1,1] \times \mathbb{S}^1\big)/_\sim$ with $(-1,\varphi) \sim (1, \varphi)$ for all $\varphi \in \mathbb{S}^1$. The metrics $\tilde{g}$ and $\hat{g}$ take the same coordinate form as that of $g$ and the isometric embeddings $\tilde{\iota}$ and $\hat{\iota}$ are the obvious ones.} \label{FigGlobExt}
\end{minipage}
\end{figure}
Note that $\overline{\tilde{\iota}(M)}$ is a manifold with boundary, while $\overline{\hat{\iota}(M)}$ is a manifold without a boundary. Clearly, these two extensions are not equivalent according to the above definition. However, there is good reason to consider them to be locally equivalent: there is a neighbourhood $\tilde{V} \subseteq \tilde{M}$ of $\tilde{p}$ and a neighbourhood $\hat{V} \subseteq \hat{M}$ of $\hat{p}$ which are isometric. 

This paper focuses on the local aspects of spacetime extensions. In order to disentangle the local from the global issues we proceed as follows: as a first step we only consider extensions $\iota : M \hookrightarrow \tilde{M}$ of $d+1$-dimensional Lorentzian manifolds $(M,g)$ where $(M,g)$ is \emph{time-oriented and globally hyperbolic}. The extension itself is not assumed to be time-oriented or globally hyperbolic, but at least of regularity $C^0$. The \textbf{future boundary} \label{DefTextFB} $\rd^+\iota(M) \subseteq \tilde{M}$ of $M$ in $\tilde{M}$ consists of all those boundary points $\tilde{p}$ in $\rd \iota(M)$ for which there exists a future directed and future inextendible timelike curve $\gamma : [-1,0) \to M$ such that $\tilde{\gamma} := \iota \circ \gamma$ extends to $\tilde{p} = \tilde{\gamma}(0)$ as a \emph{smooth} timelike curve. One can now show that the future boundary of globally hyperbolic Lorentzian manifolds is locally a Lipschitz hypersurface in the following sense: around every $\tilde{p} \in \rd^+\iota(M)$, which is defined by some future directed and future inextendible timelike curve $\gamma$ in $M$, there exists a chart $\tilde{\varphi} : \tilde{U} \to (-\varepsilon_0, \varepsilon_0) \times (-\varepsilon_1, \varepsilon_1)^d$ with coordinates $(x_0, x_1, \ldots, x_d) = (x_0, \ux)$ such that 1) the metric $\tilde{g}_{\mu \nu}$ is $C^0$-close to the Minkowski metric in these coordinates, 2) $\tilde{p}$ lies at the coordinate origin, 3) there exists a Lipschitz continuous function $f : (-\varepsilon_1, \varepsilon_1)^d \to (-\varepsilon_0, \varepsilon_0)$ such that all points below the graph of $f$, i.e., $\tilde{U}_< := \{(x_0, \ux) \in \tilde{U} \; | \; x_0 < f(\ux)\}$, lie in $\iota(M)$ while all points on the graph of $f$ lie in $\rd^+\iota(M)$. Points above the graph of $f$ may or may not lie in $\iota(M)$ -- both cases are possible, cf.\ Figure \ref{FigGlobExt}.\footnote{The example in Figure 1 in \cite{Sbie18} shows that even another branch of $\rd^+\iota(M)$ can meet $\mathrm{graph} f$ from above, showing that  in general $\rd^+ \iota(M)$ is not a Lipschitz hypersurface.} And moreover, the timelike curve $\gamma$, which asymptotes to and defines $\tilde{p}$, is eventually contained in $\tilde{U}_<$.  We call such a chart $\tilde{\varphi} : \tilde{U} \to (-\varepsilon_0, \varepsilon_0) \times (-\varepsilon_1, \varepsilon_1)^d$ a \textbf{future boundary chart}.

Given now two extensions $\tilde{\iota} : M \hookrightarrow \tilde{M}$ and $\hat{\iota} : M \hookrightarrow \hat{M}$ of a time-oriented and globally hyperbolic Lorentzian manifold $(M,g)$ and future boundary points $\tilde{p} \in \rd \tilde{\iota}(M)$ and $\hat{p} \in \rd \hat{\iota}(M)$ with future boundary charts $\tilde{U}$ and $\hat{U}$, then the idea is to consider two extensions as locally equivalent up to the boundary around $\tilde{p}$ and $\hat{p}$ if there exist open neighbourhoods $\tilde{W} \subseteq \tilde{U}$ of $\tilde{p}$ and $\hat{W} \subseteq \hat{U}$ of $\hat{p}$ such that the identification map \underline{restricted} to $\tilde{W}_< := \tilde{W} \cap \tilde{U}_<$ maps $\tilde{W}_<$ bijectively onto $\hat{W}_< := \hat{W} \cap \hat{U}_<$ and extends as a sufficiently smooth diffeomorphism to $\id|_{\tilde{W}_\leq} : \tilde{W}_\leq \to \hat{W}_\leq$, where $\tilde{W}_\leq := \tilde{W} \cap ( \tilde{U}_< \cup \mathrm{graph} f)$ -- and similarly for $\hat{W}_\leq$.\footnote{Of course, without imposing further equations/conditions on the extensions one can only hope to obtain uniqueness up to the boundary, but not beyond. } 

It is important here to recall the classical example of the Taub-NUT spacetime -- or the simpler Misner spacetime \cite{Mis67} -- which shows that there are time-oriented and globally hyperbolic  Lorentzian manifolds $(M,g)$ which have analytic extensions $\tilde{\iota} : M \hookrightarrow \tilde{M}$ and $\hat{\iota} : M \hookrightarrow \hat{M}$, which cannot be enlarged, but at the same time are not locally equivalent around any pair of future boundary points. The underlying reason for this behaviour is clearly laid out in \cite{Chrus10}: using our terminology, let $\tilde{p}$ be any future boundary point in $\rd^+ \tilde{\iota}(M)$ which is the smooth endpoint of a future directed and future inextendible timelike curve $\gamma$ in $M$. Then the curve $\hat{\iota} \circ \gamma$ does not have an end point in $\hat{M}$, i.e., it does not define a future boundary point in $\rd^+ \hat{\iota}(M)$. But at the same time an end point cannot be added to the extension $\hat{M}$ without making the manifold non-Hausdorff. The same applies to future boundary points $\hat{p}$ in $\rd^+ \hat{\iota}(M)$. The two extensions `extend through two different future boundaries' in the sense that they terminate different (but maximal) sets of timelike curves.  We refer the reader to \cite{Chrus10} for more details.

This example highlights the fact that one can only hope to establish local equivalence of two extensions $\tilde{\iota} : M \hookrightarrow \tilde{M}$ and $\hat{\iota} : M \hookrightarrow \hat{M}$ around future boundary points $\tilde{p} \in \rd^+ \tilde{\iota}(M)$ and $\hat{p} \in \rd^+ \hat{\iota}(M)$ if there is a future directed  and future inextendible timelike curve $\gamma : [-1,0) \to M$  which terminates at $\tilde{p}$ in $\tilde{M}$ and at $\hat{p}$ in $\hat{M}$. We also say that in this case the two extensions are \textbf{anchored} by the curve $\gamma$. Our main result in this paper shows that this condition is also sufficient as long as the two extensions are locally Lipschitz continuous.

\begin{figure}[h]
\centering
\begin{minipage}{0.9\textwidth}
  \centering
  \def\svgwidth{8cm}
   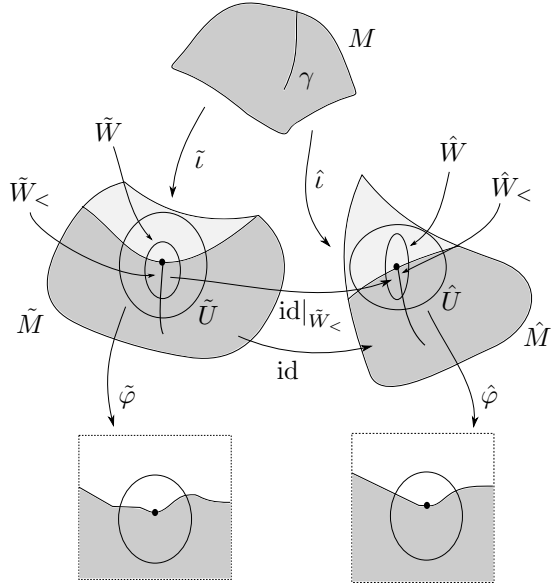 
      \caption{Recall from Figure \ref{FigGlobExt} that $\tilde{\iota}(M)$ may also lie above the graph of $f$ in the future boundary chart $\tilde{U}$ -- and similarly for the hatted extension. To keep the diagram simple this possibility is not depicted here. } \label{FigIdExt}
\end{minipage}
\end{figure}

\subsection{Main results and comparison to a result by Chru\'sciel}

We establish the following result, see Theorem \ref{ThmAnchUni} in Section \ref{SecC1}.
\begin{theorem}\label{ThmIntro}
Let $(M,g)$ be a time-oriented and globally hyperbolic Lorentzian manifold with $g \in C^1$ and let $\tilde{\iota} : M \hookrightarrow \tilde{M}$ and $\hat{\iota} : M \hookrightarrow \hat{M}$ be two $C^{0,1}_{\loc}$-extensions. Let $\gamma : [-1,0) \to M$ be a future directed and future inextendible causal $C^1$-curve in $M$ such that $\lim_{s \to 0} ( \tilde{\iota} \circ \gamma)(s) =: \tilde{p}$ and $\lim_{s \to 0}(\hat{\iota} \circ \gamma)(s) =: \hat{p}$ exist in $\tilde{M}$, $\hat{M}$, respectively. 

Then there exist future boundary charts $\tilde{\varphi} : \tilde{U} \to (-\tilde{\varepsilon}_0, \tilde{\varepsilon}_0) \times (-\tilde{\varepsilon}_1, \tilde{\varepsilon}_1)^d$ around $\tilde{p}$ and $\hat{\varphi} : \hat{U} \to (-\hat{\varepsilon}_0, \hat{\varepsilon}_0) \times (-\hat{\varepsilon}_1, \hat{\varepsilon}_1)^d$ around $\hat{p}$, with $\tilde{\iota} \circ \gamma$ being ultimately contained in $\tilde{U}_<$ and $\hat{\iota} \circ \gamma$ being ultimately contained in $\hat{U}_<$, and neighbourhoods $\tilde{W} \subseteq \tilde{U}$ of $\tilde{p}$ and $\hat{W} \subseteq \hat{U}$ of $\hat{p}$ such that $$\id|_{\tilde{W}_<} : \tilde{W}_< \to \hat{W}_<$$ is a diffeomorphism and extends as a $C^{1,1}_{\loc}$-regular (isometric) diffeomorphism to $$\id|_{\tilde{W}_{\leq}} : \tilde{W}_\leq \to \hat{W}_\leq \;.$$
\end{theorem}
Part of the statement of the theorem is that it suffices for the anchoring curve $\gamma$ to be causal and to have a mere limit point in the two extensions -- i.e., it is not needed that it extends as  a smooth curve to both extensions. This is a subtle point, see also the discussion in Section \ref{SecFR}. 

Let us also remark that by `$\id$ extending as a $C^{1,1}_{\loc}$-regular diffeomorphism to $\tilde{W}_\leq$'  we mean, as is standard, that there is a $C^{1,1}_{\loc}$-regular diffeomorphism on $\tilde{W}$ which agrees with $\id$ on $\tilde{W}_<$. By continuity, it is then also an isometry on $\tilde{W}_\leq$. Higher regularity of the extension of $\id$ is obtained if the extensions $\tilde{\iota} : M \hookrightarrow \tilde{M}$ and $\hat{\iota} : M \hookrightarrow \hat{M}$ have higher regularity, see Lemma \ref{LemBoostReg}.

We also show in Section \ref{SecCounter} that the above theorem is sharp: we give an example of a smooth time-oriented and globally hyperbolic Lorentzian manifold which admits two anchored  H\"older continuous extensions with H\"older exponent $0 < \alpha <1$ arbitrarily close to $1$, but such that the identification map $\id$ does not extend as a $C^1$-regular map. This adds to the list of qualitative changes which can occur in Lorentzian geometry at the threshold of Lipschitz continuity of the metric: the push-up property may fail to hold and light cones may fail to be hypersurfaces \cite{ChrusGra12}, maximising causal curves may fail to have definite causal character \cite{SaeSt18}, \cite{GrafLing18}, and the future, when defined via locally Lipschitz curves, may fail to be open \cite{GraKuSaSt19}. Moreover, the example given in Section \ref{SecCounter} exemplifies the  non-triviality of the question whether given $C^0$-extensions, for example at the Cauchy horizon in the interior of perturbed Kerr back holes \cite{DafLuk17}, are unique.
\newline

To the best of the author's knowledge, the literature result most pertinent to our Theorem \ref{ThmIntro} is found in the beautiful paper \cite{Chrus10} by Chru\'sciel\footnote{For more related references we refer the interested reader to the introduction of \cite{Chrus10}.}:
\begin{theorem}[Theorem 3.1 in \cite{Chrus10}] \label{ThmCh}
Let $M$ be a smooth $d+1$-manifold, $d \geq 2$, and $g_{ab}$ a smooth Lorentz-signature metric thereon. Let $(\hat{M}, \hat{g}_{ab})$ be a smooth $d+1$-manifold with boundary with smooth Lorentz-signature metric, and let $M \overset{\psi}{\rightarrow} \hat{M}$ be a smooth diffeomorphism into, such that $i)$ $\psi[M] = \hat{M} \setminus \rd \hat{M}$, and $ii)$ the $\psi$-image of $g_{ab}$ is conformal to $\hat{g}_{ab}$. Similarly for $\hat{M}'$, $\hat{g}'_{ab}$, and $\psi'$. Let $I \overset{\gamma}{\rightarrow} M$, where $I$ is an open interval in $\R$, be a directed null geodesic in $(M,g_{ab})$ such that $\hat{\gamma} = \psi \circ \gamma$ has future endpoint $p 
\in \hat{M}$ and $\hat{\gamma}' = \psi' \circ \gamma$ has future endpoint $p' \in \hat{M}'$. We also assume that $\hat{\gamma}$ meets $\rd \hat{M}$, and $\hat{\gamma}'$ meets $\rd[\hat{M}']$, transversally. Then there exist open neighbourhoods $U$ of $p$ in $\hat{M}$ and $U'$ of $p'$ in $\hat{M}'$, together with a smooth diffeomorphism $U \overset{\phi}{\rightarrow} U'$, such that $I)$ $\phi(p) = p'$, and $II)$ restricted to $U \cap \psi[M]$, $\psi' \circ \psi^{-1} = \phi$.
\end{theorem}
Although the above theorems are clearly related, there are the following important differences:
\begin{enumerate}
\item Theorem \ref{ThmCh} is stated in the smooth category. However, it generalises to Lorentzian manifolds for which the metric is at least $C^{1,1}_{\loc}$-regular (Remark 3.4  in \cite{Chrus10}). This is a sharp threshold for the null-geodesic-based proof developed in \cite{Chrus10}. One of the improvements of Theorem \ref{ThmIntro} is to lower the regularity requirements to $C^{0,1}_{\loc}$ -- which is the optimal threshold as shown by our example in Section \ref{SecCounter}. 
\item Theorem \ref{ThmCh} is more general than ours in the sense that it also shows the local uniqueness of the extension of the differentiable structure for \emph{conformal} extensions. Since the statement of Theorem \ref{ThmCh} for conformal extensions is false in dimension $1+1$ (see  Remark 3.2 in \cite{Chrus10}), Chru\'sciel's result is necessarily restricted to $d \geq 2$. Our result only considers isometric extensions and is also valid for $d=1$.
\item The set-up in \cite{Chrus10} is to investigate the uniqueness of the differentiable structure of (conformal) boundary extensions which are achieved by attaching a \underline{smooth} ($C^{k,1}_{\loc}$, $k \geq 2$) boundary to the original manifold\footnote{One of the motivations for this problem is for example the question whether the differentiable structure at future null infinity, defined via a conformal extension, is unique.}. However, note that this set-up is for example not applicable to the second extension in Figure \ref{FigGlobExt} and, more importantly, to extensions which have a rough boundary. Consider for example an open set in Minkowski spacetime which is bounded by a rough boundary. Is the obvious extension locally unique? This question is still open.

The set-up in this paper is motivated by the problem of understanding the possible extensions of the maximal globally hyperbolic development of initial data for the Einstein equations. Thus we only consider globally hyperbolic Lorentzian manifolds. As already discussed, the future boundary of a globally hyperbolic spacetime in an extension is, in a certain sense, a $C^{0,1}_{\loc}$ hypersurface. This improves on the $C^{2,1}_{\loc}$-differentiability assumption present in \cite{Chrus10}. However, let us emphasise that no differentiability assumption on the boundary has to be made in Theorem \ref{ThmIntro}, but it follows that the most general future boundary of a globally hyperbolic spacetime is locally Lipschitz regular.
\item The anchoring curve $\gamma$ in Theorem \ref{ThmIntro} does not have to be a null geodesic nor does it have to meet the boundary transversally. Indeed, it does not even have to be differentiable at the boundary.
\end{enumerate}

\subsection{Further results, outline of paper, and brief outline of proofs} \label{SecFR}

Section \ref{SecPrel} recalls the definitions of an extension and that of the future boundary as well as a basic structural result of the future boundary from \cite{Sbie18}. We continue in Section \ref{SecFuBound} with the investigation of the future boundary of globally hyperbolic and time-oriented Lorentzian manifolds $(M,g)$: we show in Proposition \ref{PropLeavingLipschitz} that one may relax the assumption on the timelikeness as well as on the smoothness of the extension of the curve $\gamma$, which defines a future boundary point (see page \pageref{DefTextFB} or Definition \ref{DefFutureBdry}), and only demand $\gamma$ to be causal and to have a continuous extension to the boundary without changing the defined set of future boundary points \emph{as long as the extension is at least locally Lipschitz continuous}. Example \ref{ExFuBound} shows in particular that this is no longer true if the extension is merely continuous. 

Proposition \ref{PropLeavingLipschitz} in particular enlarges the class of possible anchoring curves in Theorem \ref{ThmIntro}. At the same time it furnishes the first step in the proof of Theorem \ref{ThmIntro}: Let $(M,g)$ be a globally hyperbolic and time-oriented Lorentzian manifold and $\tilde{\iota} : M \hookrightarrow \tilde{M}$ and $\hat{\iota} : M \hookrightarrow \hat{M}$ be two locally Lipschitz extensions. Given a future boundary point $\tilde{p} \in \rd^+ \tilde{\iota}(M)$ by definition we can choose a future directed timelike curve $\gamma$ in $M$ which extends smoothly to $\tilde{p}$ in $\tilde{M}$. A priori it is conceivable that the curve $\gamma$ only has a continuous future end-point $\hat{p}$ in $\rd \hat{\iota}(M)$, but not a smooth one.\footnote{A posteriori Theorem \ref{ThmIntro} shows that it must also extend at least as a $C^{1,1}_{\loc}$-regular timelike curve.} Proposition \ref{PropLeavingLipschitz} now shows that even under this assumption $\hat{p}$ must be a future boundary point of $\hat{\iota}(M)$ and that we can choose future boundary charts $\tilde{U}$ around $\tilde{p}$ and $\hat{U}$ around $\hat{p}$. 

In the next step of the proof of Theorem \ref{ThmIntro} we want to go over from the arbitrary anchoring curve $\gamma$ to a specific, suitably chosen anchoring curve. The new curve can be chosen say in $\tilde{M}$, but does it still converge to $\hat{p}$ in $\hat{M}$?
More broadly this raises the question under what condition two future directed and future inextendible causal curves $\gamma_1$ and $\gamma_2$ in $M$ have the same limit points in $\tilde{M}$ (or $\hat{M}$). We answer this question completely in Section \ref{SecC0} by using a characterisation familiar from the b-boundary construction in terms of the existence of a sequence of connecting curves whose generalised affine length tends to zero (see Propositions \ref{PropBBound} and \ref{PropFindCurves} for the details). The concept of generalised affine length is recalled in Section \ref{SecCoordAffine} and some basic, yet fundamental, a priori estimates are obtained on parallel transport under Lipschitz bounds and on the relation of generalised affine length and coordinate length. 

Theorem \ref{ThmIntro} is proved in Section \ref{SecC1}. The important step is here the proof of Proposition \ref{PropC1Aux}, which is actually a stronger result than Theorem \ref{ThmIntro} but has a less concise formulation. It establishes the local equivalence of two anchored extensions $\tilde{\iota} : M \hookrightarrow \tilde{M}$ and $\hat{\iota} : M \hookrightarrow \hat{M}$ under the assumption that $(\hat{M}, \hat{g})$ is locally Lipschitz continuous, \emph{but only requires $(\tilde{M}, \tilde{g})$ to be continuous and in addition to satisfy the following property: the anchoring curve is part of a space-filling family of future directed and future inextendible causal curves in $M$ which, in $\tilde{M}$, limit to and fill a small neighbourhood of $\tilde{p}$ up to the boundary -- and, moreover, parallel transport along this space-filling family of causal curves remains continuous (and bounded) in the local coordinates.}
Thus, the parallel transport of a frame field along this space-filling family of curves encodes the $C^1$-structure of $\tilde{M}$ at the boundary. One now uses the Lipschitz condition on $\hat{g}$ to show that 1) the space-filling family of curves is also contained in the future boundary chart $\hat{U}$ and 2) that the parallel transport of the frame field is also comparable to (encodes) the differentiable structure of $\hat{M}$. Thus, the parallel transport of an auxiliary frame field (which is invariant under the isometry $\id$) allows us to bridge between the differentiable structures on $\tilde{M}$ and $\hat{M}$ and to show that they are equivalent\footnote{Proposition \ref{PropC1Aux} is a key step in future work on the $C^{0,1}_{\loc}$-inextendibility of weak null singularities without any symmetry assumption.}. 

We now come back to the case of two $C^{0,1}_{\loc}$-extensions $\tilde{\iota} : M \hookrightarrow \tilde{M}$ and $\hat{\iota} : M \hookrightarrow \hat{M}$ and the proof of Theorem \ref{ThmIntro}. Recall that we chose two future boundary charts $\tilde{U}$ around $\tilde{p}$ and $\hat{U}$ around $\hat{p}$, where $\tilde{p}$ and $\hat{p}$ are limit points of the anchoring curve $\gamma$. In the chart $\tilde{U}$ with coordinates $(x_0, x_1, \ldots, x_d) = (x_0, \ux)$ we now go over to the curve of constant $\ux = 0$ which is contained in $M$ and limits to $\tilde{p}$. By the results of Section \ref{SecC0} mentioned earlier we can take this curve as a new anchoring curve. The space-filling family of future directed and future inextendible causal curves is now taken to be the curves of constant $\ux$ in the future boundary chart $\tilde{U}$. Since $\tilde{g}$ satisfies a local Lipschitz condition one can show that parallel transport along this family of curves remains bounded and continuous with respect to the local coordinates $x^\mu$ and thus the assumptions of Proposition \ref{PropC1Aux} are met. Finally, it is shown in Lemma \ref{LemBoostReg} that if the two extensions have additional regularity, the extension of $\id$ will also enjoy additional regularity.

In Section \ref{SecCounter} we give an example of a $1+1$-dimensional cosmological spacetime which is globally hyperbolic and time-oriented and moreover admits two anchored, yet inequivalent, H\"older continuous extensions: the identification map $\id$ extends as a $C^0$-regular map but not as a $C^1$-regular map. This shows that the Lipschitz regularity assumed on the extensions in Theorem \ref{ThmIntro} is, in general, optimal.

Recall that Chru\'sciel's proof of the local uniqueness of (anchored) conformal boundary extensions in more than two spacetime dimensions, i.e., Theorem \ref{ThmCh}, is based on null geodesics. In  Appendix \ref{SecA1} we give a proof of the local uniqueness of anchored (isometric) extensions which are sufficiently smooth (at least $C^2$) which is based on timelike geodesics and valid in all spacetime dimensions. On the one hand this provides an easier method of proof than the one presented based on parallel transport of a frame field \emph{if the extensions are smooth enough}. On the other hand we show in Appendix \ref{SecA2} how this method can be employed in other settings: we answer a question by JB Manchak in the affirmative as to whether the only possible (smooth) extension of an inextendible Lorentzian manifold with one point removed is the restoration of this point.

\subsection*{Acknowledgements}

I would like to thank Mihalis Dafermos for stimulating discussions when I first started thinking about this problem several years ago, and also three anonymous referees for their valuable comments.  Moreover, I acknowledge support through the Royal Society University Research Fellowship URF\textbackslash R1\textbackslash 211216.

\section{Preliminaries} \label{SecPrel}

In this paper the regularity of timelike and causal curves will always be stated explicitly; that is, there is no implicit assumption that timelike curves have to be smooth, piecewise $C^1$, locally Lipschitz or similarly. For a smooth curve to be timelike the velocity has to be timelike everywhere; for a piecewise $C^1$-curve to be timelike the velocity has to be timelike everywhere and at points of discontinuity it has to lie in the same connected component of the timelike double cone; for a locally Lipschitz curve to be timelike the velocity has to be timelike almost everywhere -- and similarly for causal curves. The timelike future $I^+$ is defined with respect to smooth curves. In this paper all metrics are at least $C^0$, and thus this definition agrees with the ones for piecewise $C^1$ or `locally uniformly timelike' curves, see  \cite{ChrusGra12}, \cite{GraKuSaSt19}.\footnote{Note that in these references our convention  of the definition of the timelike future is denoted by $\check{I}^+$.} If the metric is at least locally Lipschitz continuous $I^+$ can also be defined with respect to locally Lipschitz timelike curves; otherwise the so defined set might be strictly larger.

Furthermore, the causal future $J^+$ is defined with respect to locally Lipschitz causal curves. Note that for smooth metrics this agrees with the definition using smooth causal curves \cite{Chrus11a}.

\begin{definition} \label{DefExtension}
Let $(M,g)$ be a Lorentzian manifold\footnote{\label{FootnoteSmooth}All manifolds are assumed to be smooth (cf.\ Remark 2.2 in \cite{Sbie22a}). For the definition of a $\Gamma$-extension it makes sense to assume that the metric $g$ is at least $\Gamma$-regular.} and let $\Gamma$ be a regularity class, for example $\Gamma = C^k$ with $k \in \N \cup \{\infty\}$ or $\Gamma = C^{0,1}_{\loc}$.  A \emph{$\Gamma$-extension of $(M,g)$} consists of a smooth isometric embedding $\iota : M \hookrightarrow \tilde{M}$ of $M$ into a Lorentzian manifold $(\tilde{M}, \tilde{g})$ of the same dimension as $M$ where $\tilde{g}$ is $\Gamma$-regular  and such that $\partial \iota(M) \subseteq \tilde{M}$ is non-empty.

If $(M,g)$ admits a $\Gamma$-extension, then we say that $(M,g)$ is \emph{$\Gamma$-extendible}, otherwise we say $(M,g)$ is \emph{$\Gamma$-inextendible}.
\end{definition}

Note that the boundary $\rd \iota(M) \subseteq \tilde{M}$ can in general be very rough and certainly $\rd \iota (M) \cup \iota(M)$ does not need to have the structure of a manifold with boundary. In the following we single out a certain subset of boundary points which, \emph{in the case of $(M,g)$ being globally hyperbolic}, can be shown to have more structure; see the proposition below. 

\begin{definition} \label{DefFutureBdry}
Let $(M,g)$ be a time-oriented Lorentzian manifold with $g \in C^0$ and let $\iota : M \hookrightarrow \tilde{M}$ be a $C^0$-extension of $M$. The \emph{future boundary of $M$} is the set $\partial^+\iota(M) $ consisting of all points $\tilde{p} \in \tilde{M}$ such that there exists a smooth timelike curve $\tilde{\gamma} : [-1,0] \to \tilde{M}$ such that $\mathrm{Im}(\tilde{\gamma}|_{[-1,0)}) \subseteq \iota(M)$, $\tilde{\gamma}(0) = \tilde{p} \in \partial \iota(M)$, and $\iota^{-1} \circ \tilde{\gamma}|_{[-1,0)}$ is future directed in $M$.
\end{definition}
Clearly we have $\partial^+\iota(M) \subseteq \partial \iota(M)$. The past boundary $\partial^- \iota(M)$ is defined analogously.

The next proposition is found in \cite{Sbie18}, Proposition 2.2.

\begin{proposition}\label{PropBoundaryChart}
Let $\iota : M \hookrightarrow \tilde{M}$ be a $C^0$-extension of a time-oriented globally hyperbolic Lorentzian manifold $(M,g)$ with $g \in C^0$ and with Cauchy hypersurface $\Sigma$ --  and let $\tilde{p} \in \partial^+ \iota(M)$. For every $\delta >0$ there exists a chart $\tilde{\varphi} : \tilde{U} \to(-\varepsilon_0, \varepsilon_0) \times  (-\varepsilon_1, \varepsilon_1)^{d} =: R_{\varepsilon_0, \varepsilon_1}$, $\varepsilon_0, \varepsilon_1 >0$ with the following properties
\begin{enumerate}[i)]
\item $\tilde{p} \in \tilde{U}$ and $\tilde{\varphi}(p) = (0, \ldots, 0)$
\item $|\tilde{g}_{\mu \nu} - m_{\mu \nu}| < \delta$, where $m_{\mu \nu} = \mathrm{diag}(-1, 1, \ldots , 1)$
\item There exists a Lipschitz continuous function $f : (-\varepsilon_1, \varepsilon_1)^d \to (-\varepsilon_0, \varepsilon_0)$ with the following property: 
\begin{equation}\label{PropF1}
\{(x_0,\underline{x}) \in (-\varepsilon_0, \varepsilon_0) \times (-\varepsilon_1, \varepsilon_1)^{d} \; | \: x_0 < f(\underline{x})\} \subseteq \tilde{\varphi} \big( \iota\big(I^+(\Sigma,M)\big)\cap \tilde{U}\big)
\end{equation} and 
\begin{equation}\label{PropF2}
\{(x_0,\underline{x}) \in (-\varepsilon_0, \varepsilon_0) \times (-\varepsilon_1, \varepsilon_1)^{d}  \; | \: x_0 = f(\underline{x})\} \subseteq \tilde{\varphi}\big(\partial^+\iota(M)\cap \tilde{U}\big) \;.
\end{equation}
Moreover, the set on the left hand side of \eqref{PropF2}, i.e. the graph of $f$, is achronal\footnote{With respect to \emph{smooth} timelike curves.} in $(-\varepsilon_0, \varepsilon_0) \times  (-\varepsilon_1, \varepsilon_1)^{d}$.
\end{enumerate}
\end{proposition}

\begin{remark} \label{RemPropBoundaryChart}
\begin{enumerate}
\item Given a future boundary point $\tilde{p}$, we call a chart as above a \emph{future boundary chart}. A chart around some point $\tilde{p}$ in a Lorentzian manifold with a continuous metric which satisfies only $i)$ and $ii)$ is called a \emph{near-Minkowskian chart}.
\item The proof of Proposition \ref{PropBoundaryChart}, found in \cite{Sbie18}, shows that if $\tilde{\gamma} : [-1,0] \to \tilde{M}$ is a smooth timelike curve such that $\tilde{\gamma}(0) = \tilde{p} \in \rd^+\iota(M)$ and with  $\tilde{\gamma}|_{[-1,0)}$ being contained in $\iota(M)$ and being future directed there, then the boundary chart can be chosen such that, after a possible restriction and reparametrisation of $\tilde{\gamma}$, we have $(\tilde{\varphi} \circ \tilde{\gamma})(s) = (s, 0, \ldots, 0)$ for $s \in (-\varepsilon_0, 0]$. In particular $\tilde{\gamma}$ is ultimately contained below the graph of $f$.
\end{enumerate}
\end{remark}

In particular the future boundary is open in the sense that around every future boundary point one can find a future boundary chart as in the proposition such that all points on the graph of $f$ are also future boundary points. 
Let us also remark that for $C^0$-extensions the assumption that the timelike curve is smooth and timelike at the boundary point is crucial. If one defined the set of future boundary points just as \emph{limit points} of smooth future directed timelike curves, then the set would be strictly larger in general and Proposition \ref{PropBoundaryChart} does not hold, see Example \ref{ExFuBound}.
However, for $C^{0,1}_{\loc}$-extensions, it is shown in Proposition \ref{PropLeavingLipschitz} that one can relax the assumption on the smoothness and timelikeness of the timelike curve $\tilde{\gamma} : [-1,0] \to \tilde{M}$ at $s=0$ and just require continuity, i.e., $\lim_{s \to 0} \tilde{\gamma}(s) = \tilde{p}$, without enlarging the so-defined set.

\section{On the definition of the future boundary} \label{SecFuBound}

Let $(M,g)$ be a spacetime and  $\iota_i : M \hookrightarrow \tilde{M}_i$, $i = 1,2$, two extensions thereof. Assume there is a future boundary point $\tilde{p}_1 \in \rd^+\iota_1(M) \subseteq \tilde{M}_1$ together with a smooth timelike curve $\tilde{\sigma}_1 : [-1,0] \to \tilde{M}_1$ such that $\tilde{\sigma}_1|_{[-1,0)}$ maps as a future directed timelike curve into $\iota_1(M)$ and $\tilde{\sigma}_1(0) = \tilde{p}_1$. Then $\sigma := \iota^{-1}_1 \circ \tilde{\sigma}_1|_{[-1,0)}$ is a smooth and future directed timelike curve in $M$. In general, it does not need to have a future limit point in the second extension $\tilde{M}_2$. But even if it \emph{has} a future limit point $\tilde{M}_2 \ni \tilde{p}_2 := \lim_{s \to 0} (\iota_2 \circ \sigma)(s)$, it is not clear whether  $\tilde{p}_2 \in \rd^+ \iota_2 (M)$ -- and if this should be the case as well, whether $\sigma$  extends then as a smooth timelike curve to $\tilde{p}_2$.

In this section we address the first question: We show more generally, that as long as the extension $\iota : M \hookrightarrow \tilde{M}$  is locally Lipschitz, a future limit point in $\tilde{M}$ of a future inextendible timelike curve $\sigma : [-1,0) \to M$ is automatically a future boundary point. We also show that this property fails for merely continuous extensions. Moreover, we further prove that for locally Lipschitz extensions one can also relax the condition of the curve being timelike to it being causal.

We start with the case of smooth extensions which will be used as an auxiliary result.

\begin{lemma} \label{LemLeavingCurveSmooth}
Let $(M,g)$ be a smooth time-oriented and globally hyperbolic Lorentzian manifold and let $\iota : M \hookrightarrow \tilde{M}$ be a smooth extension. Let $\gamma : [-1,0) \to M$ be a future directed and future inextendible causal $C^1$-curve such that $\lim_{s \to 0} (\iota \circ \gamma)(s) =: \tilde{p} \in \rd \iota(M)$ exists. Then there is a smooth timelike geodesic $\tilde{\sigma} : [-1,0] \to \tilde{M}$ with $\tilde{\sigma}|_{[-1,0)}$ mapping into  $\iota(M)$ as a future directed timelike geodesic and $\tilde{\sigma}(0) = \tilde{p}$. In particular $\tilde{p} \in \rd^+ \iota(M)$ and there exists a boundary chart as in Proposition \ref{PropBoundaryChart} such that $\iota \circ \gamma$ is ultimately contained\footnote{I.e., there exists $s_0 \in [-1,0)$ such that $(\iota \circ \gamma)\big((s_0,0)\big)$ is contained in the following set. Moreover, for the sake of brevity, we have omitted here the application of $\tilde{\varphi}^{-1}$ to the set. This identification is understood tacitly here as it is also in other parts of the paper.} in $\{(x_0, \ux) \in (-\varepsilon_0, \varepsilon_0) \times (-\varepsilon_1, \varepsilon_1)^d \; | \; x_0 < f(\ux) \}$.
\end{lemma}

\begin{proof}
Let $\tilde{V} \subseteq \tilde{M}$ be a convex neighbourhood of $\tilde{p}$. After making $\gamma$ shorter if needed we can assume that $\tilde{\gamma} := \iota \circ \gamma : [-1,0) \to \tilde{M}$ is completely contained in $\tilde{V}$ -- and by adding a timelike piece to the past we can assume that in a small neighbourhood of $\tilde{\gamma}(-1)$ the curve is timelike. We pick a Cauchy hypersurface $\Sigma$ of $M$ with $\gamma(-1) \in I^+(\Sigma, M)$. First note that by the closedness of $J^+\big(\tilde{\gamma}(-\frac{1}{2}), \tilde{V}\big)$ relative to $\tilde{V}$ we have $\tilde{p} \in J^+\big(\tilde{\gamma}(-\frac{1}{2}), \tilde{V}\big)$. By the push-up property we then have $\tilde{p} \in I^+\big(\tilde{\gamma}(-1), \tilde{V}\big)$.

We extend $\tilde{\gamma}$ to $[-1,0]$ by setting $\tilde{\gamma}(0) := \tilde{p}$. Note that $\tilde{\gamma}$ is in general not differentiable at $0$ (and if it is, it might be null there). Consider the family of timelike geodesics $\tilde{\sigma}_s(t) := \exp_{\tilde{\gamma}(-1)} \big(t \exp_{\tilde{\gamma}(-1)}^{-1} \big[ \tilde{\gamma}(s)\big] \big)$ for $t \in [0,1]$ and $s \in (-1,0]$.

\textbf{\underline{Claim:}} For all $s \in (-1,0]$ the geodesics $\tilde{\sigma}_s|_{[0,1)}$ map into $\iota(M)$.

To prove the claim we consider $$A:= \{ s \in (-1,0] \; | \; \tilde{\sigma}_{s'}\big( [ 0,1)\big) \subseteq \iota(M) \; \textnormal{ for all } s'\leq s \}\;.$$ By the openness of $\iota(M)$ we have for all $s'$ close enough to $-1$ that $\tilde{\sigma}_{s'}\big([0,1)\big) \subseteq \iota(M)$ -- and thus in particular $A \neq \emptyset$. Moreover, for $s' \in (-1,0)$ with $\tilde{\sigma}_{s'}\big([0,1)\big) \subseteq \iota(M)$ we first note that indeed $\tilde{\sigma}_{s'}\big([0,1]\big) \subseteq \iota(M)$ since $\tilde{\sigma}_{s'}(1) = \tilde{\gamma}(s')$ -- and by the compactness of $\tilde{\sigma}_{s'}\big([0,1]\big) \subseteq \iota(M)$ and the openness of $\iota(M)$ we obtain that also $\tilde{\sigma}_{s''}\big([0,1]\big) \subseteq \iota(M)$ for slightly larger $s' < s''$. It thus follows that we must have either $A = (-1,0]$ or $A = (-1,s_0)$ with $-1< s_0 \leq 0$. The first case would prove the claim; thus we now rule out the second case. 

Assuming the second case, there is then $t_0 \in (0,1)$ with $\tilde{\sigma}_{s_0}(t_0) \notin \iota(M)$. We lead this to a contradiction to the global hyperbolicity of $(M,g)$ by constructing a past inextendible timelike curve $\beta$ in $M$ which is completely contained in $I^+(\Sigma, M)$, cf.\ Figure \ref{FigTransition} below.
\begin{figure}[h] 
  \centering
  \def\svgwidth{4.5cm}
   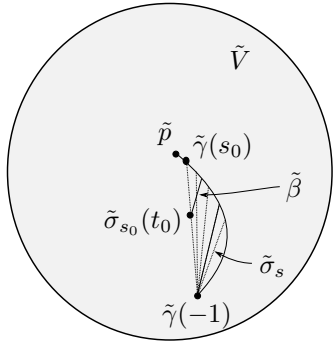 
      \caption{Construction of $\beta$.} \label{FigTransition}
\end{figure}
Although the idea is simple, the implementation is comparatively lengthy. We first note that for all $s' \in A$ we actually have $\tilde{\sigma}_{s'} \big((0,1]\big) \subseteq \iota\big( I^+(\Sigma, M)\big)$ since $\tilde{\gamma}(-1) \in \iota\big(I^+(\Sigma, M)\big)$. We now choose hyperboloidal normal coordinates at $\tilde{\gamma}(-1)$: we choose an orthonormal basis $\{e_0, e_1, \ldots, e_d\}$ of $T_{\tilde{\gamma}(-1)}\tilde{V}$ which gives us coordinates $\R^{d+1} \ni x^\mu \mapsto x^\mu e_\mu \in T_{\tilde{\gamma}(-1)} \tilde{V}$ for $T_{\tilde{\gamma}(-1)} \tilde{V}$. Setting $-\tau^2 := m_{\mu \nu} x^\mu x^\nu$ on $I^+_m(0, \R^{d+1})$, where $m_{\mu \nu} = \mathrm{diag}(-1,1, \ldots, 1)$, we obtain $-\tau d\tau = -x^0 dx^0 + x^i dx^i$ and $\mathrm{grad}_m \tau = - \frac{1}{\tau} (x^0 \rd_0 + x^i \rd_i)$. Note that $\mathrm{grad}_m \tau$ is radial. Consider the hypersurfaces $\{ \tau = \mathrm{const}, \; x^0 >0\}$, choose smooth global coordinates $y^i$ on one of these level sets and extend them to all of $I^+_m(0, \R^{d+1})$ by requiring $0 = \mathrm{grad}_{m} \tau (y^i)$. This gives us coordinates $(\tau, y^i)$ on $I^+_m(0, \R^{d+1}) \simeq I^+_{\tilde{g}|_{\tilde{\gamma}(-1)}}(0, T_{\tilde{\gamma}(-1)}\tilde{V})$.
Note that by construction $\mathrm{grad}_m \tau \sim \rd_\tau$. Moreover, we have $\mathrm{grad}_m \tau (\tau) = m^{-1}(d\tau, d\tau) = -1$ which gives $\mathrm{grad}_m \tau = - \rd_\tau$.  A suitable restriction of the $(\tau, y^i)$-coordinate domain ensures that it is contained in the domain of $\exp_{\tilde{\gamma}(-1)}$ which then gives us the hyperboloidal normal coordinates $(\tau, y^i)$ on $I^+(\tilde{\gamma}(-1), \tilde{V})$.
Moreover, it follows from the Gauss lemma that
\begin{equation*}
\tilde{g}= - d \tau^2 + \tilde{g}_{ij}(\tau, y^k) \, dy^i \otimes dy^j \;.
\end{equation*}
Since $\tilde{\gamma}$ is causal and $\tau$ is a time function on its domain of definition, we can reparametrise $\tilde{\gamma}$ by $\tau$, i.e.,
$$\tilde{\gamma}_{\mathrm{rp}}(\tau) = \big(\tau, \tilde{\gamma}_{\mathrm{rp}}^i(\tau) \big) = \tilde{\gamma}\big(s(\tau)\big) $$
with $\tau = \tilde{\gamma}^\tau(s)$ and $s(\tau) = (\tilde{\gamma}^\tau)^{-1}(\tau)$ being $C^1$.\footnote{The reason for the reparametrisation is to ensure a uniform bound on the derivative of $\tilde{\gamma}$ near its future endpoint $\tilde{p}$, cf.\ \eqref{EqUniBoundGammaTilde} below.} In $(\tau, y^i)$-coordinates $\tilde{\sigma}_s$ is given by $\tilde{\sigma}_s(t) = \big(t \cdot \tau(s), \tilde{\gamma}_{\mathrm{rp}}^i(\tau(s))\big)$. For $t \in (t_0, 1]$ and $\varepsilon >0$ sufficiently small we define in these coordinates $$\tilde{\beta}_\varepsilon(t) := \Big( t \cdot  \big[ \tau(s_0) - \varepsilon(t-t_0)\big], \tilde{\gamma}_{\mathrm{rp}}^i \big[ \tau(s_0) - \varepsilon(t-t_0)\big]\Big) \;.$$
Clearly, this is a family of $C^1$ curves depending on $\varepsilon>0$. We compute
\begin{equation}
\label{EqDerBeta}
\tilde{\beta}_\varepsilon'(t) = \big(\tau(s_0) - 2 \varepsilon t + \varepsilon t_0\big)\, \rd_\tau - \varepsilon \Big(\frac{d}{d \tau} \tilde{\gamma}_{\mathrm{rp}}^i \Big) \big[ \tau(s_0) - \varepsilon(t-t_0)\big] \, \rd_{y^i} \;.
\end{equation}
First note that this family of curves is contained in the set ruled by the timelike geodesics $\{\tilde{\sigma}_s \; | \; -1 < s<s_0\}$ -- and thus in particular in $\iota\big(I^+(\Sigma, M)\big)$ by definition of $s_0$.

Secondly, note that for all $\varepsilon >0$ sufficiently small, and since $t_0 \in (0,1)$, the family of curves $\tilde{\beta}_\varepsilon$ is contained in a compact $(\tau, y^i)$-coordinate set. Hence, on this set we have $|\tilde{g}_{ij}| \leq C$ by the continuity of the metric in these coordinates.

And finally note that the causality of $\tilde{\gamma}$ and the fact that $\tilde{\gamma}_{\mathrm{rp}}|_{\big[\tau(s_0) - \varepsilon(1- t_0), \tau(s_0)\big]}$ is contained in a compact coordinate neighbourhood gives the uniform bound
\begin{equation}\label{EqUniBoundGammaTilde}
1 \geq \tilde{g}_{ij}\big(\tau, \tilde{\gamma}_{\mathrm{rp}}^k(\tau) \big) \, \frac{d}{d\tau} \tilde{\gamma}_{\mathrm{rp}}^i(\tau) \cdot \frac{d}{d \tau} \tilde{\gamma}_{\mathrm{rp}}^j(\tau) \geq c || \frac{d}{d\tau} \tilde{\gamma}_{\mathrm{rp}}^i (\tau) ||^2_{\R^d}
\end{equation}
for all $\tau \in \big[\tau(s_0) - \varepsilon(1- t_0), \tau(s_0) \big)$ with $c>0$. Taking all this into account we obtain from \eqref{EqDerBeta}
$$\tilde{g}\big(\tilde{\beta}_\varepsilon' (t), \tilde{\beta}_\varepsilon'(t) \big) \leq -\big(\tau(s_0)\big)^2 + C\varepsilon + C \varepsilon^2 \;. $$ Since $\tau(s_0) >0$ it follows that for $\varepsilon >0$ sufficiently small $\tilde{\beta}_\varepsilon$ is a timelike curve in $\iota\big(I^+(\Sigma,M)\big)$. Pulling it back to $M$, it is also immediate that it is past inextendible. This gives the contradiction to the global hyperbolicity of $(M,g)$ and thus proves the claim.

The smooth timelike curve $\tilde{\sigma}_0$ now shows that $\tilde{p}$ is indeed in $\rd^+\iota(M)$. Thus by Proposition \ref{PropBoundaryChart} and Remark \ref{RemPropBoundaryChart} there exists a boundary chart such that $\tilde{\sigma}_0|_{[0,1)}$ is ultimately contained below the graph of $f$. It remains to show that $\tilde{\gamma}$ is also ultimately contained below the graph of $f$. Let $\tilde{W} \subseteq \tilde{M}$ be a neighbourhood of $\tilde{p}$ which is contained in $\tilde{V}$ as well as in the boundary chart. Clearly, $\tilde{\gamma}$ is ultimately contained in $\tilde{W}$ and our construction shows that each point on $\tilde{\gamma} \cap \tilde{W}$ can be connected to a point on $\tilde{\sigma}_0 \cap \tilde{W}$ by a curve in $\tilde{W}$ which is entirely contained in $\iota(M)$. This shows that $\tilde{\sigma}_0$ and $\tilde{\gamma}$ lie ultimately in the same connected component of $\iota(M)$ in the boundary chart, that is, below the graph of $f$.
\end{proof}

\begin{remark} \label{RemLemRougher}
\begin{enumerate}
\item The smoothness of the extension in Lemma \ref{LemLeavingCurveSmooth} is assumed in order to obtain that the curve $\tilde{\sigma}_{0}$ is smooth at $\tilde{p}$. Recall from Definition \ref{DefFutureBdry} that this is required for $\tilde{p}$ to be a future boundary point. 

If one only assumes $g, \tilde{g} \in C^2$, then the proof goes through to yield a timelike geodesic $\tilde{\sigma}_0$ which leaves $\iota(M)$ and is $C^1$ (indeed $C^3$) at $\tilde{p}$:  the exponential map is in $C^1$ and thus the hyperboloidal normal coordinates are also $C^1$. Thus, the metric components $\tilde{g}_{\mu \nu}$ in these coordinates are $C^0$, which is what is chiefly needed for the proof. 
One could then go on showing that if there is a timelike curve leaving $\iota(M)$ at $\tilde{p}$ and which is moreover $C^1$ at $\tilde{p}$, then $\tilde{p}$ is indeed a future boundary point according to Definition \ref{DefFutureBdry}. This being said however, note that Lemma \ref{LemLeavingCurveSmooth} is only used as an auxiliary lemma for the more general Proposition \ref{PropLeavingLipschitz}.

\item A desirable statement to have would be that, under the assumptions of Lemma \ref{LemLeavingCurveSmooth}, there is a small neighbourhood $\tilde{V} \subseteq \tilde{M}$ of $\tilde{p}$ such that $I^-(\tilde{p}, \tilde{V}) \subseteq \iota(M)$. If this statement were known a priori, the proof of Lemma \ref{LemLeavingCurveSmooth} would be immediate. The statement follows a posteriori from Proposition \ref{PropBoundaryChart} once it is established that $\tilde{p} \in \rd^+ \iota(M)$.

Note that if one assumes the extension $\iota : M \hookrightarrow \tilde{M}$ to be time-oriented as well and such that for every $\tilde{q} \in \tilde{M} \setminus \iota(M)$ we have $I^+(\tilde{q}, \tilde{M}) \cap \iota(M) = \emptyset$, then the desired statement is immediate. The complication that has to be circumvented in the general case is that there might not be a neighbourhood of $\tilde{p}$ that is disjoint from any Cauchy hypersurface,  cf.\ Remark 2.7 in \cite{Sbie18}.

\item Let us also remark that the assumption of global hyperbolicity of $(M,g)$ in Lemma \ref{LemLeavingCurveSmooth} cannot be dropped, which is clear from the example in Figure \ref{FigGHNeeded} in which the boundary point $\tilde{p}$ can only be reached by timelike curves which become null at $\tilde{p}$.
\begin{figure}[h] 
\centering
\begin{minipage}{0.8\textwidth}
  \centering
  \def\svgwidth{4.5cm}
\begingroup%
  \makeatletter%
  \providecommand\color[2][]{%
    \errmessage{(Inkscape) Color is used for the text in Inkscape, but the package 'color.sty' is not loaded}%
    \renewcommand\color[2][]{}%
  }%
  \providecommand\transparent[1]{%
    \errmessage{(Inkscape) Transparency is used (non-zero) for the text in Inkscape, but the package 'transparent.sty' is not loaded}%
    \renewcommand\transparent[1]{}%
  }%
  \providecommand\rotatebox[2]{#2}%
  \newcommand*\fsize{\dimexpr\f@size pt\relax}%
  \newcommand*\lineheight[1]{\fontsize{\fsize}{#1\fsize}\selectfont}%
  \ifx\svgwidth\undefined%
    \setlength{\unitlength}{399.95182548bp}%
    \ifx\svgscale\undefined%
      \relax%
    \else%
      \setlength{\unitlength}{\unitlength * \real{\svgscale}}%
    \fi%
  \else%
    \setlength{\unitlength}{\svgwidth}%
  \fi%
  \global\let\svgwidth\undefined%
  \global\let\svgscale\undefined%
  \makeatother%
  \begin{picture}(1,0.42190137)%
    \lineheight{1}%
    \setlength\tabcolsep{0pt}%
    \put(0,0){\includegraphics[width=\unitlength,page=1]{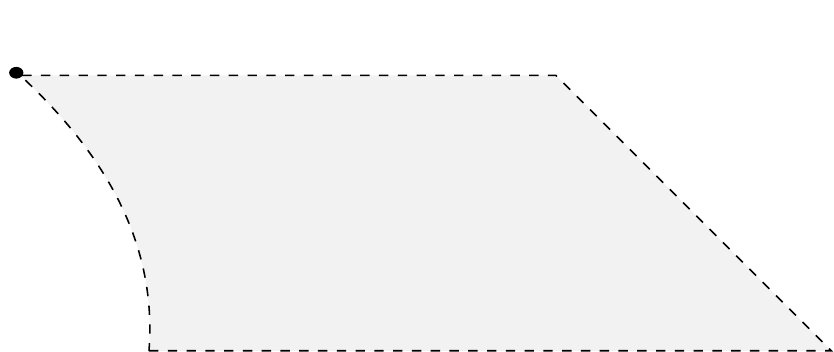}}%
    \put(-0.00622631,0.36491227){\color[rgb]{0,0,0}\makebox(0,0)[lt]{\lineheight{1.25}\smash{\begin{tabular}[t]{l}$\tilde{p}$\end{tabular}}}}%
    \put(0.66409085,0.04643584){\color[rgb]{0,0,0}\makebox(0,0)[lt]{\lineheight{1.25}\smash{\begin{tabular}[t]{l}$M$\end{tabular}}}}%
    \put(0,0){\includegraphics[width=\unitlength,page=2]{GHneeded.pdf}}%
    \put(0.19699203,0.14956155){\color[rgb]{0,0,0}\makebox(0,0)[lt]{\lineheight{1.25}\smash{\begin{tabular}[t]{l}$\gamma$\end{tabular}}}}%
  \end{picture}%
\endgroup%
 
      \caption{Consider the shaded part of $1+1$-dimensional Minkowski spacetime, where the left boundary is given by a timelike curve which becomes null at $\tilde{p}$.} \label{FigGHNeeded}
\end{minipage}
\end{figure}
\end{enumerate}
\end{remark}

The following proposition, which is based on Lemma \ref{LemLeavingCurveSmooth}, shows that Lemma \ref{LemLeavingCurveSmooth} remains true for local Lipschitz extensions.

\begin{proposition}\label{PropLeavingLipschitz}
Let $(M,g)$ be a  time-oriented and globally hyperbolic Lorentzian manifold with $g \in C^{0,1}_{\loc}$ and let $\iota : M \hookrightarrow \tilde{M}$ be a $C^{0,1}_{\loc}$-extension. Let $\gamma : [-1,0) \to M$ be a future directed and future inextendible causal $C^1$-curve such that $\lim_{s \to 0} (\iota \circ \gamma)(s) =: \tilde{p} \in \rd \iota(M)$ exists. Then $\tilde{p} \in \rd^+ \iota(M)$ and there exists a boundary chart as in Proposition \ref{PropBoundaryChart} such that $\iota \circ \gamma$ is ultimately contained in $\{(x_0, \ux) \in (-\varepsilon_0, \varepsilon_0) \times (-\varepsilon_1, \varepsilon_1)^d \; | \; x_0 < f(\ux) \}$.
\end{proposition}

\begin{proof}
\textbf{\underline{Step 0:}} Let $\tilde{\varphi} : \tilde{U} \to (- \tilde{\varepsilon}, \tilde{\varepsilon})^{d+1}$ be a near-Minkowskian neighbourhood centred at $\tilde{p}$ with $|\tilde{g}_{\mu \nu} (\tilde{x}) - m_{\mu \nu} | < \delta$, where $0 < \delta < \frac{1}{4}$, and with $\mathrm{Lip}(\tilde{g}_{\mu \nu}) \leq \Lambda$.\footnote{That is $|\tilde{g}_{\mu \nu}(\tilde{x}) - \tilde{g}_{\mu \nu}(\tilde{y})| \leq \Lambda \cdot ||\tilde{x}- \tilde{y}||_{\R^{d+1}}$ for all $\mu, \nu \in \{0, \ldots, d\}$ and $\tilde{x}, \tilde{y} \in (-\tilde{\varepsilon}, \tilde{\varepsilon})^{d+1}$.} Without loss of generality we can assume that $\tilde{\gamma}:= \iota \circ \gamma$ is completely contained in $\tilde{U}$. Moreover, since $\tilde{x}^0$ is a time-function on $\tilde{U}$, we can reparameterise $\tilde{\gamma}$ by $\tilde{x}^0$ to obtain another $C^1$-parameterisation of $\tilde{\gamma}$ in which the curve is Lipschitz with respect to the Euclidean metric associated to the coordinates $\tilde{x}^\mu$. Thus, without loss of generality, we can assume that $\tilde{\gamma} : [-\mu_0, 0) \to \tilde{U}$ is $C^1$ and causal with $\dot{\tilde{\gamma}}^0(s) = 1$ and $||\dot{\tilde{\gamma}}(s)||_{\R^{d+1}} \lesssim 1$ for all $s \in [-\mu_0, 0)$. Also recall that $\lim_{s \to 0} \tilde{\gamma}(s) = \tilde{p}$.

\textbf{\underline{Step 1:}} We now follow the construction in Lemma 1.15 in \cite{ChrusGra12} to construct a `locally uniformly timelike curve' in $\iota(M)$ with limit point $\tilde{p}$. We set
\begin{equation*}
\tilde{\Gamma}_\varepsilon(s) = \begin{pmatrix}
\tilde{\Gamma}^0_\varepsilon(s) \\ \tilde{\Gamma}_\varepsilon^i(s)
\end{pmatrix} := \begin{pmatrix}
\tilde{\gamma}^0(s) - \varepsilon f(s) \\ \tilde{\gamma}^i(s) 
\end{pmatrix} \;,
\end{equation*}
where $\varepsilon >0$ and $f$ will be chosen such that $f(0) = 0$ and $f(s) \geq 0$, cf.\ Figure \ref{FigGammaEps}. We estimate
\begin{equation*}
\begin{split}
\tilde{g}_{\tilde{\Gamma}_\varepsilon(s)} \big(\dot{\tilde{\Gamma}}_\varepsilon(s), \dot{\tilde{\Gamma}}_\varepsilon(s) \big) &= \Big[ \big(\tilde{g}_{\tilde{\Gamma}_\varepsilon(s)}\big)_{\mu \nu} - \big(\tilde{g}_{\tilde{\gamma}(s)}\big)_{\mu \nu} \Big] \dot{\tilde{\Gamma}}^\mu_\varepsilon(s) \dot{\tilde{\Gamma}}^\nu_\varepsilon(s) + \big(\tilde{g}_{\tilde{\gamma}(s)}\big) \big( \dot{\tilde{\gamma}}(s) - \varepsilon \dot{f}(s) \rd_0,  \dot{\tilde{\gamma}}(s) - \varepsilon \dot{f}(s) \rd_0\big) \\[3pt]
&\leq (d+1)^2 \Lambda || \tilde{\Gamma}_\varepsilon(s) - \tilde{\gamma}(s)||_{\R^{d+1}} \cdot || \dot{\tilde{\Gamma}}_\varepsilon(s) ||^2_{\R^{d+1}} \\[3pt]
&\qquad \quad + \tilde{g}_{\tilde{\gamma}(s)} \big( \dot{\tilde{\gamma}}(s), \dot{\tilde{\gamma}}(s) \big) - 2 \varepsilon \dot{f}(s) \tilde{g}_{\tilde{\gamma}(s)} \big( \dot{\tilde{\gamma}}(s), \rd_0\big) + \varepsilon^2 \big(\dot{f}(s)\big)^2 \tilde{g}_{\tilde{\gamma}(s)} (\rd_0, \rd_0) \\[3pt]
&\leq (d+1)^2 \Lambda \cdot \varepsilon f(s) \cdot 2 \Big( || \dot{\tilde{\gamma}}(s)||^2_{\R^{d+1}} + \varepsilon^2 \big( \dot{f}(s)\big)^2 \Big) - 2 \varepsilon\dot{f}(s) \tilde{g}_{\tilde{\gamma}(s)}\big(\dot{\tilde{\gamma}}(s), \rd_0\big) \\[3pt]
&\leq  \varepsilon \dot{f}(s) A(s) + \varepsilon f(s) B(s) + \varepsilon^3 f(s) \big(\dot{f}(s)\big)^2 C
\end{split}
\end{equation*}
with 
\begin{align*}
A(s) &= 2 \big| \tilde{g}_{\tilde{\gamma}(s)} \big( \dot{\tilde{\gamma}}(s), \rd_0 \big) \big| \\[3pt]
B(s) &= 2 (d+1)^2  \Lambda || \dot{\tilde{\gamma}}(s) ||^2_{\R^{d+1}} \\[3pt]
C&= 2(d+1)^2 \Lambda \;.
\end{align*}
Note that $A(s)$ and $B(s)$ are continuous and uniformly bounded away from $0$ and $+\infty$ on $[-\mu_0, 0)$.\footnote{To see that $A(s)$ is uniformly bounded away from $0$, let $\tilde{e}_\mu(s)$ be an ONB along $\tilde{\gamma}$ with $e_0(s) \sim \rd_0$. Thus, $|\tilde{g}\big(\dot{\tilde{\gamma}}(s), \tilde{e}_0 \big) |^2 \geq \sum_i \tilde{g}\big(\dot{\tilde{\gamma}}(s), \tilde{e}_i\big)^2 $.
Note that the matrices relating $\tilde{e}_\mu$ with $\rd_\mu$ are uniformly bounded. Thus, assuming to the contrary $|\tilde{g}\big(\dot{\tilde{\gamma}}(s), \tilde{e}_0\big)| \to 0$, then $\dot{\tilde{\gamma}}^\mu(s) \to 0$, contradicting our parametrisation $\dot{\tilde{\gamma}}^0(s) = 1$.}
Setting
$$f(s) := \int\limits_0^s - \frac{1}{A(t)} \exp \Big( \int\limits_t^s - \frac{B(r)}{A(r)} \,dr \Big) \, dt = \int\limits^0_s  \frac{1}{A(t)} \exp \Big( \int\limits_s^t  \frac{B(r)}{A(r)} \,dr \Big) \, dt \;,$$
which is a solution of  $\dot{f}(s) A(s) + f(s) B(s) = -1$ and $f(0) = 0$, we have that $f \geq 0$, $f \in C^0\big([-\mu_0, 0]\big) \cap C^1 \big([-\mu_0, 0) \big)$, and $\dot{f}$ is uniformly bounded on $[-\mu_0, 0)$. In particular $\tilde{\Gamma}_\varepsilon$ is $C^1$ on $[- \mu_0, 0)$.

It follows that $$\tilde{g}_{\tilde{\Gamma}_\varepsilon(s)} \big( \dot{\tilde{\Gamma}}_\varepsilon(s), \dot{\tilde{\Gamma}}_\varepsilon(s) \big) \leq - \varepsilon + \varepsilon^3 f(s) \big[\dot{f}(s) \big]^2 C \;.$$
Choosing $\varepsilon_0>0$ small enough we obtain that $\tilde{\Gamma}_\varepsilon(s)$ is contained in $\tilde{U}$ and
\begin{equation}
\label{EqUniTime}
\tilde{g}_{\tilde{\Gamma}_\varepsilon(s)} \big( \dot{\tilde{\Gamma}}_\varepsilon(s), \dot{\tilde{\Gamma}}_\varepsilon(s)\big) \leq - \frac{1}{2} \varepsilon \quad \textnormal{ for all } 0 \leq \varepsilon \leq \varepsilon_0\;.
\end{equation}
After making $\varepsilon_0$ even smaller if necessary we can moreover guarantee that
\begin{equation}
\label{EqContM}
\big[ \tilde{\gamma}^0(-\mu_0) - \varepsilon_0 f(-\mu_0) , \tilde{\gamma}^0(-\mu_0) \big] \times \{ \tilde{\gamma}^i(- \mu_0)\} \subseteq (-\tilde{\varepsilon}, \tilde{\varepsilon})^{d+1} \quad \textnormal{ is contained in } \iota(M).
\end{equation}
Note that $\dot{\tilde{\Gamma}}_{\varepsilon_0}(s) = \dot{\tilde{\gamma}}(s) - \varepsilon_0 \dot{f}(s) \rd_0$ is uniformly bounded in $|| \cdot ||_{\R^{d+1}}$ on $[-\mu_0, 0)$ and thus it follows from \eqref{EqUniTime} that the tangent is uniformly bounded away from the light cone. In the same way as in Lemma 1.15 and Proposition 1.2 from \cite{ChrusGra12} it follows that there is a smooth Lorentzian metric\footnote{Here, $\check{g} \prec \tilde{g}$ means that the causal cones of $\check{g}$ are contained in those of $\tilde{g}$.} $\check{g} \prec \tilde{g}$ on $\tilde{M}$ such that $\tilde{\Gamma}_{\varepsilon_0} : [- \mu_0, 0) \to \tilde{U}$ is timelike with respect to $\check{g}$.

\textbf{\underline{Step 2:}} We now show that $\{ \big( \tilde{\gamma}^0(s) - \varepsilon f(s), \tilde{\gamma}^i(s) \big) \in (-\tilde{\varepsilon}, \tilde{\varepsilon})^{d+1} \; | \; s \in [-\mu_0, 0), \; 0 \leq \varepsilon \leq \varepsilon_0 \}$ is contained in $\iota(M)$ using a similar method as in the proof of Lemma \ref{LemLeavingCurveSmooth}.

Let $\Sigma$ be a Cauchy hypersurface of $M$ with $\iota^{-1} \big( \tilde{\Gamma}_{\varepsilon_0}(-\mu_0) \big) \in I^+(\Sigma, M)$.  Let
$$A:= \Big{\{} s \in [-\mu_0, 0) \; | \; \big[\tilde{\gamma}^0(s') - \varepsilon_0 f(s'), \tilde{\gamma}^0(s') \big] \times \{\tilde{\gamma}^i(s') \} \subseteq \iota(M) \textnormal{ for all } -\mu_0 \leq s' \leq s \Big{\}} \;.$$
By \eqref{EqContM}, the compactness of $\big[\tilde{\gamma}^0(s') - \varepsilon_0 f(s'), \tilde{\gamma}^0(s') \big] \times \{\tilde{\gamma}^i(s') \}$, and the openness of $\iota(M)$,  it follows firstly that $A$ is non-empty, and, moreover, that if $s' \in A$, then we also have $s'' \in A$ for all $s''$ a little larger than $s'$. Thus, $A$ must be of the form $[-\mu_0, s_0)$ with $-\mu_0 < s_0 \leq 0$. Assuming that $s_0$ is strictly smaller than $0$, there exists then $0 < \varepsilon_1 \leq \varepsilon_0$ such that 
\begin{equation}
\label{EqNotInM}
\big(\tilde{\gamma}^0(s_0) - \varepsilon_1 f(s_0), \tilde{\gamma}^i(s_0) \big) \notin \iota(M) \;.
\end{equation}
First note that we have indeed
\begin{equation}
\label{EqIndeedinFuture}
\bigcup\limits_{-\mu_0 \leq s < s_0} \big[\tilde{\gamma}^0(s) - \varepsilon_0 f(s), \tilde{\gamma}^0(s) \big] \times \{\tilde{\gamma}^i(s) \} \subseteq \iota\big(I^+(\Sigma, M)\big) \;,
\end{equation}
since $\tilde{\Gamma}_{\varepsilon_0}$ is future directed timelike and $\rd_0$ is so, too. Now consider the family of curves $$\tilde{\beta}_\lambda(\varepsilon) := \Big( \tilde{\gamma}^0\big(s_0 - \lambda (\varepsilon_1 - \varepsilon)\big) - \varepsilon f(s_0), \tilde{\gamma}^i \big(s_0 - \lambda(\varepsilon_1 - \varepsilon)\big) \Big) \quad \textnormal{ with } \lambda >0 \textnormal{ and } \varepsilon \in [0, \varepsilon_1). $$
\begin{figure}[h] 
  \centering
  \def\svgwidth{4.5cm}
   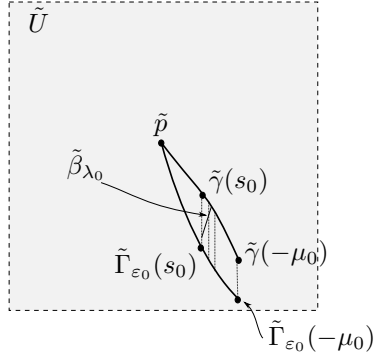 
      \caption{The curves $\tilde{\gamma}$, $\tilde{\Gamma}_{\varepsilon_0}$, and $\tilde{\beta}_{\lambda_0}$.} \label{FigGammaEps}
\end{figure}
We have  $$\dot{\tilde{\beta}}_\lambda(\varepsilon) = \Big[ \lambda \dot{\tilde{\gamma}}^0\big( s_0 - \lambda(\varepsilon_1 - \varepsilon) \big) - f(s_0) \Big] \, \rd_0 + \lambda \dot{\tilde{\gamma}}^i\big( s_0 - \lambda(\varepsilon_1 - \varepsilon)\big) \, \rd_i \;.$$
Since $||\dot{\tilde{\gamma}} ||_{\R^{d+1}} \lesssim 1$ and $f(s_0) >0$ we can fix $\lambda_0 >0$ small enough such that $\tilde{\beta}_{\lambda_0} : [0, \varepsilon_1) \to \tilde{U}$ is a past directed timelike curve. By \eqref{EqIndeedinFuture} its image under $\iota^{-1}$ is contained in $I^+(\Sigma, M)$ and clearly it is past inextendible in $M$ by \eqref{EqNotInM}. This contradicts the global hyperbolicity of $(M,g)$ and thus we must have $s_0 = 0$.

\textbf{\underline{Step 3:}} Note that it follows in particular from Step 2 that $\tilde{\Gamma}_{\varepsilon_0} : [- \mu_0, 0) \to \tilde{U}$ maps into $\iota(M)$ with $\lim_{s \to 0} \tilde{\Gamma}_{\varepsilon_0}(s) = \tilde{p}$ and, moreover, by the construction in Step 1, it is $C^1$. We now consider the smooth Lorentzian manifold $(M, \iota^* \check{g} )$ which is also globally hyperbolic since the causal cones of $\iota^* \check{g}$ are contained in those of $g$. Moreover, $(\tilde{M}, \check{g})$ is a smooth extension of $(M, \iota^*\check{g})$. We can now apply Lemma \ref{LemLeavingCurveSmooth} to infer that there exists a smooth $\check{g}$-timelike curve $\tilde{\sigma} : [-1,0] \to \tilde{M}$ with $\tilde{\sigma}|_{[-1,0)}$ being future directed in $\iota(M)$ and $\tilde{\sigma}(0) = \tilde{p}$. Since $\check{g} \prec \tilde{g}$, this curve is also $\tilde{g}$-timelike. We can now forget about the metric $\check{g}$ and apply 
Proposition \ref{PropBoundaryChart} to $(M,g)$ and $(\tilde{M}, \tilde{g})$ to obtain a boundary chart with $\tilde{\sigma}|_{[-1,0)}$ and $\tilde{\Gamma}_{\varepsilon_0}$ being ultimately contained below the graph of $f$. Finally, by Step 2, $\tilde{\Gamma}_{\varepsilon_0}$ and $\iota \circ \gamma$ are ultimately contained in the same connected component of $\iota(M)$ in the boundary chart which shows that $\iota \circ \gamma$ is ultimately contained below the graph of $f$.  
\end{proof}

The following example is instructive since it provides a counterexample to various statements about merely continuous extensions. In particular it shows that the analogue of Proposition \ref{PropLeavingLipschitz} for merely continuous extensions does not hold.
\begin{example} \label{ExFuBound}
We consider $\R^2$ with the canonical coordinates $(x,y)$ and the family of curves $[0,\frac{1}{2}] \ni s \mapsto (s,s^\alpha)$ with $\alpha>1$. The velocity vector is given by $\partial_x + \alpha s^{\alpha-1} \partial_y$. We want to turn those curves into null curves of a Lorentzian metric on a subset of $\R^2$. Moreover, on the curves the relation $\alpha = \frac{\log y}{\log x}$ holds so that  the vector field $ L:= \partial_x + \frac{y}{x} \frac{\log y}{\log x} \partial_y$, defined on $$M:= \{(x,y) \in \R^2 \; | \; x^3 < y < x^2 \; \textnormal{ and } \; 0 < x < \frac{1}{2}\}\;,$$ is tangent to the family of curves. We also set $\underline{L} := \rd_y$ and define the inverse of a smooth Lorentzian metric on $M$ by
\begin{equation*}
g^{-1} := -\big[L \otimes \underline{L} + \underline{L} \otimes L\big] = -\big[\partial_x \otimes \partial_y + \partial_y \otimes \partial_x\big] -2 \frac{y \log y}{x \log x} \,\partial_y \otimes \partial_y \;.
\end{equation*}
It thus follows that
\begin{equation*}
g = 2 \frac{y \log y}{x \log x} dx^2 - \big[dx \otimes dy + dy \otimes dx\big] 
\end{equation*}
and the vector fields $L$ and $\underline{L}$ are null by construction.

We now construct a $C^0$-extension $\iota : M \hookrightarrow \R^2_{x<\frac{2}{3}}$. We extend $L$ continuously to $\R^2_{x<\frac{2}{3}}$ by
\begin{equation*}
\tilde{L} := \begin{cases} \partial_x + \frac{y \log y}{x \log x} \partial_y \qquad &\textnormal{ for } x^3 \leq y \leq x^2 \; \textnormal{ and } \; 0< x < \frac{2}{3} \\
\partial_x + 3x^2\partial_y  &\textnormal{ for } -\infty < y \leq x^3 \; \textnormal{ and } \; 0< x < \frac{2}{3} \\
\partial_x + 2x \partial_y &\textnormal{ for } x^2 \leq y  < \infty \; \textnormal{ and } \; 0< x < \frac{2}{3} \\
\rd_x &\textnormal{ for } x \leq 0  \end{cases}
\end{equation*}
and set $\tilde{\underline{L}} := \rd_y$ on $\R^2_{x<\frac{2}{3}}$. Then $\tilde{g}^{-1} := -\big[\tilde{L} \otimes \tilde{\underline{L}} + \tilde{\underline{L}} \otimes \tilde{L}\big]$ defines the inverse of a continuous Lorentzian metric on $\R^2_{x<\frac{2}{3}}$ which, under the canonical inclusion map $\iota : M \hookrightarrow \R^2_{x<\frac{2}{3}}$, furnishes a $C^0$-extension of $(M,g)$. We fix the time orientation by stipulating that $\tilde{L}$ and $\tilde{\underline{L}}$ are future directed.

\begin{figure}[h]
  \centering
  \begin{minipage}{.47\textwidth}
  \centering
  \def\svgwidth{7cm}
   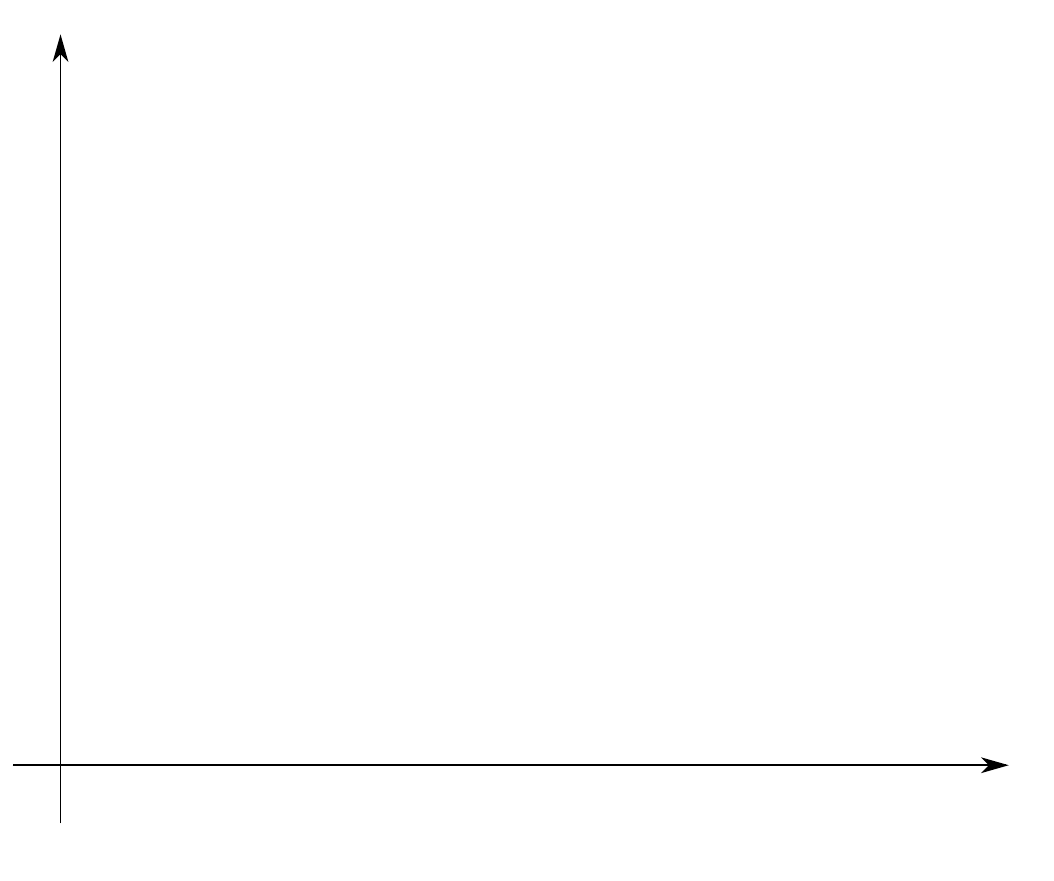
      \caption{The points $(0,0)$ and $(\frac{1}{2}, \frac{1}{8})$ lie in $\rd \iota(M)$ but not in $\rd^-\iota(M)$.}
      \end{minipage}\hspace*{5mm}
       \begin{minipage}{.47\textwidth}
  \centering
  \def\svgwidth{7cm}
   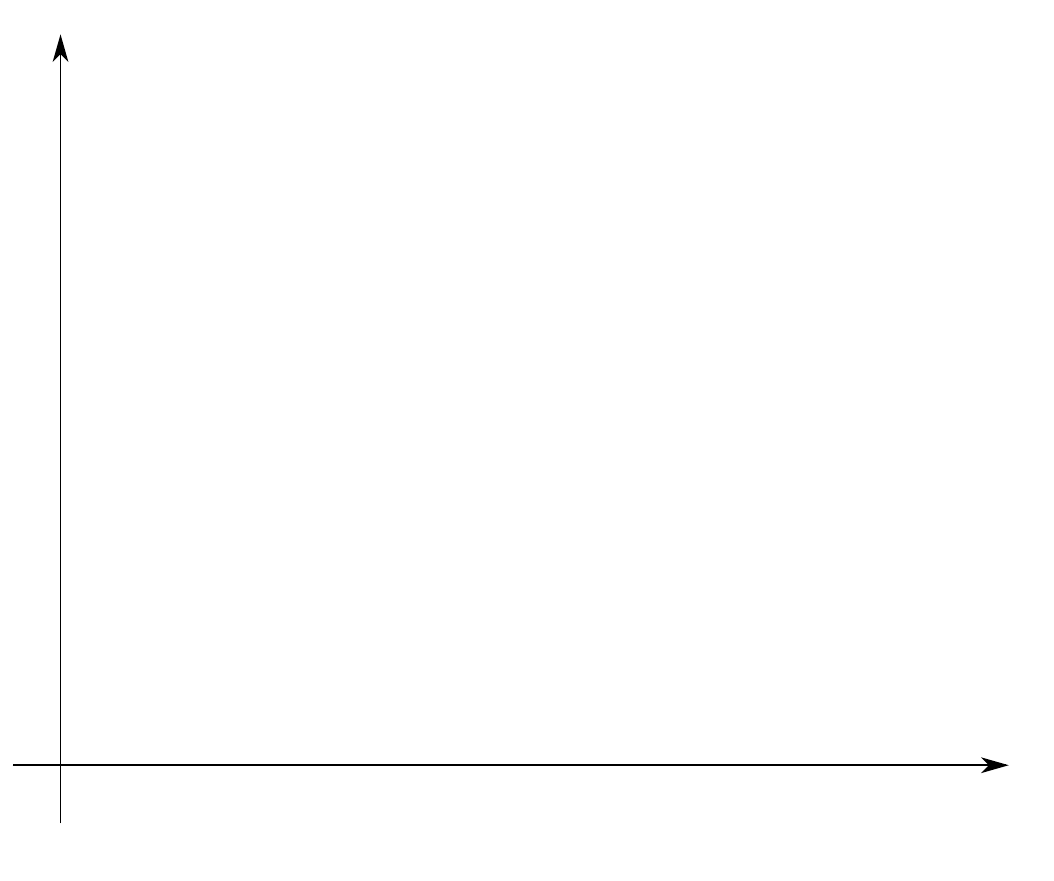
      \caption{The Cauchy hypersurface $\Sigma$ and timelike curve $\gamma$.}
      \end{minipage}
\end{figure}

Consider the point $(\frac{1}{2}, \frac{1}{4})$. Its past is delineated by the left going null curve which goes to $0$ and the right going null curve which goes straight down. Along the right going null curve, at the point $(\frac{1}{2}, \frac{1}{8})$, we switch to the left going past directed null curve which catches up with the other left and past going null curve at the origin! This set delineates $M$.

Firstly, let us note that $(M,g)$ is globally hyperbolic: since we know the exact form of the null geodesics it is easy to see that $J^+(p, M) \cap J^-(q,M)$ is compact for all $p,q \in M$. Furthermore, we do have that $y$ is strictly increasing along every future directed causal curve, which implies causality. These two conditions are sufficient for global hyperbolicity, see \cite{BerSa07}. 

One can also show that the curve $(0, \frac{1}{2}) \ni s \mapsto \Sigma(s) = (s, s^2 - s^3)$ parametrises a Cauchy hypersurface. A direct computation gives $$g\big(\dot{\Sigma}(s), \dot{\Sigma}(s) \big) =  2s^2 + 2s(1-s) \frac{\log(1-s)}{\log s} >0 \qquad \textnormal{ for } 0<s<\frac{1}{2}\;,$$ which shows that it is spacelike. It is also achronal. Let now $q \in M$ be a point above $\Sigma$. Then the left and past directed null curve starting at $q$ is of the form $s \mapsto (s, s^\alpha)$ with $2<\alpha <3$. Comparing asymptotics in $s$, this shows that the past of $q$ cuts out a compact portion of $\Sigma$, i.e., any past inextendible causal curve has to intersect $\Sigma$. Similarly, it is easily seen that if $q$ lies below $\Sigma$, the future of $q$ cuts out a compact portion of $\Sigma$ and thus any future inextendible curve has to intersect $\Sigma$. 

Secondly, consider the curve $(0, \frac{1}{2}) \ni s \mapsto \gamma(s) = (s, s^3 + s^4)$. A direct computation yields $$g\big(\dot{\gamma}(s), \dot{\gamma}(s)\big) = -2s^3 + 2s^2(1+s) \frac{\log(1+s)}{ \log s} <0 \qquad \textnormal{ for } 0 < s < \frac{1}{2}$$
which shows that $\gamma$ is a smooth timelike curve in $M$ with $\lim_{s \to 0} (\iota \circ \gamma)(s) = (0,0)$. And while it extends as a smooth curve to $\tilde{M}$, it does not extend as a smooth \emph{timelike} curve to $\tilde{M}$, since it is null at the origin. It can be easily shown that indeed any smooth and future directed timelike curve in $\tilde{M}$, which starts at $(0,0)$, cannot directly enter $\iota(M)$. Thus the origin is not a past boundary point of $\iota(M)$. 
\textbf{This shows that for globally hyperbolic \underline{$C^0$-extensions} it does make a difference whether one defines the future/past boundary with respect to smooth timelike curves in $\tilde{M}$ or just as limit points of smooth timelike curves in $M$. Thus, the analogue of Proposition \ref{PropLeavingLipschitz}  for $C^0$-extensions is false. Moreover, note that the point $(0,0)$ in this example shows that Proposition \ref{PropBoundaryChart} would not hold if one defined the future/past boundary just as limit points of smooth timelike curves in $M$.}

Finally, let us remark that in the language of \cite{ChrusGra12} the spacetime $(\R^2_{x<\frac{2}{3}}, \tilde{g})$ has a causal bubble at $(0,0)$. 
Using the terminology from \cite{ChrusGra12} it is an example of a spacetime in which the future of a point (here: the origin) with respect to `locally uniformly timelike curves' (smooth timelike curves) is a strict subset of the future with respect to locally Lipschitz timelike curves.\footnote{In the notation of \cite{ChrusGra12}, \emph{which is not used in this paper}, that is $\check{I}^+(0, \R^2_{x<\frac{2}{3}}) \subsetneq I^+(0, \R^2_{x<\frac{2}{3}})$.} Another causal bubble with this property was constructed in \cite{GraKuSaSt19}. The example here differs from that in \cite{GraKuSaSt19} qualitatively in the sense that a continuous family of left going null geodesics all intersect in one point -- while in \cite{GraKuSaSt19} the family of left going null geodesics meet, one by one, on a null curve.
\end{example}

\section{Bounding coordinate distance by generalised affine length} \label{SecCoordAffine}

Let $(M,g)$ be a Lorentzian manifold with $g \in C^1$. In particular, parallel transport is well-defined in this regularity class.
We briefly recall the \emph{generalised affine parameter} along an arbitrary  $C^1$-curve $\R \supseteq I \ni s \mapsto \gamma(s) \in M$. This concept was introduced to general relativity by Schmidt in \cite{Schm71}, but see also \cite{HawkEllis}. Let $p = \gamma(s_0)$ be a point on $\gamma$ and let $\{e_\alpha\}$, $\alpha  \in \{0, \ldots, d\}$, be an orthonormal basis of $T_pM$.\footnote{I.e., we have $g(e_\alpha, e_\beta) = \eta_{\alpha \beta}$ with $\eta = \mathrm{diag}(-1, 1, \ldots, 1)$.} By parallel transport along $\gamma$ we extend the orthonormal frame to a continuous orthonormal frame field $I \ni s \mapsto \{e_\alpha(s)\}$ along $\gamma$.  We can now expand $\dot{\gamma}(s) = \dot{\gamma}^\alpha(s) e_\alpha(s)$ and define the generalised affine parameter $\tau(s):= \int_{s_0}^s \sqrt{\sum_\alpha \dot{\gamma}^\alpha(s) \dot{\gamma}^\alpha(s)} \, ds$. It clearly depends on the base point $p = \gamma(s_0)$ and the choice of orthonormal basis $\{e_\alpha\}$ of $T_pM$. The \emph{generalised affine parameter length}, or for short \emph{gap-length}, of the curve $\gamma$ with respect to the orthonormal basis $\{e_\alpha\}$ of $T_pM$ is then given as
\begin{equation*}
\Lg(\gamma) :=  \int_I \underbrace{\sqrt{\sum_\alpha \dot{\gamma}^\alpha(s) \dot{\gamma}^\alpha(s)}}_{=:||\dot{\gamma}^\alpha(s)||_{\R^{d+1}}} \, ds \;.
\end{equation*}
Let now $q = \gamma(s_1)$ be another point on $\gamma$ and $\{f_\beta\}$, $\beta \in \{0, \ldots, d\}$ an orthonormal basis of $T_qM$ which is extended, by parallel transport, to an orthonormal frame field $\{f_\beta(s)\}$ along $\gamma$. It thus follows that there exists a Lorentz transformation $L$ with $e_\alpha (s)= L_\alpha^{\;\, \beta} f_\beta(s)$ for all $s \in I$. Writing $\dot{\gamma}(s) = \dot{\gamma}^\alpha_e(s) e_\alpha(s) = \dot{\gamma}^\alpha_f(s) f_\alpha(s)$ to distinguish the expansion of $\dot{\gamma}$ with respect to the two frame fields, we obtain $\dot{\gamma}^\alpha_e(s) L_\alpha^{\; \, \beta} = \dot{\gamma}^\beta_f(s)$ and thus
\begin{equation*}
||L^{-1}||_{\R^{d+1}}^{-1}  || \dot{\gamma}^\alpha_e(s)||_{\R^{d+1}}  \leq || \dot{\gamma}^\beta_f(s)||_{\R^{d+1}} \leq ||L||_{\R^{d+1}} || \dot{\gamma}^\alpha_e(s)||_{\R^{d+1}} 
\end{equation*}
where we have used that the matrix $L_\alpha^{\; \, \beta}$ is invertible. Hence, we also have 
\begin{equation}
\label{EqGapLorentzT}
||L^{-1}||_{\R^{d+1}}^{-1} \Lg(\gamma) \leq L_{\mathrm{gap}, f_\beta}(\gamma) \leq  ||L||_{\R^{d+1}} \Lg(\gamma) \;.
\end{equation}
Thus, although the gap-length of a curve depends on the choice of frame field, any two such choices yield comparable results. Let us remark that the concept of gap-length readily extends to curves which are only piecewise $C^1$.

Another take on the gap-length is to consider the coframe field $\{\omega^\alpha(s)\}$ which is dual to the parallel frame field $\{e_\alpha(s)\}$ along $\gamma$, i.e., so that we have $\omega^\beta(s) \big(e_\alpha(s)\big) = \delta^\beta_{\; \, \alpha}$. Clearly, $\{\omega^\alpha(s)\}$ is also parallel along $\gamma$. We can then define a Riemannian metric along $\gamma$ by $h_{e_\alpha}(s) := \omega^0 (s)\otimes \omega^0(s) + \ldots + \omega^d (s)\otimes \omega^d(s)$ to obtain $\Lg(\gamma) = \int_I \sqrt{ h_{e_\alpha}\big(\dot{\gamma}(s), \dot{\gamma}(s)\big)} \, ds$.

\begin{lemma} \label{LemGapCoord}
Consider an open subset $V \subseteq \R^{d+1}$ with global canonical coordinates $x^\mu$ and a $C^1$-regular Lorentzian metric $g$ which satisfies, in the coordinates $x^\mu$, the uniform bounds
\begin{equation}
\label{EqLemBoundsG}
|g^{\mu \nu}|,\; |g_{\mu \nu}| ,\; |\rd_\kappa g_{\mu \nu}| \leq C_g 
\end{equation}
for all $\mu, \nu, \kappa \in \{0, \ldots, d\}$ for $0<C_g < \infty$. Let $p \in V$ and let $\{e_\alpha\}$ be an orthonormal basis of $T_pV$ with $e_\alpha = e_\alpha^{\; \, \mu} \rd_\mu$ and $\rd_\mu = (e^{-1})^{\; \,\alpha}_{\mu} e_\alpha$, where the matrices $e$ and $e^{-1}$ are bounded by $0< C_e < \infty$ in operator norm:
\begin{equation}
\label{EqLemBoundsE}
||e||_{\R^{d+1}} , \; ||e^{-1}||_{\R^{d+1}} \leq C_e \;.
\end{equation}
Let $\overline{g} := (dx^0)^2 + \ldots + (dx^d)^2$ be the coordinate Euclidean metric, $0< a \leq \infty$, and let\footnote{In the case of $a< \infty$, the interval on which $\gamma$ is defined may also be closed.} $\gamma :[0, a) \to V$ be a piecewise $C^1$-curve starting at $p = \gamma(0)$ with coordinate length (i.e., with respect to the metric $\overline{g}$) bounded by $0<C_L < \infty$:
\begin{equation}
\label{EqLemBoundL}
L_{\overline{g}}(\gamma) \leq C_L\;.
\end{equation}
Let $e_\alpha(s)$ denote the parallel transport  of $e_\alpha$ along $s \mapsto \gamma(s)$ and write $e_\alpha(s) = e_\alpha^{\; \, \mu}(s) \rd_\mu$ and $\rd_\mu = (e^{-1})^{\; \,\alpha}_{ \mu}(s) e_\alpha(s)$.
Then there is a constant $0<B< \infty$ depending only on $C_g, C_e, C_L$ (and in particular independent of the point $p$) such that for the matrices $e(s)$ and $e^{-1}(s)$ we have 
\begin{equation}
\label{EqLemRes1}
||e(s)||_{\R^{d+1}}, \; ||e^{-1}(s)||_{\R^{d+1}} \leq B \quad \textnormal{ for all } s \in [0, a)\;.
\end{equation}
In particular, it follows that 
\begin{equation}\label{EqLemRes2}
B^{-1} \cdot  \Lg(\gamma) \leq L_{\overline{g}}(\gamma) \leq B \cdot \Lg(\gamma) \;.
\end{equation}
\end{lemma}
The result \eqref{EqLemRes1} can be read as the Riemannian metrics $\overline{g}|_{\gamma(s)}$ and $h_{e_\alpha}(s)$ along $\gamma$ being comparable. For this to hold it is in general needed that the curve $\gamma$ is of finite coordinate length, cf.\ Remark \ref{RemFiniteCoordLength}.
\begin{proof}
Let us first assume that $\gamma$ is $C^1$. We parameterise $\gamma$ by $\overline{g}$-arclength so that we can assume without loss of generality $\gamma : \big[0, L_{\overline{g}}(\gamma)\big) \to V$ with $||\dot{\gamma}||_{\R^{d+1}} = 1$ and $L_{\overline{g}}(\gamma) \leq C_L$. The parallel frame $e_\alpha(s)$ along $\gamma$ satisfies the parallel transport equation
\begin{equation}\label{EqPfLem1}
0 = \frac{D}{ds} X = \frac{d}{ds} X^\mu + \Gamma^\mu_{\rho \sigma}\big(\gamma(s)\big) \dot{\gamma}(s)^\rho X^\sigma(s) =: \frac{d}{ds} X^\mu +  A^\mu_{\; \, \sigma}(s) X^\sigma(s)\;,
\end{equation}
where, by \eqref{EqLemBoundsG} and $||\dot{\gamma}||_{\R^{d+1}} = 1$, we have $||A(s)||_{\R^{d+1}} \leq C(C_g)$. A standard Gronwall argument gives for solutions of \eqref{EqPfLem1}
\begin{equation} \label{EqPfLem2}
||X(s)||_{\R^{d+1}} \leq ||X(0)||_{\R^{d+1}} \exp\big(C(C_g) \cdot s\big) \;.
\end{equation}
Going over to the matrix formulation 
\begin{equation*}
0 = \frac{d}{ds} F^{\mu}_{\; \, \kappa}(s) + A^\mu_{\; \, \nu}(s) F^\nu_{\; \, \kappa}(s)
\end{equation*}
of the parallel transport equation \eqref{EqPfLem1} and prescribing $F(0)^\mu_{\; \, \kappa} = \id^\mu_{\; \, \kappa}$, we directly obtain from \eqref{EqPfLem2} for the solution $F(s)$ the estimate $||F(s)||_{\R^{d+1}} \leq \exp\big(C(C_g) \cdot s\big)$. Note that the matrix $F(s) : \R^{d+1} \to \R^{d+1}$ takes the coordinates of a vector $X \in T_pM$ as initial conditions and evolves it forward, by parallel transport, to give the coordinate expression of the vector at $\gamma(s)$. Since linearly independent initial vectors gives rise to linearly independent solutions, the matrix $F(s)$ is invertible for all $s \in \big[0, L_{\overline{g}}(\gamma)\big)$. It is easily seen that the inverse matrix $F(s)^{-1} : \R^{d+1} \to \R^{d+1}$ takes the coordinate expression of a vector at $\gamma(s)$ and evolves it, by parallel transport, back to $\gamma(0)$. This gives us the bounds
\begin{equation}\label{EqLemPfF}
||F(s)||_{\R^{d+1}},\; ||F(s)^{-1}||_{\R^{d+1}} \leq \exp\big(C(C_g) \cdot s\big) \;.
\end{equation}
Moreover, we have $e_{\alpha}^{\;\,\mu}(s) = F^{\mu}_{\; \, \kappa}(s) e^{\; \,\kappa}_{ \alpha}$ so that \eqref{EqLemRes1} follows from \eqref{EqLemPfF} together with \eqref{EqLemBoundL} and \eqref{EqLemBoundsE}.

Finally note that it follows from  $\dot{\gamma}(s) = \dot{\gamma}^{\alpha}(s) e_\alpha(s) = \dot{\gamma}^\mu(s) \rd_\mu$ that $\dot{\gamma}^\mu(s) = \dot{\gamma}^\alpha(s) e_\alpha^{\; \, \mu}(s)$ and $\dot{\gamma}^\alpha(s) = \dot{\gamma}^\mu(s) (e^{-1})_\mu^{\; \, \alpha} (s)$. Together with \eqref{EqLemRes1} we obtain for all $s$ $$||\dot{\gamma}^\mu(s)||_{\R^{d+1}} \sim ||\dot{\gamma}^\alpha(s)||_{\R^{d+1}} \;,$$ from which \eqref{EqLemRes2} follows.  

The case of $\gamma$ being piecewise $C^1$ follows analogously.
\end{proof}

\begin{remark}
\label{RemFiniteCoordLength}
The above lemma shows in particular that under the Lipschitz bound \eqref{EqLemBoundsG} on the metric the gap-length of  a curve cannot blow up in finite coordinate length. However, under the Lipschitz bound, there are curves of finite gap-length which have infinite coordinate length: consider for instance trapped null geodesics  which are inextendible but affine incomplete\footnote{Recall that for geodesics the affine parameter is in particular also a generalised affine parameter.} \cite{Mis67}, compare also with \cite{Schm73}. However, for small enough gap-lengths, where the smallness parameter depends on the geometry, the above lemma directly gives such a bound, see the following corollary. This is all that is needed for the purposes of this paper.
\end{remark}

\begin{corollary} \label{CorGapCoord}
Under the assumptions of Lemma \ref{LemGapCoord} there exists $\delta = \delta(C_g, C_e)$ and $D=D(C_g, C_e)$ such that for any curve $\gamma$ as in Lemma \ref{LemGapCoord} with $\Lg(\gamma) \leq \delta$ we have $L_{\overline{g}}(\gamma) \leq D \cdot \Lg(\gamma)$.
\end{corollary}

\begin{proof}
We choose $C_L = 1$ in Lemma \ref{LemGapCoord} to obtain 
\begin{equation}
\label{EqCorPf}
L_{\overline{g}}(\gamma) \leq B(C_g, C_e, 1) \cdot \Lg(\gamma)
\end{equation} 
for all $\gamma$ with $L_{\overline{g}}(\gamma) \leq 1$ and we set $\delta := \frac{1}{2 \cdot B(C_g, C_e, 1)}$. We claim that $D(C_g, C_e) = B(C_g, C_e, 1)$ meets the requirement. To show this let now $\gamma$ be a curve with $\Lg(\gamma) \leq \delta$ which we can assume, without loss of generality, to be parametrised by $\overline{g}$-arclength, i.e., $\gamma : [0, L_{\overline{\gamma}}) \to V$. If $L_{\overline{g}}(\gamma) \leq 1$, then the corollary follows from \eqref{EqCorPf}. Assuming now $L_{\overline{g}}(\gamma) > 1$, we consider $\gamma|_{[0, 1]}$ and apply \eqref{EqCorPf} to obtain the contradiction
$$1 = L_{\overline{g}}(\gamma|_{[0,1]}) \leq B \cdot \Lg(\gamma|_{[0,1]}) \leq B \cdot \Lg(\gamma) \leq \frac{1}{2}\;.$$
This proves the corollary.
\end{proof}

\section{Lipschitz extensions}

\subsection{Controlling the $C^0$-structure of the extension} \label{SecC0}

\begin{proposition} \label{PropBBound}
Let $(M,g)$ be a Lorentzian manifold with $g \in C^1$ and let $\iota : M \hookrightarrow \tilde{M}$ be a $C^{0,1}_{\loc}$-extension. Let $\gamma_i : [-1,0) \to M$ be two $C^1$-curves, $i = 1,2$, where $\gamma_1$ is a causal curve and $\gamma_2$ has finite gap-length. Let $e_\alpha^{(i)}$ be a parallel orthonormal frame field along $\gamma_i$, for $i = 1,2$. Moreover, let $s_n, t_n \in [-1,0)$ be two sequences with $\lim_{n \to \infty} s_n = 0 = \lim_{n \to \infty} t_n$ and $\sigma_n : [0,1] \to M$ piecewise $C^1$-curves from $\gamma_1(s_n)$ to $\gamma_2(t_n)$ such that $L_{\mathrm{gap}, e_\alpha^{(1)}}(\sigma_n) \to 0$ for $n \to \infty$ and such that, when parallely propagating $e_\alpha^{(1)}$ along $\sigma_n$, the Lorentz transformations relating $e_\alpha^{(1)}|_{\sigma_n(1)}$ and $e_\alpha^{(2)}|_{\gamma_2(t_n)}$ are uniformly bounded in $n \in \N$. 

Now, if $\lim_{s \to 0} (\iota \circ \gamma_1)(s) \in \tilde{M}$ exists, then so does $\lim_{s \to 0} (\iota \circ \gamma_2)(s) \in \tilde{M}$ exist and  we have $\lim_{s \to 0} (\iota \circ \gamma_1)(s) = \lim_{s \to 0} (\iota \circ \gamma_2)(s)$.
\end{proposition}

\begin{figure}[h] 
  \centering
  \def\svgwidth{7cm}
   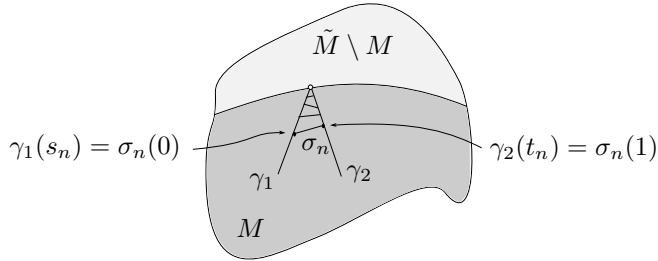 
      \caption{The set-up of Proposition \ref{PropBBound}.} \label{FigBBound}
\end{figure}

\begin{proof}
Let $\tilde{p} := \lim_{s \to 0} ( \iota \circ \gamma_1)(s)$ and let $ \tilde{\varphi} : \tilde{M} \supseteq \tilde{U} \to (-\varepsilon_0, \varepsilon_0) \times (-\varepsilon_1, \varepsilon_1)^d$ be a near-Minkowskian neighbourhood of $\tilde{p}$ with coordinates $\tilde{x}^\mu$ and uniform bounds $|\tilde{g}^{\mu \nu}|,\; |\tilde{g}_{\mu \nu}|, \; |\mathrm{Lip}(\tilde{g}_{\mu \nu})| \leq C_{\tilde{g}}$ (and $\tilde{g}_{\mu \nu}$ being close enough to the Minkowski metric $m_{\mu \nu}$ such that $\tilde{x}^0$ is in particular a time function). Note that in the open region $\tilde{U} \cap \iota(M)$, where the metric is in $C^1$, we have $|\rd_\kappa \tilde{g}_{\mu \nu}| \leq C_{\tilde{g}}$. Without loss of generality assume that $\gamma_1$ is future directed in\footnote{To be more precise, one should write $\tilde{\gamma}_1 := \iota \circ \gamma_1$, but we are slightly lax here -- and in the following -- and identify objects in $(M,g)$ with those in $(\iota(M), \tilde{g}|_{\iota(M)})$.} $\tilde{U}$ and completely contained in $\tilde{U}$. It follows from $\gamma_1$ being causal (e.g.\ by parametrising $\gamma_1$ by the $\tilde{x}^0$-coordinate) that $\gamma_1$ is of finite coordinate length in $\tilde{U} \cap \iota(M)$. We can now apply Lemma \ref{LemGapCoord} to $\tilde{U} \cap \iota(M)$ and $\gamma_1$ to obtain, from \eqref{EqLemRes1}, that 
\begin{equation}
\label{EqPfC0E}
||e^{(1)}||_{\R^{d+1}}, \; ||(e^{(1)})^{-1}||_{\R^{d+1}} \leq B \qquad \textnormal{ along }\gamma_1\;.
\end{equation}
Moreover, it follows from \eqref{EqGapLorentzT} and the assumption that $e_\alpha^{(1)}|_{\sigma_n(1)}$ and $e_\alpha^{(2)}|_{\gamma_2(t_n)}$ are related by uniformly bounded Lorentz transformations that there exists $C_L >0$ such that
\begin{equation}
\label{EqPfC0GapR}
L_{\mathrm{gap}, e_{\alpha}^{(1)}|_{\sigma_n(1)}} \big(\gamma_2|_{[t_n, 0)} \big) \leq C_L \cdot L_{\mathrm{gap}, e_\alpha^{(2)}|_{\gamma_2(t_n)}} \big(\gamma_2|_{[t_n, 0)} \big) \qquad \textnormal{ holds for all } n \in \N.
\end{equation}
We are now ready to show $\lim_{s \to 0} (\iota \circ \gamma_2)(s) = \tilde{p}$. Let $\varepsilon >0$ and let $B_{\overline{g}, \varepsilon}(\tilde{p}) \subseteq \tilde{U}$ be the coordinate ball of radius $\varepsilon$ in $\tilde{U}$. We now apply Corollary \ref{CorGapCoord} to the open subset $\tilde{U} \cap \iota(M)$ of $\tilde{U}$ to obtain a $\delta >0$ and a $D>0$ such that for all curves $\tau$ starting on $\gamma_1$ and mapping into $\tilde{U} \cap \iota(M)$, with $L_{\mathrm{gap}, e_\alpha^{(1)}}(\tau) \leq \delta$, we have
\begin{equation}
\label{EqPfC0AppCor}
L_{\overline{g}}(\tau) \leq D \cdot L_{\gap, e_\alpha^{(1)}}(\tau) \;.
\end{equation}
Note here that although $\delta$ and $D$ depend on $C_e$ in Corollary \ref{CorGapCoord}, by \eqref{EqPfC0E} $C_e$ can be chosen uniformly for $e_{\alpha}^{(1)}$ along $\gamma_1$. Let now $n_0$ be large enough such that $\gamma_1\big([s_{n_0}, 0)\big) \subseteq B_{\overline{g}, \frac{\varepsilon}{2}}(\tilde{p})$ and $n_1 \geq n_0$ large enough such that
\begin{equation*}
\begin{aligned}
&L_{\gap, e_\alpha^{(1)}}(\sigma_{n_1}) < \min \{ \frac{\varepsilon}{4D}, \frac{\delta}{2}\} \\
&L_{\gap, e_\alpha^{(2)}}(\gamma_2|_{[t_{n_1}, 0)} ) < \min \{ \frac{\varepsilon}{4 C_L \cdot D}, \frac{\delta}{2 C_L} \} \;.
\end{aligned}
\end{equation*}
It then follows from \eqref{EqPfC0GapR} that $L_{\gap, e_\alpha^{(1)}}(\gamma_2|_{[t_{n_1}, 0)} \star \sigma_{n_1}) \leq \min\{ \frac{\varepsilon}{2D} , \delta\}$.\footnote{Here, $\star$ denotes the concatenation of curves. The resulting curve is piecewise $C^1$. } Thus, by \eqref{EqPfC0AppCor} we obtain\footnote{Note that a priori $\gamma_2|_{[t_{n_1}, 0)} \star \sigma_{n_1}$ does not need to map into $\tilde{U}$, which is a prerequisite for applying \eqref{EqPfC0AppCor}. But since it starts in $\tilde{U}$ the below bound shows that it cannot leave $\tilde{U}$.}
\begin{equation} \label{EqRef}
L_{\overline{g}}(\gamma_2|_{[t_{n_1}, 0)} \star \sigma_{n_1}) \leq D \cdot L_{\gap, e_\alpha^{(1)}}(\gamma_2|_{[t_{n_1}, 0)} \star \sigma_{n_1} ) < \frac{\varepsilon}{2}
\end{equation}
and hence $\gamma_2\big([t_{n_1}, 0)\big) \subseteq B_{\overline{g}, \varepsilon}(\tilde{p})$ as desired.
\end{proof}

\begin{remark} \label{RemIfInBoundary}
\begin{enumerate}
\item Note that, in contrast to the scenario depicted in Figure \ref{FigBBound}, $\tilde{p} = \lim_{s \to 0} (\iota \circ \gamma_1)(s)$ does not have to lie in the boundary $\rd \iota(M)$ -- it could also lie in $\iota(M) \simeq M$. However, the case of it lying in the boundary is the most relevant for the remainder of the paper.
\item Let us assume in addition $M$ to be time-oriented and globally hyperbolic, and $\gamma_1$ to be future directed. Then, in the case of $\tilde{p} = \lim_{s \to 0} (\iota \circ \gamma_1)(s) \in \rd\iota(M)$ lying on the boundary,  we can choose the near-Minkowskian neighbourhood $\tilde{U}$ in the above proof to be  a future boundary chart $\tilde{\varphi} : \tilde{U} \to (-\varepsilon_0, \varepsilon_0) \times (-\varepsilon_1, \varepsilon_1)^d$ around the future boundary point $\tilde{p}$ by Propositions \ref{PropBoundaryChart} and \ref{PropLeavingLipschitz} such that $\iota \circ \gamma_1$ lies eventually below the graph of the graphing function $f$. It then follows from \eqref{EqRef} that $\iota \circ (\gamma_2|_{[t_{n_1}, 0)} \star \sigma_{n_1})$ also maps into $B_{\overline{g}, \varepsilon}(\tilde{p}) \subseteq \tilde{U}$ and thus $\iota \circ \gamma_2$ has to lie eventually in the same connected component of $\iota(M) \cap \tilde{U}$, that is, below the graph of $f$.
\item  The reader familiar with the b-boundary construction by Schmidt (\cite{Schm71}; see also \cite{Cla93}, \cite{Fried74}) will recognise that Proposition \ref{PropBBound} is inspired by this construction. 
\item Given a spacetime $(M,g)$ and an extension $\iota : M \hookrightarrow \tilde{M}$ it is often of interest to know whether two given inextendible causal curves $\gamma_1$, $\gamma_2$ in $M$ leave for $\tilde{M}$ through the same boundary point in $\rd \iota(M) \subseteq \tilde{M}$. Proposition \ref{PropBBound} provides a convenient sufficient criterion to answer that question for Lipschitz extensions. The next proposition shows that the criterion is also necessary for Lipschitz extensions. 

Let us also remark that methods for ensuring that $\gamma_1$ and $\gamma_2$ leave through the same boundary point, and which are based on causal homotopies, have already been developed in earlier work, see the proof of Theorem 4.3 in \cite{Sbie22a}. The advantage of those methods is that they even apply to $C^0$-extensions. However, the criterion given in Proposition \ref{PropBBound} is easier to implement in practise (provided the extension is Lipschitz regular).
\end{enumerate}
\end{remark}

\begin{proposition} \label{PropFindCurves}
Let $(M,g)$ be a time-oriented and globally hyperbolic Lorentzian manifold with $g \in C^1$ and let $\iota : M \hookrightarrow \tilde{M}$ be a $C^{0,1}_{\loc}$-extension. 
Let $\tilde{p} \in \rd^+\iota(M)$ be a future boundary point and let $\tilde{\varphi} : \tilde{U} \to (-\varepsilon_0, \varepsilon_0) \times (-\varepsilon_1, \varepsilon_1)^d$ be a boundary chart as in Proposition \ref{PropBoundaryChart} with coordinates $(x^0, x^1, \ldots, x^d) = (x^0, \ux)$ and Lipschitz continuous graphing function $f : (-\varepsilon_1, \varepsilon_1)^d \to (-\varepsilon_0, \varepsilon_0)$ with Lipschitz constant\footnote{I.e., we have $|f(\ux) - f(\underline{y})| \leq C_f \cdot || \ux - \underline{y}||_{\R^d}$ for all $\ux, \underline{y} \in (-\varepsilon_1, \varepsilon_1)^d$.} $C_f$. Since the extension is locally Lipschitz we can, in addition to the bound $|\tilde{g}_{\mu \nu} - m_{\mu \nu}| < \delta$ from Proposition \ref{PropBoundaryChart}, also assume $|\mathrm{Lip}(\tilde{g}_{\mu \nu})| \leq C_{\tilde{g}}$ after making the chart slightly smaller if necessary. Denote by $\tilde{U}_< := \{(x_0, \ux) \in \tilde{U} \; | \; x_0 < f(\ux)\} \subseteq \iota(M)$ the region below the graph of $f$.

Let now $\tilde{\gamma}_i : [-1,0) \to \tilde{U}_<$, $i = 1,2$, be two  future directed causal $C^1$-curves with $\lim_{s \to 0} \tilde{\gamma}_1(s) = \lim_{s \to 0} \tilde{\gamma}_2(s) \in \mathrm{graph}(f) \subseteq \tilde{U}$. Then $\tilde{\gamma}_i$, $i=1,2$, have finite gap-length and there exists a sequence of curves $\tilde{\sigma}_n : [0,1] \to \tilde{U}_<$ as in Proposition \ref{PropBBound}. 

More precisely, there exist sequences $s_n, t_n \in [-1,0)$ with $s_n, t_n \to 0$ for $n \to \infty$ and a sequence of piecewise $C^1$-curves $\tilde{\sigma}_n : [0,1] \to \tilde{U}_<$ from $\tilde{\gamma}_1(s_n)$ to $\tilde{\gamma}_2(t_n)$ such that, if $\tilde{e}_\alpha^{(i)}$ denote parallel orthonormal frame fields along $\tilde{\gamma}_i$, we have $L_{\gap, \tilde{e}_\alpha^{(1)}}(\tilde{\sigma}_n) \to 0$ for $n \to \infty$ and, when parallely propagating $\tilde{e}_\alpha^{(1)}$ along $\tilde{\sigma}_n$, the Lorentz transformations relating $\tilde{e}_\alpha^{(1)}|_{\tilde{\sigma}_n(1)}$ and $\tilde{e}_\alpha^{(2)}|_{\tilde{\gamma}_2(t_n)}$ remain uniformly bounded in $n \in \N$.
\end{proposition}

\begin{proof}
The causal curves $\tilde{\gamma}_i$ can be parametrised by the time function $x^0$ on $\tilde{U}$. It follows that they have finite coordinate length and thus the finiteness of their gap-lengths follows directly from \eqref{EqLemRes2} in Lemma \ref{LemGapCoord}. Moreover, \eqref{EqLemRes1} in Lemma \ref{LemGapCoord} also gives that the parallel and orthonormal frame fields $\tilde{e}^{(i)}_\alpha$ along $\tilde{\gamma}_i$, $i=1,2$, remain uniformly bounded
\begin{equation}
\label{EqPfPropFindCurves1}
|| \tilde{e}^{(i)}||_{\R^{d+1}}, \; || (\tilde{e}^{(i)})^{-1}||_{\R^{d+1}} \leq C \;.
\end{equation}
Let now $s_n, t_n \in [-1,0)$ be two sequences with $\lim_{n \to \infty} s_n = 0 = \lim_{n \to \infty} t_n$. We claim that since $\tilde{U}_<$ is a Lipschitz domain we can find piecewise smooth curves $\tilde{\sigma}_n : [0,1] \to \tilde{U}_<$ with $\tilde{\sigma}_n(0) = \tilde{\gamma}_1(s_n)$ and $\tilde{\sigma}_n(1) = \tilde{\gamma}_2(t_n)$ and 
\begin{equation}
\label{EqPfPropFindCurves2}
L_{\overline{g}}(\tilde{\sigma}_n) \leq C \cdot d_{\overline{g}, \tilde{U}}\big(\tilde{\gamma}_1(s_n), \tilde{\gamma}_2(t_n) \big) \;.
\end{equation}
This will be shown below but is also a well-known property of Lipschitz domains. Assuming it for the time being it follows directly from \eqref{EqPfPropFindCurves1}, \eqref{EqPfPropFindCurves2} and Lemma \ref{LemGapCoord} that $L_{\gap, \tilde{e}^{(1)}_\alpha}(\tilde{\sigma}_n) \to 0$ for $n \to \infty$. Using in addition \eqref{EqLemRes1} it follows that the Lorentz transformations relating the parallely propagated frame $\tilde{e}_\alpha^{(1)}$ along $\tilde{\sigma}_n$ and $\tilde{e}_\alpha^{(2)}$ remain uniformly bounded in $n \in \N$.

It remains to show the claim. Let $y = (y^0, \uy)$ and $z = (z^0, \uz)$ lie in $\tilde{U}_<$. If $\uy = \uz$ then we can just connect the points by the vertical (shortest) curve  which is also  contained in $\tilde{U}_<$ and satisfies \eqref{EqPfPropFindCurves2} with $C=1$. Thus, we now assume $\uy \neq \uz$ and we set $$s_0 := \frac{y^0 - z^0 + C_f ||\uy - \uz||_{\R^d}}{2 C_f ||\uy - \uz||_{\R^d}} \;.$$
We now define the piecewise smooth curve $\tilde{\sigma}$ by
$$[0,1] \ni s \overset{\tilde{\sigma}}{\mapsto} \begin{pmatrix}
\tilde{\sigma}^0(s) = \begin{cases} y^0 - s \cdot C_f ||\uy - \uz||_{\R^d} \quad &\textnormal{ for } 0 \leq s \leq s_0 \\
y^0 - s_0 \cdot C_f || \uy - \uz ||_{\R^d} + (s - s_0) \cdot C_f ||\uy - \uz||_{\R^d} \quad &\textnormal{ for } s_0 \leq s \leq 1
\end{cases} \\
\underline{\tilde{\sigma}}(s) = s \cdot \uz + (1-s) \cdot \uy
\end{pmatrix} $$
It is easily checked that $\tilde{\sigma}(0) = y$ and $\tilde{\sigma}(1) = z$. Moreover, an each smooth piece $\tilde{\sigma}$ has (negative) slope $C_f$ and thus, since $y$ and $z$ lie below the graph of $f$, so does all of $\tilde{\sigma}$. We have $$\dot{\tilde{\sigma}}(s) = \begin{pmatrix}
\pm C_f ||\uy - \uz||_{\R^d} \\ \uz - \uy
\end{pmatrix}$$
and thus $L_{\overline{g}}(\tilde{\sigma}) = ||\uy - \uz||_{\R^d} \sqrt{1 + C_f^2} \leq \sqrt{1 + C_f^2} \cdot ||y - z||_{\R^{d+1}}$, which shows the claim.
\end{proof}

Although not needed for the remainder of this paper, let us remark that if the $C^{0,1}_{\loc}$-extension is in addition time-oriented then the formulation of the above proposition can be simplified:

\begin{corollary} \label{CorFindCurves}
Let $(M,g)$ be a time-oriented and globally hyperbolic Lorentzian manifold with $g \in C^1$ and let $\iota : M \hookrightarrow \tilde{M}$ be a $C^{0,1}_{\loc}$-extension which is also time-oriented. Let $\gamma_i : [-1,0) \to M$, $i = 1,2$, be two  future directed causal $C^1$-curves with $\lim_{s \to 0} ( \iota \circ \gamma_1)(s) = \lim_{s \to 0} ( \iota \circ \gamma_2)(s) \in \tilde{M}$. Then $\gamma_i$, $i=1,2$, have finite gap-length and there exists a sequence of curves $\sigma_n : [0,1] \to M$ as in Proposition \ref{PropBBound}.
\end{corollary}

The most interesting case covered by the corollary is of course $\lim_{s \to 0} (\iota \circ \gamma_1) = \lim_{s \to 0} (\iota \circ \gamma_2) \in \rd\iota(M)$.

\begin{proof}
Without loss of generality we can assume that the isometric embedding $\iota$ preserves the time-orientation. 

We will only discuss the case $\lim_{s \to 0} (\iota \circ \gamma_1) = \lim_{s \to 0} (\iota \circ \gamma_2) \in \rd\iota(M)$. The case of $\lim_{s \to 0} \gamma_1 = \lim_{s \to 0} \gamma_2 \in M$ is treated similarly but is simpler -- we leave the details to the reader.

By Proposition \ref{PropLeavingLipschitz} we have that $\tilde{p} := \lim_{s \to 0}(\iota \circ \gamma_1)(s) $ lies in $\rd^+ \iota(M)$ and that there is a boundary chart $\tilde{\varphi} : \tilde{U} \to (-\varepsilon_0, \varepsilon) \times (-\varepsilon_1, \varepsilon_1)^d$ such that $\iota \circ \gamma_1$ lies eventually below that graph of $f$. The corollary follows from Proposition \ref{PropFindCurves} if we can show that $\tilde{\gamma}_2:=\iota \circ \gamma_2$ also lies eventually below the graph of $f$. 
To this effect we make the following

\textbf{Claim:} $J^-(\tilde{p}, \tilde{U})$ does not contain any points lying above the graph of $f$.

Let us accept the claim for the moment. Since $\tilde{\gamma}_2$ is a locally Lipschitz future directed causal curve converging to $\tilde{p}$ it must eventually lie in $J^-(\tilde{p}, \tilde{U})$. Since $\tilde{\gamma}_2$ maps into $\iota(M)$, it cannot lie on the graph of $f$. Thus the claim entails that $\tilde{\gamma}_2$ must ultimately map into the region of $\tilde{U}$  below the graph of $f$ -- which concludes the proof.

To prove the claim, let us denote by $\tilde{U}_<$ and $\tilde{U}_>$ the points in $\tilde{U}$ lying below and above the graph of $f$, respectively. Suppose now there was $\tilde{q} \in J^-(\tilde{p}, \tilde{U}) \cap \tilde{U}_>$.  Since $\tilde{U}_>$ is open and the push-up property holds for $C^{0,1}_{\loc}$-regular metrics (see \cite{ChrusGra12}) there is also $\tilde{r} \in I^-(\tilde{p}, \tilde{U}) \cap \tilde{U}_>$. Thus, we can find a smooth past directed timelike curve starting at $\tilde{p}$ and going to $\tilde{r}$ -- and by concatenating it with the smooth and past directed timelike curve with tangent $-\rd_0$, which has to intersect $\mathrm{graph} f$, we have thus constructed a piecewise smooth timelike curve starting and ending on the graph of $f$. This contradicts the achronality of $\mathrm{graph}(f)$ known from Proposition \ref{PropBoundaryChart}.
\end{proof}

Note that Corollary \ref{CorFindCurves} does not hold if one drops the assumption of $(\tilde{M}, \tilde{g})$ being time-orientable. This follows from the example given in Figure \ref{FigCounter}.
\begin{figure}[h] 
  \centering
  \begin{minipage}{0.8\textwidth}
  \centering
  \def\svgwidth{6cm}
\begingroup%
  \makeatletter%
  \providecommand\color[2][]{%
    \errmessage{(Inkscape) Color is used for the text in Inkscape, but the package 'color.sty' is not loaded}%
    \renewcommand\color[2][]{}%
  }%
  \providecommand\transparent[1]{%
    \errmessage{(Inkscape) Transparency is used (non-zero) for the text in Inkscape, but the package 'transparent.sty' is not loaded}%
    \renewcommand\transparent[1]{}%
  }%
  \providecommand\rotatebox[2]{#2}%
  \newcommand*\fsize{\dimexpr\f@size pt\relax}%
  \newcommand*\lineheight[1]{\fontsize{\fsize}{#1\fsize}\selectfont}%
  \ifx\svgwidth\undefined%
    \setlength{\unitlength}{396.85454601bp}%
    \ifx\svgscale\undefined%
      \relax%
    \else%
      \setlength{\unitlength}{\unitlength * \real{\svgscale}}%
    \fi%
  \else%
    \setlength{\unitlength}{\svgwidth}%
  \fi%
  \global\let\svgwidth\undefined%
  \global\let\svgscale\undefined%
  \makeatother%
  \begin{picture}(1,0.24893383)%
    \lineheight{1}%
    \setlength\tabcolsep{0pt}%
    \put(0,0){\includegraphics[width=\unitlength,page=1]{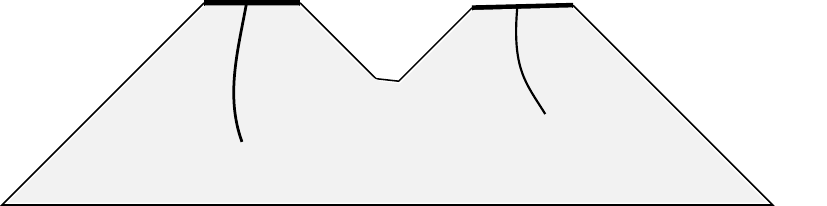}}%
    \put(0.29878906,0.09264831){\color[rgb]{0,0,0}\makebox(0,0)[lt]{\lineheight{1.25}\smash{\begin{tabular}[t]{l}$\gamma_1$\end{tabular}}}}%
    \put(0.65337574,0.1629543){\color[rgb]{0,0,0}\makebox(0,0)[lt]{\lineheight{1.25}\smash{\begin{tabular}[t]{l}$\gamma_2$\end{tabular}}}}%
    \put(0.85512327,0.12321625){\color[rgb]{0,0,0}\makebox(0,0)[lt]{\lineheight{1.25}\smash{\begin{tabular}[t]{l}$M$\end{tabular}}}}%
  \end{picture}%
\endgroup%
 
      \caption{Drawn is a subset of $1+1$-dimensional Minkowski spacetime. The extension $\tilde{M}$ is constructed by adding and identifying the bold boundaries such that $\gamma_1$ and $\gamma_2$ have the same limit point in $\tilde{M}$. } \label{FigCounter}
      \end{minipage}
\end{figure}
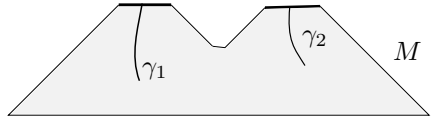

\subsection{Controlling the $C^1$-structure of the extension} \label{SecC1}

Let us agree that if $\tilde{\varphi} : \tilde{U} \to (-\varepsilon_0, \varepsilon_0) \times (-\varepsilon_1, \varepsilon_1)^d$ is a  future boundary chart as in Proposition \ref{PropBoundaryChart} then we denote by $\tilde{U}_<$ all the points of $\tilde{U}$ lying strictly below the graph of $f$ and with $\tilde{U}_\leq$ all the points of $\tilde{U}$ lying on or below the graph of $f$.

\begin{proposition} \label{PropC1Aux}
Let $(M,g)$ be a time-oriented and globally hyperbolic Lorentzian manifold with $g \in C^1$, $\tilde{\iota} : M \hookrightarrow \tilde{M}$ a $C^0$-extension, and $\hat{\iota} : M \hookrightarrow \hat{M}$ a $C^{0,1}_{\loc}$-extension. Moreover, let $\tilde{p} \in \rd^+ \tilde{\iota}(M)$ and $ \tilde{\varphi} : \tilde{U} \to (-\varepsilon_0, \varepsilon_0) \times (-\varepsilon_1, \varepsilon_1)^d$ be a future boundary chart around $\tilde{p}$ as in Proposition \ref{PropBoundaryChart}.

We assume the existence of a map $\Phi : \R^{d+1} \supseteq D \to \tilde{U}$ with the following properties:
\begin{itemize}
\item $D= \{ (z_0, \uz) \in \R^{d+1} \; | \; \uz \in B^d_\rho(0), \; 0 \leq z_0 \leq \tilde{F}(\uz)\}$, where $B^d_{\rho}(0)$ denotes the coordinate ball of radius $\rho$ around $0$ in $\R^d$ for some $\rho >0$, and $\tilde{F} : B^d_\rho(0) \to (0, \infty)$ is some continuous function. We also set $D_< = \{ (z_0, \uz) \in \R^{d+1} \; | \; \uz \in B^d_\rho(0), \; 0 \leq z_0 < \tilde{F}(\uz)\}$.
\item $\Phi|_{D_<}$ is a diffeomorphism onto its image and $\Phi$ is a homeomorphism onto its image. 
\item The family of curves $\big[0, \tilde{F}(\uz)\big) \ni z_0 \mapsto \tilde{\gamma}_{\uz}(z_0):= \Phi(z_0, \uz)  $ maps into $\tilde{U}_<$ for all $\uz \in B_\rho^d(0)$ and, under $\tilde{\iota}^{-1}|_{\tilde{\iota}(M)} : \tilde{\iota}(M) \to M$, they correspond to future directed and future inextendible causal curves in $M$. The components of the velocity vector $\partial_{z_0} \Phi^\mu(z_0, \underline{z})$ with respect to the coordinates of the future boundary chart are uniformly bounded on $D_<$. Moreover, we have $\Phi\big(\tilde{F}(0), 0\big) = \tilde{p}$.
\item There exists a continuous orthonormal frame field $\tilde{e}_{\alpha}$ on $\Phi(D) \subseteq \tilde{U}$ which is parallel along $\tilde{\gamma}_{\uz}$ for all $\uz \in B_{\rho}^d(0)$.
\end{itemize}

Now, if $\lim_{z_0 \to \tilde{F}(0)} (\id \circ \tilde{\gamma}_0)(z_0) =: \hat{p}$ exists in $\hat{M}$, where $\id := \hat{\iota} \circ \tilde{\iota}^{-1}|_{\tilde{\iota}(M)}: \tilde{M} \supseteq \tilde{\iota}(M) \to \hat{\iota}(M) \subseteq \hat{M}$ is the identification map, then there exists a future boundary chart $\hat{\varphi} : \hat{U} \to (-\varepsilon_2, \varepsilon_2) \times (-\varepsilon_3, \varepsilon_3)^d$ around $\hat{p}$ and neighbourhoods $\tilde{W} \subseteq \tilde{U}$ of $\tilde{p}$ and $\hat{W} \subseteq \hat{U}$ of $\hat{p}$ such that $$\id|_{\tilde{W}_<} : \tilde{W}_< \to \hat{W}_<$$ is a diffeomorphism and extends as a $C^1$-regular isometric diffeomorphism to $$\id|_{\tilde{W}_\leq} : \tilde{W}_\leq \to \hat{W}_\leq \;.$$
Here we have set $\tilde{W}_< := \tilde{W} \cap \tilde{U}_<$, $\tilde{W}_\leq := \tilde{W} \cap \tilde{U}_\leq$ --- and analogously in the hatted case.
\end{proposition} 

\begin{remark}
\begin{enumerate}
\item Note that it is implicit in the formulation of the proposition that $\Phi(\mathrm{graph} \tilde{F} ) \subseteq \mathrm{graph} f \subseteq \tilde{U}$ and $\Phi(D_<) \subseteq \tilde{U}_<$.
\item The orthonormal frame field $\tilde{e}_{\alpha}$ can of course be constructed by parallel propagation along  $\tilde{\gamma}_{\uz}$. The assumption is here that it extends continuously to $\mathrm{graph} f$.
\item An immediate consequence of the proposition is that one can change the $C^2$-differentiable structure on $\tilde{W}$ (in a way that is compatible with the $C^1$-differentiable structure on $\tilde{W}$) such that the regularity of $\tilde{g}$ is boosted from $C^0$ to $C^{0,1}_{\loc}$ in $\tilde{W}_\leq$.
\item The formulation of the proposition in terms of `a $C^0$-extension with an additional property' and a $C^{0,1}_{\loc}$-extension is used for future work on the $C^{0,1}_{\loc}$-inextendibility of weak null singularities without any symmetry assumptions. The special case of two $C^{0,1}_{\loc}$-extensions will be inferred from the proposition in Theorem \ref{ThmAnchUni}.
\end{enumerate}
\end{remark}

\begin{proof}
\textbf{\underline{Step 0:}} Since $(\hat{M}, \hat{g})$ is a $C^{0,1}_{\loc}$-extension it follows from Proposition \ref{PropLeavingLipschitz} that $\hat{p} \in \rd^+ \hat{\iota}(M)$ and that there exists a future boundary chart $\hat{\varphi} : \hat{U} \to (-\varepsilon_2, \varepsilon_2) \times (-\varepsilon_3, \varepsilon_3)^d$ around $\hat{p}$ such that $\id \circ \tilde{\gamma}_0$ lies eventually below the graphing function. By making the chart slightly smaller if necessary we can assume $\mathrm{Lip}(\hat{g}) \leq C$ on $\hat{U}$.

Let us denote the coordinates on $\tilde{U}$, induced by $\tilde{\varphi}$, by $x^\mu$ and those on $\hat{U}$, induced by $\hat{\varphi}$, by $y^\nu$. The push-forward of $\tilde{e}_\alpha|_{\Phi(D_<)}$ via $\id$ to $\hat{M}$ is denoted by $\hat{e}_\alpha$. We write $\tilde{e}_\alpha = \tilde{e}_\alpha^{\; \, \mu} \frac{\rd}{\rd x^\mu}$, where $\tilde{e}_\alpha^{\; \, \mu}$ is continuous on $\Phi(D)$, and $\hat{e}_\alpha = \hat{e}_\alpha^{\; \, \nu} \frac{\rd}{\rd y^\nu}$, where $\hat{e}_\alpha^{\; \, \nu}$ is continuous on its domain of definition $\hat{U}_< \cap \id\big( \Phi(D_<)\big)$. Note that this set is non-empty since $\id \circ \tilde{\gamma}_0$ maps eventually into $\hat{U}_<$.

\textbf{\underline{Step 1:}} We show that there is $s_1 \in \big[0, \tilde{F}(0) \big)$ and $0 < \rho_1 < \rho$ such that if we set $$D_1 := \{(z_0, \uz) \in D \; | \; \uz \in \overline{B_{\rho_1}^d(0)}, \; s_1 \leq z_0 < \tilde{F}(\uz)\} $$ we have firstly $\id \big( \Phi(D_1)\big) \subseteq \hat{U}_<$ and $\id \big( \Phi(D_1)\big)$ has compact closure in $\hat{U}$ --  and secondly $\frac{\rd y^\nu}{\rd x^\mu}$ $\Big( = \frac{\rd \id^\nu}{\rd x^\mu} \Big)$, which is smooth on $\Phi(D_1)$ by our definition of an extension, extends continuously to $\Phi(\overline{D_1})$ for all $\nu, \mu \in \{0, \ldots, d\}$ and, moreover,  satisfies $\det \frac{\rd y^\nu}{\rd x^\mu} \neq 0$.

\textbf{\underline{Step 1.1:}} To start the proof of what is claimed in Step 1 we observe that there is $0 \leq s_0 < \tilde{F}(0)$ such that $\id \Big( \Phi \big( (s_0, 0)\big) \Big) \in \hat{U}_<$. By the continuity of $\tilde{F}$ and that of $\id \circ \Phi|_{D_<}$ there is $0 < \rho_0 < \rho$ such that 
\begin{equation*}
\{s_0\} \times \overline{B^d_{\rho_0}(0)} \subseteq \mathrm{int}(D) \quad \textnormal{ and } \quad \id \Big( \Phi\big(\{s_0\} \times \overline{B^d_{\rho_0}(0)} \big) \Big) \subseteq \hat{U}_< \;. 
\end{equation*}
For $\uz \in \overline{B^d_{\rho_0}(0)}$ the curves $\hat{\gamma}_{\uz} := \id \circ \tilde{\gamma}_{\uz}|_{[s_0, \tilde{F}(\uz))}$ start at parameter $s_0$ in $\hat{U}$ but, a priori, they might leave $\hat{U}$ later on, cf.\ Figure \ref{FigTwoExtPf}.
\begin{figure}[h] 
  \centering
  \begin{minipage}{0.8\textwidth}
  \centering
  \def\svgwidth{8cm}
   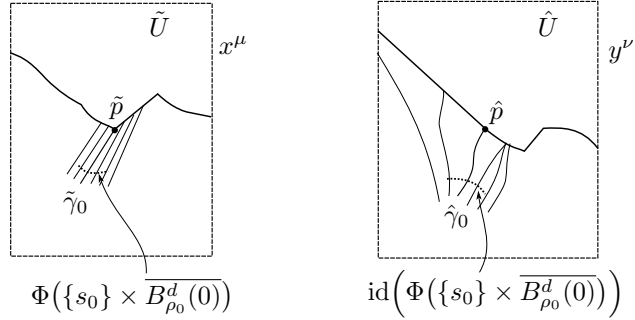 
      \caption{The set-up. Note that, a priori, the (asymptotic) behaviour of the curves $\hat{\gamma}_{\uz}$ can be wild and does not need to have the same qualitative features as those of $\tilde{\gamma}_{\uz}$.} \label{FigTwoExtPf}
      \end{minipage}
\end{figure}
Since $\hat{U}$ is open there is $\hat{F}(\uz) \in (s_0, \tilde{F}(\uz)]$ such that $\big[s_0, \hat{F}(\uz) \big)$  is the connected component of $\hat{\gamma}_{\uz}^{-1}(\hat{U})$ that contains $s_0$. Also note that since $\hat{\gamma}_{\uz}|_{[s_0, \hat{F}(\uz))}$ is causal and contained in $ \hat{U}$, its $y$-coordinate length is bounded from above uniformly in $\uz \in \overline{B^d_{\rho_0}(0)}$. Moreover, $\hat{e}_\alpha^{\; \, \nu}$ and $(\hat{e}^{-1})_\nu^{\; \, \alpha}$ are uniformly bounded in operator norm on $\id \Big( \Phi \big( \{s_0\} \times \overline{B_{\rho_0}^d(0)}\big) \Big)$ by continuity. It now follows from Lemma \ref{LemGapCoord} that
\begin{equation}
\label{EqPfPropTwoExtHat}
||\hat{e}||_{\R^{d+1}}, \; || \hat{e}^{-1} ||_{\R^{d+1}} \leq C \quad \textnormal{ on } \quad \bigcup_{\uz \in \overline{B^d_{\rho_0}(0)}} \hat{\gamma}_{\uz} \Big( \big[s_0, \hat{F}(\uz) \big) \Big) \;.
\end{equation}
Note that by definition we have
\begin{equation}
\label{EqPfPropTwoExtDef}
\id \Big( \bigcup_{\uz \in \overline{B^d_{\rho_0}(0)}} \tilde{\gamma}_{\uz} \big( [s_0, \hat{F}(\uz))\big)\Big) = \bigcup_{\uz \in \overline{B^d_{\rho_0}(0)}}  \hat{\gamma}_{\uz} \big( [s_0, \hat{F}(\uz))\big) 
\end{equation}
and that the continuity of $\tilde{e}_\alpha$ on $\Phi(D)$ implies that 
\begin{equation}
\label{EqPfPropTwoExtCont}
||\tilde{e}||_{\R^{d+1}}, \; ||\tilde{e}^{-1}||_{\R^{d+1}} \leq C \quad \textnormal{ on } \bigcup_{\uz \in \overline{B^d_{\rho_0}(0)}} \tilde{\gamma}_{\uz} \big( [s_0, \tilde{F}(\uz))\big) \;.
\end{equation}
Treating $y^\nu$ as another set of coordinates on $\tilde{\iota}(M) \cap \id^{-1}(\hat{U}_<) \subseteq \tilde{M}$ we obtain $\tilde{e}_\alpha^{\; \, \mu} \frac{\rd}{\rd x^\mu} = \hat{e}_\alpha^{\; \, \nu} \frac{\rd}{\rd y^\nu}$ and thus $\frac{\rd}{\rd x^\mu} = \big( \tilde{e}^{-1} \big)^\alpha_{\; \, \mu} \hat{e}_\alpha^{\; \, \nu} \frac{\rd}{\rd y^\nu}$. Comparison with  $\frac{\rd}{\rd x^\mu} = \frac{\rd y^\nu}{\rd x^\mu} \frac{\rd}{\rd y^\nu}$ gives
\begin{equation}
\label{EqPfPropTwoExtE}
\frac{\rd y^\nu}{\rd x^\mu} = \big( \tilde{e}^{-1} \big)^\alpha_{\; \, \mu} \hat{e}_\alpha^{\; \, \nu} \;.
\end{equation}
It now follows from \eqref{EqPfPropTwoExtHat}, \eqref{EqPfPropTwoExtDef}, \eqref{EqPfPropTwoExtCont}, and \eqref{EqPfPropTwoExtE} that there is $C_0 >0$ such that 
\begin{equation}
\label{EqDiffBounded}
||\frac{\rd y^\nu}{\rd x^\mu}||_{\R^{d+1}} \leq C_0 \quad \textnormal{ on } \quad \bigcup_{\uz \in \overline{B^d_{\rho_0}(0)}} \tilde{\gamma}_{\uz} \big( [s_0, \hat{F}(\uz))\big) \;.
\end{equation}

\textbf{\underline{Step 1.2:}} By assumption we also have 
\begin{equation}
\label{EqBoundDPhi}
||\partial_{z_0}\Phi^\mu||_{\R^{d+1}} \leq C_1 
\end{equation}
on $D_<$. Recall that $\hat{p}$ is at the coordinate origin in the $y^\nu$-coordinates for $\hat{U}$. Let $\mu >0$ be such that the coordinate ball of radius $\mu$ around the origin, $B_\mu(0) \subseteq \hat{U}$, is pre-compact in $\hat{U}$. We now choose $0 \leq s_0 \leq s_1 < \tilde{F}(0)$ such that
\begin{itemize}
\item $\hat{\gamma}_0 \big([s_1, \tilde{F}(0))\big) \subseteq B_{\nicefrac{\mu}{2}}(0)$
\item $\tilde{F}(0) - s_1 < \frac{\mu}{4 \cdot C_0 \cdot C_1}$
\end{itemize}
and $0 < \rho_1 \leq \rho_0 < \rho$ such that
\begin{itemize}
\item $0 < \tilde{F}(\uz) - s_1 < \frac{\mu}{2 \cdot C_0 \cdot C_1}$ for all $\uz \in \overline{B^d_{\rho_1}(0)}$
\item $\hat{F}(\uz) > s_1$ for all $\uz \in \overline{B^d_{\rho_1}(0)}$
\item $\id \Big( \Phi\big( \{s_1\} \times \overline{B^d_{\rho_1}(0)} \big) \Big) \subseteq B_{\nicefrac{\mu}{2}}(0) \;,$
\end{itemize}
where the first is by the continuity of $\tilde{F}$ and the last two by the continuity of $\id \circ \Phi|_{D_<}$.
\begin{figure}[h] 
  \centering
  \def\svgwidth{4.5cm}
\begingroup%
  \makeatletter%
  \providecommand\color[2][]{%
    \errmessage{(Inkscape) Color is used for the text in Inkscape, but the package 'color.sty' is not loaded}%
    \renewcommand\color[2][]{}%
  }%
  \providecommand\transparent[1]{%
    \errmessage{(Inkscape) Transparency is used (non-zero) for the text in Inkscape, but the package 'transparent.sty' is not loaded}%
    \renewcommand\transparent[1]{}%
  }%
  \providecommand\rotatebox[2]{#2}%
  \newcommand*\fsize{\dimexpr\f@size pt\relax}%
  \newcommand*\lineheight[1]{\fontsize{\fsize}{#1\fsize}\selectfont}%
  \ifx\svgwidth\undefined%
    \setlength{\unitlength}{349.10134623bp}%
    \ifx\svgscale\undefined%
      \relax%
    \else%
      \setlength{\unitlength}{\unitlength * \real{\svgscale}}%
    \fi%
  \else%
    \setlength{\unitlength}{\svgwidth}%
  \fi%
  \global\let\svgwidth\undefined%
  \global\let\svgscale\undefined%
  \makeatother%
  \begin{picture}(1,1.11368402)%
    \lineheight{1}%
    \setlength\tabcolsep{0pt}%
    \put(0,0){\includegraphics[width=\unitlength,page=1]{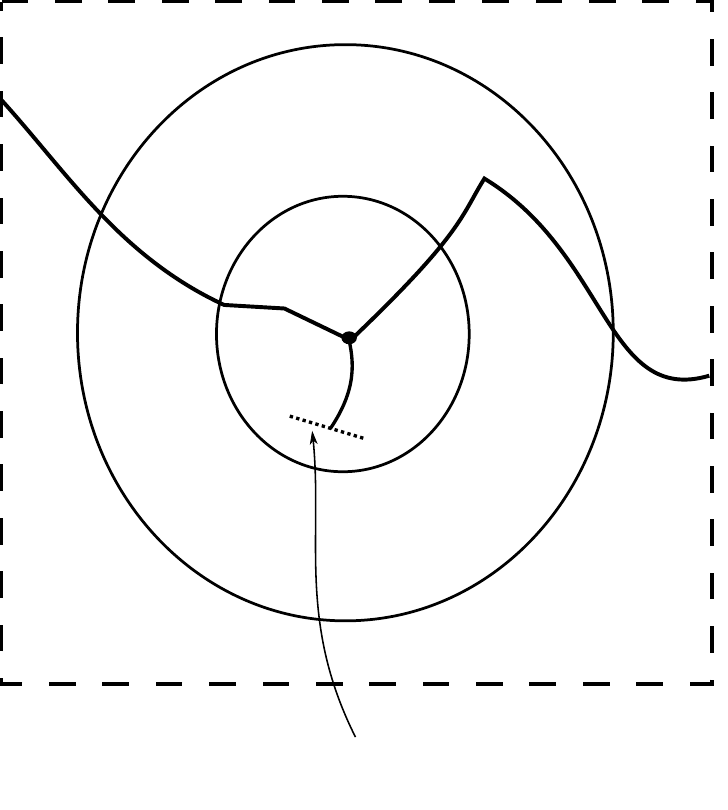}}%
    \put(0.84322326,0.97586841){\color[rgb]{0,0,0}\makebox(0,0)[lt]{\lineheight{1.25}\smash{\begin{tabular}[t]{l}$\hat{U}$\end{tabular}}}}%
    \put(0.06484228,1.00366785){\color[rgb]{0,0,0}\makebox(0,0)[lt]{\lineheight{1.25}\smash{\begin{tabular}[t]{l}$B_\mu(0)$\end{tabular}}}}%
    \put(0.29071183,0.87509589){\color[rgb]{0,0,0}\makebox(0,0)[lt]{\lineheight{1.25}\smash{\begin{tabular}[t]{l}$B_{\nicefrac{\mu}{2}}(0)$\end{tabular}}}}%
    \put(0.44708303,0.70482509){\color[rgb]{0,0,0}\makebox(0,0)[lt]{\lineheight{1.25}\smash{\begin{tabular}[t]{l}$\hat{p}$\end{tabular}}}}%
    \put(0.12391586,0.02026676){\color[rgb]{0,0,0}\makebox(0,0)[lt]{\lineheight{1.25}\smash{\begin{tabular}[t]{l}$\id \Big(\Phi\big(\{s_1\} \times \overline{B^d_{\rho_1}(0)}\big)\Big)$\end{tabular}}}}%
    \put(0.50386272,0.5646474){\color[rgb]{0,0,0}\makebox(0,0)[lt]{\lineheight{1.25}\smash{\begin{tabular}[t]{l}$\hat{\gamma}_0$\end{tabular}}}}%
  \end{picture}%
\endgroup%
 
      \caption{Controlling the length of the curves $\hat{\gamma}_{\uz}$ in $\hat{U}$.} \label{FigBalls}
\end{figure}
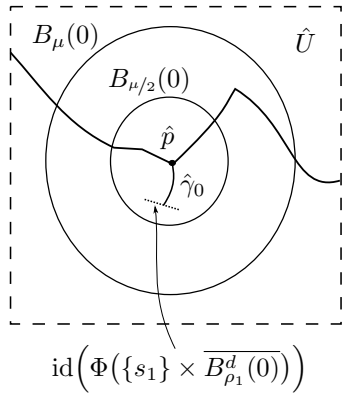

For all $s_1 \leq s < \hat{F}(\uz) \leq \tilde{F}(\uz)$ and $\uz \in \overline{B^d_{\rho_1}(0)}$ we have $$\dot{\hat{\gamma}}^\nu_{\uz}(s) = \frac{\rd y^\nu}{\rd x^\mu} \dot{\tilde{\gamma}}^\mu_{\uz}(s) = \frac{\rd y^\nu}{\rd x^\mu} \big(D\Phi(s,\uz) \big)^\mu_{\; \, 0}$$ and thus by \eqref{EqDiffBounded} and \eqref{EqBoundDPhi}
\begin{equation}
\label{EqHatVelBounded}
||\dot{\hat{\gamma}}(s)||_{\R^{d+1}} \leq C_0 \cdot C_1 \;.
\end{equation}
As a consequence we have $L_{\overline{g}} ( \hat{\gamma}_{\uz}|_{[s_1, \hat{F}(\uz))}) \leq C_0 \cdot C_1 \cdot \frac{\mu}{2 \cdot C_0 \cdot C_1} = \frac{\mu}{2}$ and thus $\hat{\gamma}_{\uz}|_{[s_1, \hat{F}(\uz))}$ must map into $B_\mu(0) \subset \subset \hat{U}$. This is however only compatible with the definition of $\hat{F}(\uz)$ if we have $\hat{F}(\uz) = \tilde{F}(\uz)$ for all $\uz \in \overline{B^d_{\rho_1}(0)}$. This shows the first part of Step 1, i.e., $\id \big( \Phi(D_1)\big) \subseteq B_\mu(0) \cap \hat{U}_<$.

\textbf{\underline{Step 1.3:}} We now show that the functions $\frac{\rd y^\nu}{\rd x^\mu}$ extend continuously to $\Phi(\overline{D}_1)$ with non-vanishing determinant.  This follows directly from \eqref{EqPfPropTwoExtE} if we can show that $\hat{e}_\alpha^{\; \, \nu} \circ \id$ extends continuously to $\Phi(\overline{D}_1)$ with non-vanishing determinant. Here, it is most convenient  to think of $\hat{e}_\alpha^{\;\, \nu}$ again as $\tilde{e}_\alpha$ in the $y^\nu$-coordinates on $\Phi(D_1) \subseteq \tilde{U}$ given by $\hat{\varphi} \circ \id$.

We first recall that $\hat{e}_\alpha$ satisfies $$0 = \frac{d}{ds} \hat{e}_\alpha^{\; \, \nu} \big( \hat{\gamma}_{\uz}(s) \big) + \hat{\Gamma}^\nu_{\kappa \rho} \big( \hat{\gamma}(s) \big) \dot{\hat{\gamma}}_{\uz}^\kappa(s) \hat{e}_\alpha^{\; \, \rho} \big( \hat{\gamma}_{\uz}(s) \big) \;. $$
Also recall \eqref{EqHatVelBounded}, $s \in [s_1, \tilde{F}(\uz))$, and $|\hat{\Gamma}^\nu_{\kappa \rho}| \lesssim 1$. Gronwall gives $||\hat{e}_\alpha^{\; \, \nu} (s)||_{\R^{d+1}} \lesssim 1$ for all $s_1 \leq s < \tilde{F}(\uz)$. It now follows from 
\begin{equation}
\label{EqIntODE}
\hat{e}_\alpha^{\; \, \nu} \big( \hat{\gamma}_{\uz}(s_3) \big) = \hat{e}_\alpha^{\; \, \nu} \big( \hat{\gamma}_{\uz}(s_2) \big) - \int\limits_{s_2}^{s_3} \hat{\Gamma}^\nu_{\kappa \rho} \big( \hat{\gamma}(s) \big) \dot{\hat{\gamma}}_{\uz}^\kappa(s) \hat{e}_\alpha^{\; \, \rho} \big( \hat{\gamma}_{\uz}(s) \big) \, ds
\end{equation}
for all $s_1 \leq s_2 < s_3 < \tilde{F}(\uz)$ that $\lim_{s \to \tilde{F}(\uz)} \hat{e}_\alpha^{\, \; \nu} \big( \hat{\gamma}_{\uz}(s) \big)$ exists. Moreover, we have $\det \hat{e}(s) = \det \hat{e}(s_1) \cdot \exp \big[- \int_{s_1}^s \hat{\Gamma}^\nu_{\kappa \nu} \big( \hat{\gamma}_{\uz}(s')\big) \dot{\hat{\gamma}}^\kappa_{\uz}(s') \, ds' \big]$. This gives us an extension with non-vanishing determinant of $\hat{e}_\alpha^{\, \; \nu} \circ \id$ to $\overline{D_1}$ using the $(z_0, \uz)$-coordinates. To check the continuity in the $(z_0, \uz)$-coordinates (which are continuously equivalent to the $x^\mu$-coordinates by the assumption that $\Phi$ is a homeomorphism on $D$) at $\overline{D_1} \cap \mathrm{graph} \tilde{F}$ we observe that \eqref{EqIntODE} together with the above mentioned bounds imply on $\overline{D_1}$
$$\Big| \hat{e}_\alpha^{\; \, \nu} \big( \tilde{F}(\uz), \uz \big) - \hat{e}_\alpha^{\; \, \nu} (z_0, \uz) \Big| \leq C \cdot |\tilde{F}(\uz) - z_0|\;,$$
where the constant $C$ is uniform on $\overline{D_1}$. Since $\hat{e}_\alpha^{\; \, \nu}$ is continuous in $D_1$ this implies also the continuity at $\overline{D_1} \cap \mathrm{graph} \tilde{F}$ in a straightforward manner. This concludes the proof of Step 1.

\textbf{\underline{Step 2:}} We choose a small cube $Q:= (-\varepsilon, \varepsilon)^{d+1}$ in the $x^\mu$-coordinates centred at $\tilde{p}$ such that $Q \cap \tilde{U}_< \subseteq \Phi(D_1)$. By Step 1, $\frac{\rd y^\nu}{\rd x^\mu}$ extends continuously to $Q \cap \tilde{U}_\leq$.  Since the graphing function $f$ is Lipschitz, the domain $Q \cap \tilde{U}_<$ is quasiconvex\footnote{Recall that a set $Q \subseteq \R^{n}$ is called \emph{quasiconvex} if there exist $\omega >0$ such that any two points $y,z \in Q$ can be joined by a curve $\alpha$ \emph{in $Q$} such that $L(\alpha) \leq \omega \cdot d_{\R^n}(y,z)$, where the length and the distance are with respect to the Euclidean metric on $\R^n$.}  (see also the proof of Proposition \ref{PropFindCurves}) and the result \cite{Whitney34} by Whitney shows that we can extend the functions $y^\nu(x^\mu)$ for each $\nu \in \{0, \ldots, d\}$ as $C^1$-functions to all of $Q$.\footnote{Since the domain is Lipschitz, one can also apply the extension procedure by Stein, see \cite{Stein70}, chapter VI, 3.2, proof of Theorem 5'.} It follows from the pre-compactness of $\id (Q \cap \tilde{U}_<) \subseteq \id\big(\Phi(D_1)\big)$ in $\hat{U}$ that after making $Q$ slightly smaller if necessary the extension of the functions $y^\nu(x^\mu)$ maps into the domain $(-\varepsilon_2, \varepsilon_2) \times (-\varepsilon_3, \varepsilon_3)^d$ of the $y$-coordinates and thus gives rise to a $C^1$-map $\mathfrak{I} : \tilde{U} \supseteq\tilde{\varphi}^{-1}(Q) \to \hat{U}$. Again by Step 1 the differential of $\mathfrak{I}$ at $\tilde{p}$ is non-vanishing and we can thus find neighbourhoods $\tilde{W} \subseteq \tilde{U}$ of $\tilde{p}$ and $\hat{W} \subseteq \hat{U}$ of $\hat{p}$ such that $\mathfrak{I} : \tilde{W} \to \hat{W}$ is a diffeomorphism. 

Without loss of generality we can choose $\hat{W}_<$ to be connected.  Since $\mathfrak{I}|_{\tilde{W}_<} = \id|_{\tilde{W}_<}$ it follows from Step 1 that $\mathfrak{I} (\tilde{W}_<) \subseteq \hat{W}_<$. We want to show that $\mathfrak{I} (\tilde{W}_<) = \hat{W}_<$. Assume there is a point $\tilde{q} \in \tilde{W} \setminus \tilde{W}_<$ with $\mathfrak{I}(\tilde{q}) \in \hat{W}_<$. Then choose a point $\tilde{r} \in \tilde{W}_<$ and connect the points $\mathfrak{I}(\tilde{q})$ and $\mathfrak{I}(\tilde{r})$ by a curve in $\hat{W}_<$. The pull-back of this curve in $\tilde{W}$ must intersect $\mathrm{graph} f$ which then gives us the contradiction that a point on the graph, which is a limit point of a future inextendible causal curve in $M$, is mapped into $\hat{W}_< \subseteq \hat{\iota}(M)$. 

Hence, $\mathfrak{I}|_{\tilde{W}_<} = \id|_{\tilde{W}_<} : \tilde{W}_< \to \hat{W}_<$ is a diffeomorphism. This concludes the proof.
\end{proof}

The next lemma shows that if the two extensions in Proposition \ref{PropC1Aux} have higher regularity then the identification map $\id$ also extends as a more regular map to the boundary.

\begin{lemma} \label{LemBoostReg}
Let $(M,g)$ be a smooth time-oriented and globally hyperbolic Lorentzian manifold and let $\iota_i : M \hookrightarrow \tilde{M}_i$ be two $C^{k_i, a_i}_{\loc}$-extensions, $i = 1,2$, with $k_i \in \N_0$ and $a_i \in \{0,1\}$.\footnote{It suffices if $(M,g)$ is as regular as the rougher of the two extensions. By $C^{k_i, 0}_{\loc}$ we mean here just $C^{k_i}$. The lemma also holds true for $a_i \in [0,1]$, but this will not be needed in this paper.}  Let $\tilde{\varphi}_i : \tilde{U}_i \to (-\varepsilon_{0,i}, \varepsilon_{0,i}) \times (-\varepsilon_{1,i}, \varepsilon_{1,i})^d$ be future boundary charts and assume there exist open neighbourhoods $\tilde{W}_i \subseteq \tilde{U}_i$ of the centres such that $$\id|_{\tilde{W}_{1,<}} : \tilde{W}_{1,<} \to \tilde{W}_{2,<}$$ is a diffeomorphism and extends as  a $C^1$-regular map to $$\id|_{\tilde{W}_{1, \leq}} : \tilde{W}_{1, \leq} \to \tilde{W}_{2,\leq } \;.$$ Then, $\id|_{\tilde{W}_{1, \leq}}$ is indeed a $C^{k+1, a}_{\loc}$-regular map, where $C^{k, a}_{\loc} = C^{k_1, a_1}_{\loc} \cap C^{k_2, a_2}_{\loc}$ is the minimal regularity of the two extensions.
\end{lemma}

Recall that if $A \subseteq \R^n$ is a closed set then we say that a function $f : A \to \R^m$ is of regularity $C^{k,a}_{\loc}$ if, and only if, it can be extended to a function $F : \R^n \to \R^m$ of the given regularity.

\begin{proof}
We resort to a well-known argument.\footnote{See for example \cite{Chrus11}, proof of Proposition 2.4.} Recall that $\id^* \tilde{g}_2 = \tilde{g}_1$ in $\tilde{W}_{1,<}$ and that $\id|_{\tilde{W}_{1,<}} = \iota_2 \circ \iota_1|_{\tilde{W}_{1,<}}^{-1}$ is smooth. In the local coordinates $x^\mu$ for $\tilde{U}_1$ and $y^\alpha$ for $\tilde{U}_2$ we thus have 
\begin{equation}
\label{EqIso}
(\tilde{g}_1)_{\mu \nu} = (\tilde{g}_2)_{\alpha \beta} \circ \id \, \frac{\partial \id^\alpha}{\partial x^\mu} \frac{\partial \id^\beta}{\partial x^\nu} \;.
\end{equation} 
Differentiating with respect to $x^\kappa$ in $\tilde{W}_{1,<}$ and using the resulting expression to build the Christoffel symbols $\big(\tilde{\Gamma}_1\big)^\lambda_{\mu \kappa}$ we obtain
\begin{equation}\label{EqBoostRegIso}
\frac{\partial^2 \id^\alpha}{\partial x^\mu \partial x^\kappa} = \big(\tilde{\Gamma}_1\big)^\lambda_{\mu \kappa} \frac{\partial \id^\alpha}{\partial x^\lambda} - \big(\tilde{\Gamma}_2\big)^\alpha_{\delta \gamma} \circ \id \cdot \frac{\partial \id^\gamma}{\partial x^\kappa} \frac{\partial \id^\delta}{\partial x^\mu} \;.
\end{equation}
Consider first the case $k \in \N_0$ and $a = 0$. Without loss of generality let $k \geq 1$. Since $\id|_{\tilde{W}_{1,<}}$ extends as a $C^1$-map to $\tilde{W}_{1,\leq}$, it follows that the right hand side of  \eqref{EqBoostRegIso} extends continuously to $\tilde{W}_{1, \leq}$. Thus, the second derivatives of $\id$ extend continuously. By further differentiation of \eqref{EqBoostRegIso} and an inductive argument we find that $k+1$-derivatives of $\id$ extend continuously to the boundary. The theorem in \cite{Whitney34} shows that $\id|_{W_{1, <}}$ can be extended as a $C^{k+1}$-regular map.

Now consider the case $a = 1$. Again, we proceed inductively over $k$. For $k = 0$ we obtain that the right hand side of \eqref{EqBoostRegIso}, and thus also the left hand side, is bounded in $\tilde{W}_{1,<}$.\footnote{The sets $\tilde{W}_{i,<}$ can, without loss of generality, be assumed to be pre-compact in $\tilde{U}_i$.} Differentiating \eqref{EqBoostRegIso} further we obtain that up to and including $k+2$-derivatives of $\id$ are bounded in $\tilde{W}_{i,<}$. The extension is now guaranteed by Theorem 5' in chapter VI of \cite{Stein70}.\footnote{There is no particular reason, apart from the convenience of citing an explicit statement available in the literature, to use Stein's extension operator in the case $a=1$ and Whitney's family of extension operators in the case $a=0$.}
\end{proof}

We now come to the main uniqueness result of this paper, which shows that anchored Lipschitz extensions are locally unique.

\begin{theorem}
\label{ThmAnchUni}
Let $(M,g)$ be a time-oriented and globally hyperbolic Lorentzian manifold with $g \in C^1$ and let $\tilde{\iota} : M \hookrightarrow \tilde{M}$ and $\hat{\iota} : M \hookrightarrow \hat{M}$ be two $C^{0,1}_{\loc}$-extensions. Let $\gamma : [-1,0) \to M$ be a future directed and future inextendible causal $C^1$-curve in $M$ such that $\lim_{s \to 0} ( \tilde{\iota} \circ \gamma)(s) =: \tilde{p}$ and $\lim_{s \to 0}(\hat{\iota} \circ \gamma)(s) =: \hat{p}$ exist in $\tilde{M}$, $\hat{M}$, respectively. 

Then there exist future boundary charts $\tilde{\varphi} : \tilde{U} \to (-\tilde{\varepsilon}_0, \tilde{\varepsilon}_0) \times (-\tilde{\varepsilon}_1, \tilde{\varepsilon}_1)^d$ around $\tilde{p}$ and $\hat{\varphi} : \hat{U} \to (-\hat{\varepsilon}_0, \hat{\varepsilon}_0) \times (-\hat{\varepsilon}_1, \hat{\varepsilon}_1)^d$ around $\hat{p}$, with $\tilde{\iota} \circ \gamma$ being ultimately contained in $\tilde{U}_<$ and $\hat{\iota} \circ \gamma$ being ultimately contained in $\hat{U}_<$, and neighbourhoods $\tilde{W} \subseteq \tilde{U}$ of $\tilde{p}$ and $\hat{W} \subseteq \hat{U}$ of $\hat{p}$ such that $$\id|_{\tilde{W}_<} : \tilde{W}_< \to \hat{W}_<$$ is a diffeomorphism and extends as a $C^{1,1}_{\loc}$-regular isometric diffeomorphism to $$\id|_{\tilde{W}_{\leq}} : \tilde{W}_\leq \to \hat{W}_\leq \;.$$
\end{theorem}

Let us remark that if one assumes additional regularity on the two extensions, then, in line with Lemma \ref{LemBoostReg}, one also obtains additional regularity on the extension $\id|_{\tilde{W}_{\leq}}$: the extension of $\id$ is one degree more regular than the minimal regularity of the spacetime extensions. However, as the example in Section \ref{SecCounter} shows, an analogue of Theorem \ref{ThmAnchUni} breaks down below Lipschitz regularity; it is not a matter of just sacrificing the same degree of regularity of the extension $\id|_{\tilde{W}_{\leq}}$.

\begin{proof}
By Proposition \ref{PropLeavingLipschitz} we have $\tilde{p} \in \rd^+ \tilde{\iota}(M)$ and $\hat{p} \in \rd^+\hat{\iota}(M)$ and the future boundary charts exist with $\tilde{\iota} \circ \gamma$ and $\hat{\iota} \circ \gamma$ lying eventually below the respective graphs. Consider the curve $\tilde{\tau}$ defined in the $x^\mu$-coordinates on $\tilde{U}$ by $$(-\tilde{\varepsilon}_0, 0) \ni s \mapsto \tilde{\tau}(s) = (s,0, \ldots, 0) \;.$$ We obviously have $\lim_{s \to 0} \tilde{\tau}(s) = \tilde{p}$. The curve $\tilde{\tau}$ corresponds to a future directed and future inextendible timelike curve $\tau := \tilde{\iota}^{-1} \circ \tilde{\tau}$ in $M$. We apply Proposition \ref{PropFindCurves} with $\tilde{\gamma}_1 = \tilde{\gamma} := \iota \circ \gamma$ and $\tilde{\gamma}_2 = \tilde{\tau}$ to obtain a sequence of curves $\tilde{\sigma}_n$ connecting $\tilde{\gamma}$ and $\tilde{\tau}$ which has the properties stated in Proposition \ref{PropFindCurves}. By Proposition \ref{PropBBound} applied to the extension $\hat{\iota} : M \hookrightarrow \hat{M}$ we infer that $\lim_{s \to 0} (\hat{\iota} \circ \tau)(s) = \hat{p}$. The second part of Remark \ref{RemIfInBoundary} shows that $\hat{\iota} \circ \tau$ must eventually lie below the graph of the graphing function.

We now define a map $\Phi : D \to \tilde{U}$ as in Proposition \ref{PropC1Aux}. Choose a $\rho >0$ such that $B_\rho^d(0) \subset \subset (-\tilde{\varepsilon}_1, \tilde{\varepsilon}_1)^d$ and let $-\tilde{\varepsilon}_0 < - \alpha_0 < \min_{\ux \in \overline{B_\rho^d(0)}} f(\ux)$, where $f$ is the graphing function of the boundary chart $\tilde{U}$. We set $$D = \{(z_0, \uz) \in \R^{d+1} \;  | \; \uz \in B_\rho^d(0), \; 0 \leq z_0 \leq f(\uz) + \alpha_0 \} $$ and define the map $\Phi$ with respect to $z$ and $x$-coordinates by $$D \ni (z_0, \uz) \overset{\Phi}{\mapsto} (z_0 - \alpha_0, \uz) = (x_0, \ux) \;.$$
The curves $z_0 \overset{\tilde{\gamma}_{\uz}}{\mapsto} \Phi(z_0, \uz)$ are clearly future directed and future inextendible timelike curves in $M$ and we can construct a parallel frame field by making a continuous choice on $\Phi\big(\{0\} \times B_\rho^d(0) \big)$ and parallely propagating it along the curves $\tilde{\gamma}_{\uz}$. It follows as in  Step 1.3 of the proof of Proposition \ref{PropC1Aux} that this frame extends continuously to $\Phi(D) \cap \mathrm{graph} f$.

Proposition \ref{PropC1Aux} now yields the existence of neighbourhoods $\tilde{W} \subseteq \tilde{U}$ of $\tilde{p}$ and $\hat{W} \subseteq \hat{U}$ of $\hat{p}$ such that $\id|_{\tilde{W}_<} : \tilde{W}_< \to \hat{W}_<$ is a diffeomorphism and extends as a $C^1$-regular isometric diffeomorphism to $\id|_{\tilde{W}_\leq} : \tilde{W}_{\leq} \to \hat{W}_\leq$. The additional regularity finally follows from Lemma \ref{LemBoostReg}.
\end{proof}

\section{Non-uniqueness of anchored spacetime extensions below Lipschitz regularity} \label{SecCounter}

We consider the $1+1$-dimensional cosmological model $(M,g)$, where $M = (0, \infty) \times \R$  with $(t,x)$ coordinates and Lorentzian metric 
\begin{equation}
\label{MetricCos}
g = - dt^2 + a(t)^2 \; dx^2 \;.
\end{equation}  
Here, $ a : (0,\infty) \to (0,\infty)$ is a smooth function with $a(t) \to 0$ for $t \to 0$ and $\int_0^1 \frac{1}{a(t')} \; dt' < \infty$.  The latter condition ensures that we have particle horizons, i.e. along the backward light cones (the future direction being given by $\partial_t$) the $x$-coordinates tend to finite limits when $t \to 0$. Clearly, $(M,g)$ is globally hyperbolic.
\vspace*{3mm}

\underline{\textbf{Constructing the first extension:}}
We first define the null coordinate $v(t,x) = \int_0^t \frac{1}{a(t')} \; dt' + x$. In $(t,v)$ coordinates the metric \eqref{MetricCos} takes the form $g = a(t)^2 \, dv^2 - a(t) [ dv \otimes dt + dt \otimes dv]$. Moreover, defining $\tau (t) = \int_0^t a(t') \; dt'$, the metric in $(\tau,v)$ coordinates reads $g  = a(t)^2 \,dv^2 - [dv \otimes d\tau + d\tau \otimes dv]$. We have thus constructed a $C^0$-extension 
\begin{align*}
\tilde{M}_1 &= \R \times \R \qquad \textnormal{ with $(\tau,v)$ coordinates} \\
\tilde{g}_1 &=  \tilde{a}(\tau)^2 \, dv^2 - [dv \otimes d\tau + d\tau \otimes dv]   \\
\iota_1 (t,x) &= (\int_0^t a(t') \; dt', \int_0^t \frac{1}{a(t')} \; dt' + x) \;,
\end{align*}
where $$\tilde{a}(\tau) = \begin{cases} a(t(\tau)) \qquad &\textnormal{ for } \tau >0 \\
a(0) = 0 &\textnormal{ for } \tau \leq 0 \end{cases} \;.$$
In the above $(\tau, v)$-coordinates the vector fields $
\partial_\tau $ and $ \partial_v + \frac{a^2(t)}{2} \partial_\tau $ are null and  the boundary hypersurface $\{\tau = 0\}$ is a null hypersurface from which the null geodesics bifurcate off.
\begin{figure}[h]
  \centering
  \def\svgwidth{6cm}
    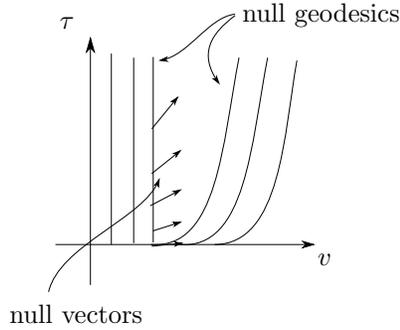
      \caption{The first $C^0$-extension}
\end{figure}
The spacetime $(\tilde{M}_1, \tilde{g}_1)$ is causally bubbling (see \cite{ChrusGra12}, \cite{GraKuSaSt19}): for $\tilde{p} \in \{\tau =0\}$ we have that  $J^+(\tilde{p}, \tilde{M}_1) \setminus I^+(\tilde{p}, \tilde{M}_1)$ has non-empty interior. As such the metric $\tilde{g}_1$ cannot be Lipschitz continuous. However, choosing $a(t) = t^p$ with $0<p<1$ (which meets the conditions assumed on $a(t)$) we find that the extension $(\tilde{M}_1, \tilde{g}_1)$ is  H\"older continuous with exponent $\alpha = \frac{2p}{p+1} \in (0,1)$. To see this note that $\tau(t) = \int_0^t s^p \, ds = \frac{t^{p+1}}{p+1}$. Hence, we have for $\tau \geq 0$
\begin{equation*}
\tilde{a}^2(\tau) = \big(t(\tau)\big)^{2p} = \big[(p+1) \tau\big]^{\frac{2p}{p+1}}\;.
\end{equation*}

We also note that the only non-vanishing partial weak derivative of $\tilde{g}_1$ in $\tau, v$ coordinates is $$\partial_\tau \big(\tilde{g}_1\big)_{vv} = \partial_\tau \big(\tilde{a}(\tau)\big)^2 = \begin{cases} \sim \partial_\tau \tau^{\frac{2p}{p+1}} \sim \tau^{\frac{p-1}{p+1}} \qquad \textnormal{ for } \tau > 0 \\
0 \qquad \textnormal{ for } \tau < 0 \end{cases}  \;.$$
We thus obtain that the weak derivatives of $\tilde{g}_1$ in the $(\tau, v)$ coordinates are locally in $L^2$ (and thus also the Christoffel symbols, since the metric is continuous) if $-1 < 2\frac{p-1}{p+1} $, which is the case for $ \frac{1}{3} < p <1$.
\vspace*{3mm}

\underline{\textbf{Constructing the second extension:}} The second extension is constructed analogously by straightening out the other set of null geodesics. We define the null coordinate $u = \int_0^t \frac{1}{a(t')} \; dt' - x$ and keep the same definition of $\tau$. In these coordinates the metric \eqref{MetricCos} takes the form $g = a(t)^2 \,du^2 - \big[ du \otimes d\tau + d\tau \otimes du\big]$. We have thus constructed another $C^0$-extension
\begin{align*}
\tilde{M}_2 &= \R \times \R \qquad \textnormal{ with $(\tau,u)$ coordinates} \\
\tilde{g}_2 &=  \tilde{a}(\tau)^2 \, du^2 - [du \otimes d\tau + d\tau \otimes du]   \\
\iota_2 (t,x) &= (\int_0^t a(t') \; dt', \int_0^t \frac{1}{a(t')} \; dt' - x) \;,
\end{align*}
where $\tilde{a}(\tau)$ is as before. Clearly, the first and the second extension are isometric and have the same regularity properties.
\vspace*{3mm}

\underline{\textbf{The two extensions are anchored but give different $C^1$-boundary extensions:}} Consider the future directed and past inextendible timelike curve $\gamma : (0,1]  \to M$ given, in $(t,x)$ coordinates, by $\gamma(s) = (s,x_0)$, where $x_0 \in \R$. In $(\tau, v)$ coordinates it is given by 
$$s \mapsto \big( \int_0^s a(t') \; dt' , \int_0^s \frac{1}{a(t')} \; dt' + x_0\big) $$ with tangent vector$$
 \frac{\partial}{\partial t} = a(t) \frac{\partial}{\partial \tau} + \frac{1}{a(t)} \frac{\partial}{\partial v} \;,$$
while in $(\tau, u)$ coordinates it is given by
$$s \mapsto \big( \int_0^s a(t') \; dt' , \int_0^s \frac{1}{a(t')} \; dt' - x_0\big) $$ with tangent vector$$
 \frac{\partial}{\partial t} = a(t) \frac{\partial}{\partial \tau} + \frac{1}{a(t)} \frac{\partial}{\partial u} \;.$$
It follows that $\gamma$ extends continuously in both extensions, and thus serves as an anchoring. Also note that if one reparametrises $\gamma$ such that the tangent vector is given by $a(t) \partial_t$, then the curve extends as a causal $C^1$ curve in both extensions. However, in the first extension $(\tilde{M}_1, \tilde{g}_1)$ it becomes tangent to the null geodesics given by $u= \mathrm{const}$ on the boundary $\{\tau =0\}$ while in the second extension it becomes tangent to the null geodesics given by $v = \mathrm{const}$ on the boundary $\{\tau = 0\}$. This shows that the two extensions differ by an infinite boost towards $\{\tau = 0\}$. Another manifestation of this infinite boost is that the future directed and past inextendible null curve $\sigma : (0,1] \to M$, given in $(\tau, v)$ coordinates by $\sigma(s) = (s,v_0)$ extends to the boundary in the first extension as a $C^1$ curve with velocity $\dot{\sigma} = \frac{\partial}{\partial \tau}\Big|_v$, while it follows from
\begin{equation*}
\begin{split}
\frac{\partial}{\partial \tau}\Big|_v &= \frac{1}{a(t)} \frac{\partial}{\partial t}\Big|_v \\
&= \frac{1}{a(t)} \Big[ \frac{\partial t}{\partial t}\Big|_v \frac{\partial}{\partial t}\Big|_u + \frac{\partial u }{\partial t}\Big|_v \frac{\partial}{\partial u} \Big|_t\Big] \\
&= \frac{1}{a(t)} \frac{\partial}{\partial t}\Big|_u + 2 \frac{1}{a^2(t)} \frac{\partial}{\partial u } \Big|_t \\
&= \frac{\partial}{\partial \tau}\Big|_u + 2 \frac{1}{a^2(t)} \frac{\partial}{\partial u} \Big|_\tau \;.
\end{split}
\end{equation*}
that it does not extend as a $C^1$ curve in the extension $(\tilde{M}_2, \tilde{g}_2)$. Hence, while $\tilde{\id}$ extends as a continuous map to the boundary, it does not extend as a $C^1$-regular map to the boundary\footnote{This obviously follows from the above discussion. But one can also write down $\id$ directly with respect to the $(\tau, v)$ and $(\tau, u)$ coordinates:
\begin{equation*}
\begin{split}
\tau & = \tau \\
u &= - v + 2 \int_0^{t(\tau)} \frac{1}{a(t')} \, dt' \;.
\end{split}
\end{equation*}
This shows that $\id$ extends continuously to $\tau = 0$ in $\tilde{M}_1$, but it follows from $\frac{\rd u }{\rd \tau}\Big|_v = \frac{\rd u}{\rd t}\Big|_v \frac{\rd t}{\rd \tau} = \frac{2}{a(t)} \cdot \frac{1}{a(t)} \to \infty$ as $\tau \to 0$ that is does not extend as a $C^1$-map.}. This shows that the analogue of Theorem \ref{ThmAnchUni} for merely H\"older-regular extensions is wrong.

\appendix

\section{A simpler method for proving uniqueness properties of sufficiently smooth extensions}

The purpose of the appendix is to demonstrate how normal coordinates can be used in a straightforward way to address questions of uniqueness of extensions which are at least $C^2$.
In Appendix \ref{SecA1} we show how local uniqueness of anchored extensions which are at least $C^2$ can be proved in a short and direct way. Although the result obtained is completely contained in Theorem \ref{ThmAnchUni} (together with Lemma \ref{LemBoostReg}), the method of proof is of interest in its own right. In Appendix \ref{SecA2} we show how it can be used to prove the statement that if one removes a single point from an inextendible spacetime (not necessarily globally hyperbolic), then the only possible extension is to reinstate the point removed. This question was raised to the author by JB Manchak in private communication \cite{JBM19}.

\subsection{Local uniqueness of anchored extensions} \label{SecA1}

For convenience we will formulate the result for smooth extensions and outline the marginal differences for limited regularity after the proof.
We begin with a corollary of Lemma \ref{LemLeavingCurveSmooth}.

\begin{corollary} \label{CorGeodSmooth}
Let $(M,g)$ be a smooth time-oriented and globally hyperbolic Lorentzian manifold and let $\tilde{\iota} : M \hookrightarrow \tilde{M}$ and $\hat{\iota} : M \hookrightarrow \hat{M}$ be two smooth extensions. Let $\gamma : [-1,0) \to M$ be a future directed and future inextendible causal $C^1$-curve in $M$ such that $\lim_{s \to 0} ( \tilde{\iota} \circ \gamma)(s) =: \tilde{p}$ and $\lim_{s \to 0}(\hat{\iota} \circ \gamma)(s) =: \hat{p}$ exist in $\tilde{M}$, $\hat{M}$, respectively.

Then there is a smooth future directed timelike geodesic $\sigma :[-1,0) \to M$ with $\lim_{s \to 0} ( \tilde{\iota} \circ \sigma)(s) =: \tilde{p}$ and $\lim_{s \to 0}(\hat{\iota} \circ \sigma)(s) =: \hat{p}$. Moreover in boundary charts $\tilde{U}$ around $\tilde{p}$ and $\hat{U}$ around $\hat{p}$ we have that $\tilde{\iota} \circ \gamma$ and $\tilde{\iota} \circ \sigma$ are ultimately contained in $\tilde{U}_<$ and $\hat{\iota} \circ \gamma$ and $\hat{\iota} \circ \sigma$ are ultimately contained in $\hat{U}_<$.\footnote{This guarantees that in going over from the curve $\gamma$ to the geodesic $\sigma$ one stays on the same side of the boundary, cf.\ Figure \ref{FigCounter}.}
\end{corollary}

Note that if a geodesic in a smooth spacetime extends continuously, then it in fact extends smoothly. Thus the implications $\tilde{p} \in \rd^+ \tilde{\iota}(M)$ and $\hat{p} \in \rd^+\hat{\iota}(M)$ are immediate.

Corollary \ref{CorGeodSmooth} follows directly from Lemma \ref{LemLeavingCurveSmooth} applied to $\tilde{\iota} : M \hookrightarrow \tilde{M}$, then applying Proposition \ref{PropFindCurves} to $\tilde{\iota} : M \hookrightarrow \tilde{M}$ and the two curves $\gamma$ and $\sigma$, and finally using Proposition \ref{PropBBound} for the other extension $\hat{\iota} : M \hookrightarrow \hat{M}$ to infer that $\hat{\iota} \circ \gamma$ and $\hat{\iota} \circ \sigma$ have the same limit points. However, since the purpose of this section is to demonstrate simpler methods in higher regularity, let us show how the proof of Lemma \ref{LemLeavingCurveSmooth} essentially gives Corollary \ref{CorGeodSmooth}:

\begin{proof}
We follow the proof of Lemma \ref{LemLeavingCurveSmooth}. We choose convex neighbourhoods $\tilde{V} \subseteq \tilde{M}$ of $\tilde{p}$ and $\hat{V} \subseteq \hat{M}$ of $\hat{p}$ such that, after making $\gamma$ shorter if necessary, we have $\tilde{\gamma}(-1) \in \tilde{V}$ and $\hat{\gamma}(-1) \in \hat{V}$, where, as usual, we have set $\tilde{\gamma} := \tilde{\iota} \circ \gamma$ and $\hat{\gamma} := \hat{\iota} \circ \gamma$. The reader can easily check that the construction that follows can be simultaneously carried out in $\tilde{V}$ and in $\hat{V}$ to yield the timelike geodesics $\tilde{\sigma}_s(t) = \exp_{\tilde{\gamma}(-1)} \big( t \exp^{-1}_{\tilde{\gamma}(-1)} \big[ \tilde{\gamma}(s)\big]\big)$ and $\hat{\sigma}_s(t) = \exp_{\hat{\gamma}(-1)} \big( t \exp^{-1}_{\hat{\gamma}(-1)} \big[ \hat{\gamma}(s)\big]\big)$, which, for $s=0$ and $t \to 1$, have the limit points $\tilde{p}$ and $\hat{p}$, respectively. It remains to show that we indeed have $\id \circ \tilde{\sigma}_0 = \hat{\sigma}_0$, where as usual $\id = \hat{\iota} \circ \tilde{\iota}^{-1}|_{\tilde{\iota}(M)}$, so that we can set $\sigma := \tilde{\iota}^{-1} \circ \tilde{\sigma}_0|_{[0,1)}$.

This follows if we can show
\begin{equation}
\label{EqLimitExp}
\id_* \big( \underbrace{\exp^{-1}_{\tilde{\gamma}(-1)} \big[ \tilde{\gamma}(0)\big]}_{\in T_{\tilde{\gamma}(-1)} \tilde{\iota}(M)}\big) = \underbrace{\exp^{-1}_{\hat{\gamma}(-1)} \big[\hat{\gamma}(0)\big]}_{\in T_{\hat{\gamma}(-1)}\hat{\iota}(M)} \;.
\end{equation}
But since we have $\id_* \big( \exp^{-1}_{\tilde{\gamma}(-1)} \big[ \tilde{\gamma}(s)\big] \big) = \exp^{-1}_{\hat{\gamma}(-1)}\big[ \hat{\gamma}(s)\big] \big)$ for all $ s \in (-1,0)$ by virtue of $\id \big( \tilde{\gamma}(s)\big) = \hat{\gamma}(s)$ and $\id$ being an isometry and thus mapping geodesics to geodesics, \eqref{EqLimitExp} follows by continuity.
\end{proof}

\begin{proposition} \label{PropSmoothExt}
Let $(M,g)$ be a smooth time-oriented and globally hyperbolic Lorentzian manifold  and let $\tilde{\iota} : M \hookrightarrow \tilde{M}$ and $\hat{\iota} : M \hookrightarrow \hat{M}$ be two smooth extensions. Let $\gamma : [-1,0) \to M$ be a future directed and future inextendible causal $C^1$-curve in $M$ such that $\lim_{s \to 0} ( \tilde{\iota} \circ \gamma)(s) =: \tilde{p}$ and $\lim_{s \to 0}(\hat{\iota} \circ \gamma)(s) =: \hat{p}$ exist in $\tilde{M}$, $\hat{M}$, respectively. 

Then there exist future boundary charts $\tilde{\varphi} : \tilde{U} \to (-\tilde{\varepsilon}_0, \tilde{\varepsilon}_0) \times (-\tilde{\varepsilon}_1, \tilde{\varepsilon}_1)^d$ around $\tilde{p}$ and $\hat{\varphi} : \hat{U} \to (-\hat{\varepsilon}_0, \hat{\varepsilon}_0) \times (-\hat{\varepsilon}_1, \hat{\varepsilon}_1)^d$ around $\hat{p}$ and neighbourhoods $\tilde{W} \subseteq \tilde{U}$ of $\tilde{p}$ and $\hat{W} \subseteq \hat{U}$ of $\hat{p}$ such that $$\id|_{\tilde{W}_<} : \tilde{W}_< \to \hat{W}_<$$ is a diffeomorphism and extends as a smooth  diffeomorphism to $$\id|_{\tilde{W}_{\leq}} : \tilde{W}_\leq \to \hat{W}_\leq \;.$$
\end{proposition}

As before we have set $\tilde{W}_< = \tilde{W} \cap \tilde{U}_< = \tilde{W} \cap \{(x_0, \ux) \in \tilde{U} \; | \; x_0 < f(\ux)\}$, where $f$ is the graphing function of the boundary chart, and similarly for $\tilde{W}_\leq$, etc.

\begin{proof}
By Corollary \ref{CorGeodSmooth} there is a future directed timelike geodesic $\sigma : [-1,0) \to M$ such that $\lim_{s \to 0} \tilde{\sigma}(s) = \tilde{p}$ and $\lim_{s \to 0} \hat{\sigma}(s) = \hat{p}$, where $\tilde{\sigma} := \tilde{\iota} \circ \sigma$ and $\hat{\sigma} := \hat{\iota} \circ \sigma$. Given boundary charts $\tilde{\varphi} : \tilde{U} \to (-\tilde{\varepsilon}_0, \tilde{\varepsilon}_0) \times (- \tilde{\varepsilon}_1, \tilde{\varepsilon}_1)^d$ around $\tilde{p}$ and $\hat{\varphi} : \hat{U} \to (-\hat{\varepsilon}_0, \hat{\varepsilon}_0) \times (- \hat{\varepsilon}_1, \hat{\varepsilon}_1)^d$ around $\hat{p}$ we choose convex neighbourhoods $\tilde{V} \subseteq \tilde{U}$ of $\tilde{p}$ and $\hat{V} \subseteq \hat{U}$ of $\hat{p}$. After making the curves slightly shorter we can without loss of generality assume that $\tilde{\gamma}$, $ \tilde{\sigma}$ lie in $\tilde{V} \cap \tilde{U}_<$ and $\hat{\gamma}$, $\hat{\sigma}$ lie in $\hat{V} \cap \hat{U}_<$. 

Let $\tilde{X} := \exp^{-1}_{\tilde{\sigma}(-1)} (\tilde{p}) \in T_{\tilde{\sigma}(-1)}\tilde{V}$ and $\hat{X} := \exp^{-1}_{\hat{\sigma}(-1)}(\hat{p}) \in T_{\hat{\sigma}(-1)}\hat{V}$. Since $\id_* \tilde{X} = \hat{X}$ and since the domains of the exponential maps are open there is a small neighbourhood $\tilde{A} \subseteq \exp^{-1}_{\tilde{\sigma}(-1)}(\tilde{V}) \subseteq T_{\tilde{\sigma}(-1)}\tilde{V}$ of $\tilde{X}$ such that $\id_*(\tilde{A}) \subseteq \exp^{-1}_{\hat{\sigma}(-1)}(\hat{V}) \subseteq T_{\hat{\sigma}(-1)}\hat{V}$.  We define the open cone $$\tilde{C} := \{ \lambda \tilde{Y} \in \exp^{-1}_{\tilde{\sigma}(-1)}(\tilde{V}) \; | \; \tilde{Y} \in \tilde{A}, \; \lambda \in (0,1] \} \;.$$
Since the domains are star-shaped we moreover have $$\hat{C} := \id_*(\tilde{C}) \subseteq \exp^{-1}_{\hat{\sigma}(-1)}(\hat{V}) \;.$$
We set $\tilde{W} := \exp_{\tilde{\sigma}(-1)}(\tilde{C}) \subseteq \tilde{V}$ and $\hat{W} := \exp_{\hat{\sigma}(-1)}(\hat{C}) \subseteq \hat{V}$, which are open neighbourhoods of $\tilde{p}$ and $\hat{p}$, respectively --- see also Figure \ref{FigBoundaryExp}. Since $\id$ is an isometry we have
\begin{equation}
\label{EqExtSm}
\id|_{\tilde{W}_<} = \exp_{\hat{\sigma}(-1)} \circ \id_* \circ \exp^{-1}_{\tilde{\sigma}(-1)}|_{\tilde{W}_<} \;.
\end{equation}
\begin{figure}[h]
\centering
 \def\svgwidth{10cm}
\begingroup%
  \makeatletter%
  \providecommand\color[2][]{%
    \errmessage{(Inkscape) Color is used for the text in Inkscape, but the package 'color.sty' is not loaded}%
    \renewcommand\color[2][]{}%
  }%
  \providecommand\transparent[1]{%
    \errmessage{(Inkscape) Transparency is used (non-zero) for the text in Inkscape, but the package 'transparent.sty' is not loaded}%
    \renewcommand\transparent[1]{}%
  }%
  \providecommand\rotatebox[2]{#2}%
  \newcommand*\fsize{\dimexpr\f@size pt\relax}%
  \newcommand*\lineheight[1]{\fontsize{\fsize}{#1\fsize}\selectfont}%
  \ifx\svgwidth\undefined%
    \setlength{\unitlength}{577.86598025bp}%
    \ifx\svgscale\undefined%
      \relax%
    \else%
      \setlength{\unitlength}{\unitlength * \real{\svgscale}}%
    \fi%
  \else%
    \setlength{\unitlength}{\svgwidth}%
  \fi%
  \global\let\svgwidth\undefined%
  \global\let\svgscale\undefined%
  \makeatother%
  \begin{picture}(1,0.33068327)%
    \lineheight{1}%
    \setlength\tabcolsep{0pt}%
    \put(0,0){\includegraphics[width=\unitlength,page=1]{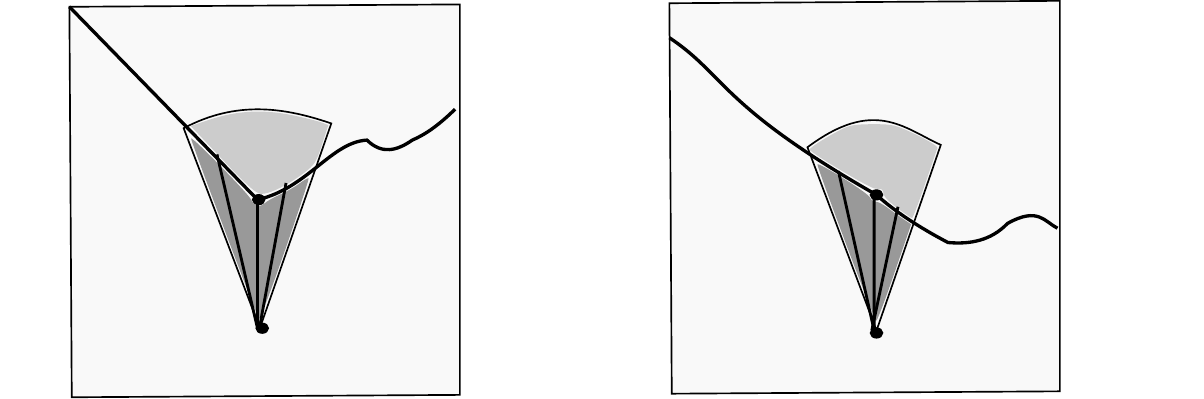}}%
    \put(-0.00367732,0.16880267){\color[rgb]{0,0,0}\makebox(0,0)[lt]{\lineheight{0}\smash{\begin{tabular}[t]{l}$\tilde{U}$\end{tabular}}}}%
    \put(0.90409618,0.17078037){\color[rgb]{0,0,0}\makebox(0,0)[lt]{\lineheight{0}\smash{\begin{tabular}[t]{l}$\hat{U}$\end{tabular}}}}%
    \put(0.23958224,0.25977776){\color[rgb]{0,0,0}\makebox(0,0)[lt]{\lineheight{0}\smash{\begin{tabular}[t]{l}$\tilde{W}$\end{tabular}}}}%
    \put(0.7379677,0.25384461){\color[rgb]{0,0,0}\makebox(0,0)[lt]{\lineheight{0}\smash{\begin{tabular}[t]{l}$\hat{W}$\end{tabular}}}}%
    \put(0.25144858,0.10353787){\color[rgb]{0,0,0}\makebox(0,0)[lt]{\lineheight{0}\smash{\begin{tabular}[t]{l}$\tilde{W}_{<}$\end{tabular}}}}%
    \put(0.75378948,0.07782753){\color[rgb]{0,0,0}\makebox(0,0)[lt]{\lineheight{0}\smash{\begin{tabular}[t]{l}$\hat{W}_{<}$\end{tabular}}}}%
    \put(0,0){\includegraphics[width=\unitlength,page=2]{LocalExp.pdf}}%
    \put(0.42153251,0.21231249){\color[rgb]{0,0,0}\makebox(0,0)[lt]{\lineheight{0}\smash{\begin{tabular}[t]{l}$\id|_{\tilde{W}_{<}}$\end{tabular}}}}%
  \end{picture}%
\endgroup%

      \caption{Constructing the local boundary identification} \label{FigBoundaryExp}
\end{figure}
And since $\exp_{\tilde{\sigma}(-1)}|_{\tilde{C}}$ and $\exp_{\hat{\sigma}(-1)}|_{\hat{C}}$ are smooth diffeomorphisms, it is immediate that $\id|_{\tilde{W}_<}$ extends to a smooth diffeomorphism $\id|_{\tilde{W}_\leq} : \tilde{W}_\leq \to \hat{W}_\leq$.
\end{proof}

Note that here in the smooth category we constructed the extension map explicitly in terms of exponential maps.

If the extensions (and possible also $(M,g)$ itself) are only $C^k$-regular with $k \geq 2$, then we first point the reader to Remark \ref{RemLemRougher} to see that Lemma \ref{LemLeavingCurveSmooth}, and by extension also Corollary \ref{CorGeodSmooth}, remain valid (and produce a timelike geodesic of regularity $C^{k+1}$). The proof of Proposition \ref{PropSmoothExt} also goes through up to \eqref{EqExtSm}. Here we note that the exponential maps are only of regularity $C^{k-1}$. Thus, since $k \geq 2$ we still obtain a $C^1$-extension -- and additional regularity, if present in the extensions, can again be gained from Lemma \ref{LemBoostReg}.

\subsection{Unique extensions of holes in inextendible spacetimes} \label{SecA2}

We now apply this method of probing sufficiently smooth extensions with geodesics to a problem of different flavour. The question, posed to the author by JB Manchak, is the following: Say we have a Lorentzian manifold which is $C^{\infty}$-inextendible. If we remove a point from this manifold, is the only possible extension the one that reinserts this point? -- or are new ways of extending the punctured Lorentzian manifold opening up?
We show the following

\begin{theorem} \label{ThmExtHoles}
Let $(M,g)$ be a connected time-oriented Lorentzian manifold with $g \in C^\infty$ that is $C^\infty$-inextendible. Let $p \in M$. Then every $C^\infty$-extension of $M \setminus \{p\}$ is smoothly isometric to $(M,g)$.
\end{theorem}


We need the following elementary lemma which is a strengthening of Lemma 2.17 in \cite{Sbie15} for sufficiently smooth extensions.

\begin{lemma} \label{LemSmoothG}
Let $(M,g)$ be a connected Lorentzian manifold with $g \in C^\infty$  and $\iota : M \hookrightarrow \tilde{M}$ be a $C^\infty$-extension of $M$. Then there exists a timelike geodesic $\gamma : [-1,0) \to M$ with $\lim_{s \to 0} (\iota \circ \gamma)(s) \in \partial \iota(M)$. 
\end{lemma}

\begin{proof}
By definition of a $C^\infty$-extension there exists a $\tilde{q} \in \partial \iota(M) \subseteq \tilde{M}$. Since $\tilde{g} \in C^\infty$ we can choose a convex neighbourhood $\tilde{U} \subseteq \tilde{M}$ of $\tilde{q}$. Consider a point $\tilde{r} \in I^-(\tilde{q}, \tilde{U})$. If $\tilde{r} \in \iota(M)$, then consider the timelike geodesic starting at $\tilde{r}$ and ending at $\tilde{q} \notin \iota(M)$. This gives rise to a timelike geodesic as in the lemma. If $\tilde{r} \notin \iota(M)$, then consider $I^+(\tilde{r}, \tilde{U})$ which is an open neighbourhood of $\tilde{q}$ and thus contains a point $\tilde{u} \in \iota(M)$. Now apply the same argument as before.
\end{proof}

Let us note that Lemma \ref{LemSmoothG} also holds if the extension is only continuous (\cite{GalLinSbi17}, \cite{MinSuhr19}). The proof, however, is more difficult.

\begin{proof}[Proof of Theorem \ref{ThmExtHoles}:]
Let $\iota : M \setminus \{p\} \hookrightarrow \tilde{M}$ be a $C^\infty$-extension of $M \setminus \{p\}$ and let $\gamma : [-1, 0) \to M \setminus \{p\}$ be a timelike geodesic with $\lim_{s \to 0} (\iota \circ \gamma)(s) = \tilde{q} \in \partial \iota(M)$ as ensured by the above lemma. Without loss of generality assume that $\gamma$ is future directed.
\vspace*{3mm} 

\textbf{Step 1:} We show that $\lim_{s \to 0} \gamma (s) = p$, where we have made the trivial identification of $\gamma$ as a curve in $M \setminus \{p\}$ with $\gamma$ as a curve in $M$.
\vspace*{3mm}

We prove this by contradiction. So assume $\lim_{s \to 0} \gamma(s) \neq p$.
\vspace*{2mm}

\textbf{Step 1.1:} We show that $p$ is then also not an accumulation point of $\gamma : [-1,0) \to M \setminus \{p\}$ considered as a curve in $M$.
\vspace*{2mm}

If it were an accumulation point then there exists a sequence $s_n \in [-1,0)$ with $s_n \to 0$ for $n \to \infty$ with $\lim_{n \to \infty} \gamma(s_n) = p$. We now choose a globally hyperbolic near-Minkowskian neighbourhood $U \subseteq M$ of $p$ and let $V \subseteq U$ be a neighbourhood of $p$ with compact closure in $U$. Since $\gamma$ is future inextendible in $M \setminus \{p\}$ and $\lim_{s \to 0} \gamma(s) \neq p$, for every $s_n$ there must exist a $0> \tilde{s}_n > s_n$ with $\gamma(\tilde{s}_n) \in U \setminus V$. Moreover, it is easy to arrange for all $\gamma(\tilde{s}_n)$ to lie in a compact set by choosing $\tilde{s}_n$ such that $\gamma(\tilde{s}_n)$ lies uniformly away from the boundary of $U$. Thus there exists a point $r \in U \setminus V$ with $\gamma(\tilde{s}_n) \to r$ for $n \to \infty$, after possibly choosing a subsequence. 
Hence, we have shown that there exists another accumulation point $r \in M\setminus \{p\}$ of $\gamma$.  Since $\tilde{\gamma} = \iota \circ \gamma$ has a future limit point $\tilde{q} \in \partial\iota(M \setminus \{p\}) \subseteq \tilde{M}$, any other accumulation point in $\tilde{M}$ must coincide with $\tilde{q}$. However, we have $\iota(r) \in \iota(M \setminus \{p\}) \subseteq \tilde{M}$, which is a contradiction. Thus we conclude that $p$ is not an accumulation point of $\gamma$.
\vspace*{2mm}

\textbf{Step 1.2:} We show that we can then construct a $C^\infty$-extension of $M$, thus obtaining a contradiction.
\vspace*{2mm}

Since $p$ is not an accumulation point of $\gamma$, we can choose a small coordinate ball $B_\varepsilon(p) \subseteq M$ centred at $p$ such that $\gamma$ does not contain any points in $B_\varepsilon(p)$. Consider now the compact annulus region $K:= \overline{B_{\frac{2}{3}\varepsilon}(p) \setminus B_{\frac{1}{3}\varepsilon}(p)}$ and set $\tilde{K} := \iota(K)$. Then $\tilde{M}_1 := \tilde{M} \setminus \tilde{K}$ is a possibly disconnected time-oriented Lorentzian manifold.\footnote{If it is connected, then the extension $\tilde{M}$ would have crept into the space that we opened up by removing $p$. Otherwise it is disconnected. This construction is to deal with the first case, since in the second case it is easy to construct an extension of $M$ just by filling in $p$ again.} In case it is disconnected we only keep the connected component containing $\iota \circ \gamma$ and call it again $\tilde{M}_1$. We now take a \emph{new} copy of $(B_\varepsilon(p), g|_{B_\varepsilon(p)})$, for example $B_\varepsilon(p) \times \{\emptyset\}$, and glue it back in along $B_\varepsilon(p) \setminus \overline{B_{\frac{2}{3}\varepsilon}(p)}$ to obtain a connected time-oriented Lorentzian manifold $\tilde{M}_2$. More precisely, we take the set-theoretic disjoint union $\tilde{M}_1 \sqcup (B_\varepsilon(p) \times \{\emptyset\})$ and consider the equivalence relation generated by $$\tilde{M}_1 \supseteq B_\varepsilon(p) \setminus \overline{B_{\frac{2}{3}\varepsilon}(p)} \ni x \sim (x , \emptyset) \in  (B_\varepsilon(p) \setminus \overline{B_{\frac{2}{3}\varepsilon}(p)}) \times \{\emptyset\} \subseteq B_\varepsilon(p) \times \{\emptyset\} \;.$$
Then topologically we have $\tilde{M}_2 := [\tilde{M}_1 \sqcup (B_\varepsilon(p) \times \{\emptyset\})]/_{ \sim}$ and it is straightforward to equip $\tilde{M}_2$ canonically with a smooth differentiable structure and smooth Lorentzian metric.\footnote{For more details on such gluing constructions we refer the interested reader for example to Section 3.3 in \cite{Sbie13a} or to Chapter 16 in \cite{Ring}.}
Clearly we can construct an isometric embedding $\iota_2 : M \hookrightarrow \tilde{M}_2$ from the old one and our canonical glueing construction. It remains to show that $\tilde{M}_2$ contains a point that does not lie in $\iota_2(M)$, i.e., that $\tilde{M}_2$ is not isometric to $M$. However, this follows easily since $\tilde{q} = \lim_{s \to 0} (\iota \circ \gamma)(s) \in \tilde{M}_2 $ and $\tilde{q} \notin \iota(M \setminus \{p\})$. And in particular the point $p$ was newly filled in. Thus $\tilde{q} \notin \iota_2(M)$. This contradicts the $C^\infty$-inextendibility of $M$. 

Thus, the claim in Step 1 is proven. It shows that any $C^\infty$-extension of $M \setminus \{p\}$ has to extend through the region that was previously close to $p$.
\vspace*{3mm}

\textbf{Step 2:} We show that $\tilde{M}$ is isometric to $M$.
\vspace*{3mm}

Let $\tilde{U} \subseteq \tilde{M}$ be a convex neighbourhood of $\tilde{q}$ and $U \subseteq M $ be a convex neighbourhood of $p$. Consider $\delta <0$ such that $\gamma(\delta) =: r \in U$ and $\iota(r) =: \tilde{r} \in \tilde{U}$. Let $w_0 \in T_rU$ be such that $p = \exp_r(w_0)$. Since $\iota$ is an isometry, it follows that $\tilde{q} = \exp_{\tilde{r}}(d\iota|_r (w_0))$. We can now choose a star-shaped (with respect to $0$) open set $W \subseteq T_rU$ containing $w_0$ such that $\exp_r(W)$ is a normal neighbourhood of $r$ and such that $\exp_{\tilde{r}}(d\iota|_r(W))$ is a normal neighbourhood of $\tilde{r}$. 

Now note that for all $w \in W \setminus \{\lambda \cdot w_0 \; | \; \lambda \geq 1\}$, the geodesics $[0,1] \ni s \mapsto \exp_r(s \cdot w)$ are contained in $M \setminus \{p\}$. Moreover, since  $\iota$ is an isometry, it commutes with the exponential map, thus obtaining
\begin{equation*}
\iota(\exp_r(w)) = \exp_{\tilde{r}}(d\iota|_r (w)) \qquad \textnormal{ for all } w \in W \setminus \{\lambda \cdot w_0 \; | \; \lambda \geq 1\} \;.
\end{equation*}
By the continuity of $\iota : M \setminus \{p\} \hookrightarrow \tilde{M}$, we thus obtain
\begin{equation} \label{Commute}
\iota(\exp_r(w)) = \exp_{\tilde{r}}(d\iota|_r (w)) \qquad \textnormal{ for all } w \in W \setminus \{ w_0\} \;.
\end{equation}
We can now define an embedding $\iota_1 : M \hookrightarrow \tilde{M}$ by
\begin{equation*}
\iota_1(q) = \begin{cases} \iota(q) \qquad &\textnormal{ for } q \in M \setminus \{p\} \\
\exp_{\tilde{r}}(d\iota|_r(w)) \qquad &\textnormal{ for } q \in \exp_r(W), q = \exp_r(w) \;. \end{cases}
\end{equation*}
By \eqref{Commute} this is well-defined and since $g$, $\tilde{g}$ are smooth, $\iota_1$ is also smooth. By continuity, it is also an isometry. Thus we have shown that $\iota_1 : M \hookrightarrow \tilde{M}$ is an isometric embedding. Since $M$ is inextendible, they have to be isometric.
This concludes the proof.


\end{proof}

Again, one can generalise the theorem to regularities $C^k$ with $k \geq 2$. We leave the details to the reader.

\bibliographystyle{acm}
\bibliography{Bibly}

\begin{thebibliography}{10}

\bibitem{BerSa07}
{\sc Bernal, A.~N., and S{\'{a}}nchez, M.}
\newblock {Globally hyperbolic spacetimes can be defined as `causal' instead of
  `strongly causal'}.
\newblock {\em Classical and Quantum Gravity 24}, 3 (jan 2007), 745--749.

\bibitem{Chrus11}
{\sc Chru\'sciel, P.}
\newblock {On maximal globally hyperbolic vacuum space-times}.
\newblock {\em J. Fixed Point Theory Appl. 14\/} (2013), 325--353.

\bibitem{ChrusGra12}
{\sc Chru\'sciel, P., and Grant, J.}
\newblock {On Lorentzian causality with continuous metrics}.
\newblock {\em Class. Quantum Grav. 29\/} (2012).

\bibitem{Chrus10}
{\sc Chrusciel, P.~T.}
\newblock {Conformal boundary extensions of Lorentzian manifolds}.
\newblock {\em Journal of Differential Geometry 84}, 1 (2010), 19--44.

\bibitem{Chrus11a}
{\sc Chru{\'s}ciel, P.~T.}
\newblock Elements of causality theory.
\newblock {\em arXiv:1110.6706\/} (2011).

\bibitem{Cla93}
{\sc Clarke, C. J.~S.}
\newblock {\em {The analysis of space-time singularities}}.
\newblock No.~1. Cambridge University Press, 1993.

\bibitem{DafLuk17}
{\sc Dafermos, M., and Luk, J.}
\newblock {The interior of dynamical vacuum black holes I: The $C^0$-stability
  of the Kerr Cauchy horizon}.
\newblock {\em arXiv:1710.01772\/} (2017).

\bibitem{Fried74}
{\sc Friedrich, H.}
\newblock {Construction and properties of space-time b-boundaries}.
\newblock {\em General relativity and gravitation 5}, 6 (1974), 681--697.

\bibitem{GalLinSbi17}
{\sc Galloway, G., Ling, E., and Sbierski, J.}
\newblock {Timelike completeness as an obstruction to $C^0$-extensions}.
\newblock {\em Comm. Math. Phys. 359}, 3 (2018), 937--949.

\bibitem{GrafLing18}
{\sc Graf, M., and Ling, E.}
\newblock {Maximizers in Lipschitz spacetimes are either timelike or null}.
\newblock {\em Class. Quantum Grav. 35}, 8 (2018).

\bibitem{GraKuSaSt19}
{\sc Grant, J., Kunzinger, M., S\"amann, C., and Steinbauer, R.}
\newblock {The future is not always open}.
\newblock {\em Lett. Math. Phys. 110\/} (2020), 83--103.

\bibitem{HawkEllis}
{\sc Hawking, S., and Ellis, G.}
\newblock {\em The large scale structure of space-time}.
\newblock Cambridge University Press, 1973.

\bibitem{JBM19}
{\sc Manchak, J.}
\newblock private communication, 2019.

\bibitem{MinSuhr19}
{\sc Minguzzi, E., and Suhr, S.}
\newblock {Some regularity results for Lorentz-Finsler spaces}.
\newblock {\em Ann. Glob. Anal. Geom. 56\/} (2019), 597--611.

\bibitem{Mis67}
{\sc Misner, C.~W.}
\newblock {Taub-NUT space as a counterexample to almost anything}.
\newblock {\em Relativity theory and astrophysics 1\/} (1967), 160.

\bibitem{Ring}
{\sc Ringstr\"om, H.}
\newblock {\em The Cauchy Problem in General Relativity}.
\newblock European Mathematical Society, 2009.

\bibitem{SaeSt18}
{\sc S{\"a}mann, C., and Steinbauer, R.}
\newblock On geodesics in low regularity.
\newblock {\em Journal of Physics: Conference Series 968\/} (feb 2018), 012010.

\bibitem{Sbie13a}
{\sc Sbierski, J.}
\newblock {On the Existence of a Maximal Cauchy Development for the Einstein
  Equations: a Dezornification}.
\newblock {\em Ann. Henri Poincar\'e 17\/} (2016), 301--329.

\bibitem{Sbie18}
{\sc Sbierski, J.}
\newblock {On the proof of the $C^0$-inextendibility of the Schwarzschild
  spacetime}.
\newblock {\em J. Phys. Conf. Ser. 968\/} (2018).

\bibitem{Sbie15}
{\sc Sbierski, J.}
\newblock {The $C^0$-inextendibility of the Schwarzschild spacetime and the
  spacelike diameter in Lorentzian geometry}.
\newblock {\em J. Diff. Geom. 108}, 2 (2018), 319--378.

\bibitem{Sbie22a}
{\sc Sbierski, J.}
\newblock {On holonomy singularities in general relativity and the
  ${C_{\mathrm{loc}}^{0,1}}$-inextendibility of space-times}.
\newblock {\em Duke Mathematical Journal 171}, 14 (2022), 2881 -- 2942.

\bibitem{Schm71}
{\sc Schmidt, B.}
\newblock A new definition of singular points in general relativity.
\newblock {\em General relativity and gravitation 1}, 3 (1971), 269--280.

\bibitem{Schm73}
{\sc Schmidt, B.}
\newblock {The local $ b $-completeness of space-times}.
\newblock {\em Communications in Mathematical Physics 29}, 1 (1973), 49--54.

\bibitem{Stein70}
{\sc Stein, E.~M.}
\newblock {\em {Singular integrals and differentiability properties of
  functions}}, vol.~2.
\newblock Princeton University Press, 1970.

\bibitem{Whitney34}
{\sc Whitney, H.}
\newblock Functions differentiable on the boundaries of regions.
\newblock {\em Annals of Mathematics 35}, 3 (1934), 482--485.

\end{thebibliography}

\end{document}